\documentclass{article}
\usepackage[11pt]{extsizes}
\usepackage[utf8]{inputenc}
\usepackage{ulem}
\usepackage[english]{babel}
\usepackage[T1]{fontenc}
\usepackage{listings}
\usepackage{amsmath}
\usepackage{mathrsfs}
\usepackage{hyperref}
\usepackage{ifthen}
\usepackage{blindtext}
\usepackage{mathtools}
\usepackage{amsmath}
\usepackage{amssymb}
\usepackage{url}
\usepackage{textcomp}
\usepackage{dsfont}
\usepackage{hyperref}
\usepackage{bbm}
\usepackage{verbatim}
\usepackage{glossaries}
\usepackage{appendix}

\usepackage[utf8]{inputenc}
\usepackage[T1]{fontenc}
\usepackage{array,multirow,makecell}
\setcellgapes{1pt}
\makegapedcells
\usepackage[table]{xcolor}
\newcolumntype{R}[1]{>{\raggedleft\arraybackslash }b{#1}}
\newcolumntype{L}[1]{>{\raggedright\arraybackslash }b{#1}}
\newcolumntype{C}[1]{>{\centering\arraybackslash }b{#1}}

\usepackage{array}
\usepackage{booktabs}
\usepackage{tabularx}
\usepackage{color}
\definecolor{mygreen}{RGB}{28,172,0} 
\definecolor{mylilas}{RGB}{170,55,241}
\newcommand{\R}{\mathbb{R}}

\newcommand{\N}{\mathbb{N}}
\newcommand{\E}{\mathbb{E}}

\newcommand{\Prob}{\mathbb{P}}

\usepackage{mathtools}

\DeclarePairedDelimiterX{\expectarg}[1]{(}{)}{%
  \ifnum\currentgrouptype=16 \else\begingroup\fi
  \activatebar#1
  \ifnum\currentgrouptype=16 \else\endgroup\fi
}

\newcommand{\innermid}{\nonscript\;\delimsize\vert\nonscript\;}
\newcommand{\activatebar}{%
  \begingroup\lccode`\~=`\|
  \lowercase{\endgroup\let~}\innermid 
  \mathcode`|=\string"8000
}

\newtheorem{theo}{Theorem}

\newtheorem{Corollary}{Corollary}
\newtheorem{prop}{Proposition}
\newtheorem{lem}{Lemma}
\newtheorem{remark}{Remark}
\newtheorem{proof}{Proof}
\newtheorem{mydeff}{Definition}
\usepackage{graphicx}
\usepackage[utf8]{inputenc}
\usepackage{mathenv}
\usepackage{gensymb}
\usepackage{subcaption}
\usepackage{tikz}
\usetikzlibrary{shapes}
\usetikzlibrary{trees}
\usetikzlibrary{calc}
\newcommand\tab[1][0.5cm]{\hspace*{#1}}
\usepackage[top=2cm, bottom=2cm, left=2cm, right=2cm]{geometry}
\usepackage{relsize}
\usepackage{pgfplots}
\usepackage{multirow}
\usepackage{lipsum}
\usepackage{bigints}            
  {

  }
\usepackage{hyperref}
\usepackage[section]{placeins}
\usepackage{caption}
\usepackage{float}
\usepackage[justification=centering]{caption}
\usepackage{bm}
\usepackage[linesnumbered,ruled,vlined]{algorithm2e}

\usepackage{mathtools, stmaryrd}
\usepackage{xparse} \DeclarePairedDelimiterX{\Iintv}[1]{\llbracket}{\rrbracket}{\iintvargs{#1}}
\NewDocumentCommand{\iintvargs}{>{\SplitArgument{1}{,}}m}
{\iintvargsaux#1} %
\NewDocumentCommand{\iintvargsaux}{mm} {#1\mkern1.5mu..\mkern1.5mu#2}

\hypersetup{
   colorlinks=true, 
   breaklinks=true, 
   urlcolor= blue, 
   linkcolor= blue,
   pdfcreator  = {\LaTeX},
   pdfproducer = {Kile}
}

\usepackage{hyperref}

\usepackage{anysize}
\marginsize{2cm}{2cm}{2cm}{2cm}
\usepackage{anyfontsize}

\usepackage{appendix}

\usepackage{listings} 
\definecolor{dkgreen}{rgb}{0,0.6,0} 
\definecolor{gray}{rgb}{0.5,0.5,0.5} 
\usepackage{fancyhdr} 
\usepackage{amssymb}
\usepackage{enumitem}   

\usepackage{hyperref}

\newcommand{\indep}{\perp \!\!\! \perp}

\usepackage{authblk}
\title{The Log S-fBM model: Statistical analysis}
\author[1]{Othmane Zarhali}
\author[1]{Emmanuel Bacry}
\author[2]{Jean-François Muzy}
\affil[1]{Ceremade, CNRS-UMR 7534, Université Paris-Dauphine PSL, Place du Maréchal de Lattre de Tassigny, 75016 Paris, France}
\affil[2]{SPE CNRS-UMR 6134, Université de Corse BP 52, 20250 Corte, France}

\date{}
\hypersetup{
    colorlinks=False,
    linkcolor=black,
    citecolor=black,
    urlcolor=blue
}

\newcounter{subsubsubsection}[subsubsection]

\makeatletter
\renewcommand\paragraph{\@startsection{paragraph}{4}{\z@}%
  {3.25ex \@plus1ex \@minus.2ex}%
  {1em}%
  {\normalfont\normalsize\bfseries}}
\renewcommand\subparagraph{\@startsection{subparagraph}{5}{\parindent}%
  {3.25ex \@plus1ex \@minus .2ex}%
  {-1em}%
  {\normalfont\normalsize\bfseries}}
\makeatother

\begin{document}

\maketitle
\begin{abstract}
The Log S-fBM model is a stochastic volatility model introduced by Wu\textit{ et al.} in \cite{wu2022rough}. It is characterized by a stochastic log volatility being a stationary fractional brownian motion (S-fBM) process: a stationary gaussian process whose autocovariance function has a power decay driven by the Hurst exponent $H$ and variance having a multiplicative coefficient namely the intermittency coefficient. One of the main properties of the Log S-fBM model is that it conciliates the rough volatility setting where the Hurst exponent is typically near $0.1$ (see \cite{gatheral2018volatility}) and the multifractal volatility setting where the Hurst exponent is of order $0$ as introduced in \cite{bacry2001multifractal,bacry2001modelling} as follows: the volatility measure of the Log S-fBM model tends to the one of multifractal volatility when the Hurst exponent goes to $0$. According to the numerical findings in \cite{wu2022rough}, this intermittency was observed to be of order $0.02$ across a large range of financial assets which justifies the meaningfulness of considering a small intermittency approximation of log volatility moments as a model calibration method (so called general method of moment known as GMM). In this work, we perform a statistical analysis of the Log S-fBM model. First, we derive scaling properties related to the S-fBM process as well as the Log S-fBM integrated volatility measure. Then, we present deviation inequalities of the Log S-fBM process with a precise description of its tail distribution emphasizing on its sensitivity with respect to the Hurst exponent and the intermittency coefficient. Moreover, we develop a statistical hypothesis testing in order to test the null hypothesis of the Hurst exponent, in other words, test the rough VS multifractal hypothesis.  Last but not least, we revisit the scale invariance properties of the log volatility increment process with explicit formulas leveraging the small intermittency approximation, enabling to reproduce the analoguous properties to the rough as well as the multifractal volatility settings.
\end{abstract}


\section*{Introduction}

Volatility is a critical component in the field of asset pricing, influencing the valuation of financial instruments, the assessment of risk and the formulation of investment strategies. In other words, the accurate modeling of volatility is crucial for several reasons. It directly affects the estimation of risk premia, influences the pricing of derivatives and provides insights into the market’s response to macroeconomic events. Moreover, understanding the dynamics of volatility is essential for effective risk management and the development of trading strategies that can adapt to changing market conditions. Traditional asset pricing models, such as the Capital Asset Pricing Model also known as the CAPM introduced in \cite{sharpe1964capital} and the Black-Scholes option pricing model (see the seminal work of \cite{black1973pricing}), often assume that volatility is either constant or deterministic. However, real-world data suggest that volatility is inherently stochastic and exhibits time-varying characteristics, such as clustering, mean reversion and asymmetry (for more details see \cite{engle1982autoregressive}). These observations have led to the development of advanced volatility models that better capture the complexity of financial markets.\\
The introduction of stochastic volatility models, such as the Heston model in \cite{heston1993closed}, marked a significant advancement in pricing theory. These models incorporate an additional stochastic process defined through its own dynamic to describe volatility, allowing more accurate pricing of derivatives and a better understanding of market dynamics. Furthermore, they address key empirical phenomena like the volatility smile observed in options markets, which cannot be explained by constant volatility models \cite{rubinstein1994implied}. In fact, many stochastic volatility models, such as the Heston model, allow to derive a smile dynamic controlled via a finite number of parameters together with a clear interpretation of every parameter and its impact on the smile.\\\\
On the other hand, the Multifractal Random Walk (MRW) is a model that has gained prominence in the modeling of financial markets and turbulent systems. Introduced by Bacry, Muzy and Delour \cite{bacry2001multifractal} (see also \cite{bacry2001multifractal,muzy2000modelling,muzy2002multifractal}), the MRW is designed to capture the complex scaling properties and long-range dependence observed in empirical data. Traditional models like geometric Brownian motion, used in the Black-Scholes framework, assume homogeneity in the underlying stochastic processes. However, empirical studies have demonstrated that financial time series exhibit multifractality, meaning their statistical properties vary across different time scales \cite{mandelbrot1997multifractal}.\\
The main advantage of the MRW is its ability to reproduce the heavy tails, volatility persistence and scaling laws commonly observed in financial returns. These characteristics are important for accurately modeling risk and pricing derivatives, particularly in markets that deviate from the assumptions of constant volatility and Gaussian returns. From a mathematical perspective, the MRW builds on the theory of multifractals and cascades, which were originally developed in the context of turbulence and chaos theory \cite{frisch1995turbulence}. In the MRW, the volatility process is itself stochastic, driven by a Gaussian process with long memory, resulting in a non-linear, multifractal structure for the asset returns.\\\\
In recent years, the concept of rough volatility, as proposed by Gatheral et al. \cite{gatheral2018volatility}, has gained traction. Rough volatility models represent volatility as a process with rough paths, characterized by fractional Brownian motion with a Hurst parameter less than 0.5. These models have been shown to provide a superior fit to high-frequency data, offering new insights into the microstructure of financial markets and improving the accuracy of derivative pricing.\\\\
At this stage, one may wonder about the relationship between multifractal models and rough volatility models. In fact, Wu \textit{et.al} in \cite{wu2022rough} introduced the Log S-fBM model: a stochastic volatility model whose log volatility is a gaussian process, with the following autocovariance function:

\begin{equation}
C_{\omega_{H,T}}(\tau)
= \operatorname{cov}\bigl(\omega_{H,T}(t), \omega_{H,T}(t+\tau)\bigr)
= \frac{\nu^2}{2}\left(1-\left(\frac{|\tau|}{T}\right)^{2H}\right)\mathds{1}_{\{|\tau|<T\}}.
\label{eq:covsfbm}
\end{equation}
Where $H$ and $T$ are respectively the Hurst exponent and the decorrelation limit. Its variance is parameterized through the parameter $\nu$. It has been proved in \cite{wu2022rough} that the volatility measure also known as the multifractal random measure of the Log S-fBM model converges to multifractal random measure of the MRW model when $H$ is closed to $0$. Consequently, the Log S-fBM model adopts a rough volatility behaviour when $H$ is different from $0$ and  a multifractal behaviour when $H$ is closed to $0$.\\
The calibration of the Log S-fBM model is done in the small intermittency regime, where the squared intermittency parameter is around $0.02$ which is inline with the one measure from observed market data according to the evidence in figure 11 in \cite{wu2022rough}. Leveraging the general method of moments in \cite{bolko2020roughness}, one can derive a small intermittency approximation of the autocovariance log volatility, then, the model parameters are determined by minimizing its discrepancy with the one measure from the observed data.\\\\
 
 This paper aims to provide a comprehensive study of the Log S-fBM model, building upon the framework introduced in \cite{wu2022rough}. We begin by presenting the definition and main structural properties of the S-fBM process, and we establish its fundamental scaling and self-similarity features, which play a central role in understanding its multiscale behaviour. Exploiting the Gaussian structure of the S-fBM marginals, we then derive quantitative large deviation inequalities for both the S-fBM process itself and the associated Log S-fBM random measure, obtaining explicit upper bounds that characterize the concentration of logarithmic fluctuations. Then, we propose a statistical test based on Log S-fBM random measure observations in the small intermittency regime to assess the null Hurst exponent hypothesis as a way to disentangle the rough and multifractal volatility regimes. In a subsequent step, we leverage the small intermittency regime in order to derive closed-form scale invariance relations for the moments, thereby clarifying the connection between the model and its Gaussian approximation. Such scale invariance properties have previously been studied in the context of rough volatility models and multifractal random measures; here, we investigate that these properties are well reproduced in the Log S-fBM framework and show how analogous scaling structures naturally emerge in this setting. Finally, numerical experiments based on synthetically generated data are conducted to illustrate and validate the established closed form scaling formulas.

\section{The Log S-fBM model}

This section introduces the Log Stationary fractional Brownian motion (Log S-fBM) model proposed by Wu \textit{et al.}~\cite{wu2022rough}. Its construction provides a unified framework linking rough and multifractal volatility models. The presentation proceeds by first defining the Stationary fractional Brownian motion (S-fBM) relying on a geometric representation of the so called time–scale. We present its autocorrelation structure parametrized by the triplet $(H,\lambda^2,T)$ respectively the Hurst exponent, intermittency coefficient and the correlation scale. Then we introduce the associated Log S-FBM random measure and its connection with multifractal random measures as its weak limit as $H\rightarrow 0$  and finally embedding this construction to the Log S-fBM process.\\

The model is built upon a time–scale domain encoding how fluctuations at different temporal scales contribute to the value of the process at a given time. For any $(\ell,T)\in\mathbb{R}_+^2$ and reference time $t^\ast\in\mathbb{R}$, define:
\begin{equation}
 \forall T\in \R_{+},\forall y\in \R,\tab C_{\ell,T}(y)
:=\left\{(t,h)\;:\;h>\ell,\; |t-y|<\tfrac{1}{2}\min(h,T)\right\}.
\end{equation}
\begin{figure}[H]
    \centering
    \includegraphics[width=90mm,height=90mm]{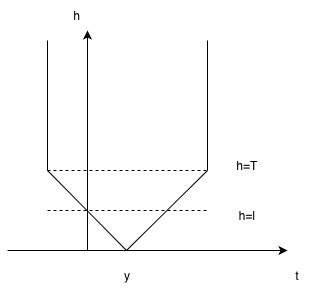} 
    \caption{Representation of $C_{\ell,T}(.)$.}
    \label{fig:timescaledomain}
\end{figure}
This set corresponds geometrically to a truncated cone in the $(t,h)$-plane, where $t$ represents time localization and $h$ the observation scale. The parameter $T$ acts as a maximal interaction scale, while $\ell$ introduces a small-scale cutoff. In the limiting case $\ell=0$, one recovers the cone $C_T(.):=C_{0,T}(.)$:
\begin{eqnarray}
        \forall T\in \R_{+},\forall t^{*}\in \R,\tab 
        C_{T}(t):=\left\{(t,h)/h>0 ,|t-t^{*}|<\frac{1}{2}\min(h,T)\right\}
\end{eqnarray}
which naturally appears in the construction of log-correlated Gaussian fields and multifractal cascades
\cite{KozhemyakBacryMuzy2008,bacry2001multifractal}.

Based on this geometric structure, the Stationary fractional Brownian motion (S-fBM) is defined as an integrated Gaussian noise over the time-scale domain $C_T(.)$:
\begin{equation}
\omega_{H,T}(t)
= \mu_H + \int_{C_T(t)} \mathrm{d}G_H,
\end{equation}
where $\mathrm{d}G_H$ denotes a non-homogeneous Gaussian white noise whose variance measure is given by
\begin{equation}
\mathbb{E}\!\left[\mathrm{d}G_H(t,h)^2\right]
= \lambda^2 \frac{h^{2H-2}}{T^{2H}}\,\mathrm{d}t\,\mathrm{d}h.
\end{equation}
Here $H\in\left]0,\tfrac12\right[$ is a Hurst exponent governing the local roughness of the process and $\lambda>0$ is an intermittency parameter controlling fluctuation intensity. The process is stationary since the covariance depends only on the overlap of cones translated in time. The normalization constant $\mu_H$ is fixed by the exponential normalization condition:
\[
\mathbb{E}\!\left[e^{\omega_{H,T}(t)}\right]=1,
\]
which yields
\begin{equation}
\mu_H
= -\frac{\lambda^2}{4H(1-2H)}
= -\frac{\nu^2}{4},
\qquad
\nu^2:=\frac{\lambda^2}{H(1-2H)}.
\end{equation}
The S-fBM process is stationary and Gaussian with autocovariance function admiting the explicit expression as demonstrated in Appendix A.1 in \cite{wu2022rough}:

\begin{equation}
C_{\omega_{H,T}}(\tau)
= \operatorname{cov}\bigl(\omega_{H,T}(t), \omega_{H,T}(t+\tau)\bigr)
= \frac{\nu^2}{2}\left(1-\left(\frac{|\tau|}{T}\right)^{2H}\right)\mathds{1}_{\{|\tau|<T\}}.
\label{eq:covsfbm}
\end{equation}
showing that the parameter $T$ induces a finite decorrelation horizon while the exponent $H$ determines the short-time rough behavior.

Exponentiating the Gaussian field naturally leads to the Log S-fBM random measure for $H\in\left]0,\tfrac12\right[$ defined as follows:
\begin{equation}
M_{H,T}(\mathrm{d}t)
:=\exp\!\big(\omega_{H,T}(t)\big)\,\mathrm{d}t,
\end{equation}
which belongs to the class of Gaussian multiplicative chaos measures. The normalization ensures that the measure preserves the mean time scale, namely $\mathbb{E}[M_{H,T}(\mathrm{d}t)]=\mathrm{d}t$. A fundamental property established in \cite{wu2022rough} is that this construction provides a continuous interpolation between rough models and multifractal models. In the level of the S-fBM process,  as $H$ decreases the covariance becomes increasingly singular near the origin, leading to an exploding variance. Whereas in the level of the Log S-fBM random measure, the latter admits a weak limit as $H$ goes to $0$. More precisely, consider a log-correlated Gaussian field $\tilde{\omega}_{\ell,T}$ defined as:
\begin{eqnarray}
    \tilde{\omega}_{\ell,T}(t)
    =
    \mu_{\ell,T}
    + \int_{C_{\ell,T}(t)} dG_{0}(t, h).
    \label{eq:omegalT}
\end{eqnarray}
where $dG_0 (t, h)$ is a gaussian white noise of variance:
\begin{eqnarray}
    dp_{0}(t,h) = \lambda^{2} h^{-2}\, dh\, dt
    \label{eq:dp0}
\end{eqnarray}

Furthermore, satisfies the normalisation condition $\mathbb{E}(e^{\tilde{\omega}_{\ell,T}(t)})=1$ and covariance of the form:
\begin{equation}
C_{\tilde{\omega}_{\ell,T}}(\tau)
= -\lambda^2\!\left(\ln(\ell/T)-1+\tau/\ell\right)\mathds{1}_{\{|\tau|<\ell\}}
\;-\;\lambda^2\ln(\tau/T)\,\mathds{1}_{\{\ell<|\tau|<T\}}.
\label{eq:logmrm}
\end{equation}

Defining $\widetilde{M}_{\ell,T}(\mathrm{d}t)
=\exp(\tilde{\omega}_{\ell,T}(t))\,\mathrm{d}t$, one obtains the multifractal random measure (MRM) as the weak limit:
\begin{eqnarray}
    \widetilde{M}_{\ell,T}(\mathrm{d}t)
\xrightarrow[\ell\to0]{\mathrm{w}}
\widetilde{M}_T(\mathrm{d}t).
\end{eqnarray}

Wu \textit{et al.} prove that the Log S-fBM random measure satisfies the convergence:
\[
M_{H,T}(\mathrm{d}t)
\xrightarrow[H\to0]{\mathrm{w}}
\widetilde{M}_T(\mathrm{d}t),
\]
where both convergences are considered in the topology of weak convergence of random measures.\\

Consequently, this establishes a bridge between the Log S-fBM random measure construction and the multifractal random walk framework of Bacry, Muzy and Delour
\cite{bacry2001multifractal,bacry2001modelling,muzy2000modelling,muzy2002multifractal}. The vanishing Hurst exponent therefore corresponds to a transition from rough Gaussian dynamics to logarithmically correlated multifractal behavior. Consequently, the Log S-fBM random measure is well defined for $H$ lying in $\left[0,\tfrac12\right[$.\\

The Log S-fBM  process is obtained from the mLog S-fBM random measure as follows:
\begin{eqnarray}
\label{eq:LogSfbmmodel}
\mathrm{d}X_t
= \sqrt{\frac{M_{H,T}(dt)}{dt}}\,\mathrm{d}B_t
\end{eqnarray}
where $(B_t)_t$ is a standard Brownian motion independent of $M_{H,T}$. Equivalently, it can be defined as a subordinated process of the form:
\begin{eqnarray}
X_t = B_{M_{H,T}(t)}.
\end{eqnarray}
This representation highlights the interpretation of the random measure as a stochastic clock governing the dynamic of $X$.

In the rough volatility regime ($H\neq 0$), $X$ takes the local martingale form:
\begin{eqnarray}
\label{eq:LogSfbmmodellocalmart}
\mathrm{d}X_t
= e^{\frac{\omega_{H,T}(t)}{2}}\,\mathrm{d}B_t
\end{eqnarray}
 $(B_t)_t$ being independent of $\omega_{H,T}$. In this formulation, $\omega_{H,T}$ plays the role of a stationary, non-Markovian log-volatility exhibiting long memory and rough sample paths.\\

The Log S-fBM model therefore interpolates between two empirically relevant regimes: for $H>0$ it produces rough stochastic volatility with finite correlation range, while in the limit $H\to0$ it converges toward multifractal dynamics displaying logarithmic correlations. The framework thus provides a unified description of volatility fluctuations across time scales within a single construction. From a financial modeling standpoint, this framework was introduced as a stochastic volatility model capable of simultaneously capturing rough and multifractal behaviors in log-volatility trajectories. In particular, numerical investigations reported in Figure 10 of \cite{wu2022rough} illustrate a clear heterogeneity across asset classes: broad market indices are associated with moderately rough dynamics, typically characterized by a Hurst exponent around $H\simeq 0.1$, whereas individual stocks exhibit markedly rougher volatility paths, with estimates of the Hurst parameter closer to $H\simeq 0.01$. This empirical contrast highlights the flexibility of the model in preserving a coherent theoretical structure for both roughness regimes.

\section{Properties of the S-fBM process}
This section is dedicated to the analysis of several structural and pathwise properties of the S-fBM process. These properties stem from the geometric nature of its construction and from the powerly decaying form of its covariance structure. On the one hand, we explore the scaling features that arise from the underlying time scale domain as well as aggregation and characteristics of increments of S-fBM processes. On the other hand, we demonstrate the non semimartingality of the S-fBM  process retrieving the same property of the fractional brownian motion, the latter exhibiting long range dependences and persistence over time.

\subsection{Scaling features and self similarity type properties}

A central ingredient in the construction of the S-fBM is the time--scale domain $C_T(t)$ (see Eq.~(55) in \cite{wu2022rough}). Its geometry is entirely determined by two parameters: the intermittency coefficient $\lambda$, which modulates fluctuation intensity and the correlation limit $T$, which fixes the maximal temporal range of dependence. This domain encodes nontrivial self similarity properties that naturally propagate to the associated Gaussian process and random measure.

\begin{prop}[\textbf{Self similarity of the time scale domain}]
\label{prop:takenakaselfsimilarity}
Let $t \geq 0$ be fixed. For any $\alpha>0$, one has
\begin{equation}
\label{eq:scalingtakenaka}
C_{\alpha T}(\alpha t)=C_{T}(t).
\end{equation}
In addition, for a fixed $t\geq 0$, the family $(C_T(t))_{T>0}$ is increasing in the sense of set inclusion.
\end{prop}

The proof is deferred to Appendix~\ref{app:Scaling properties} (Proof~\ref{proof:ScalingpropertiesTakenaka}). This identity shows that time rescaling can be exactly compensated by a corresponding transformation of the decorrelation parameter. As a consequence, scaling operation does not alter the geometric nature of the domain but only shifts the effective correlation scale.\\\\
This invariance immediately carries over to the S-fBM itself. Indeed, for any $\alpha>0$, the defining relation yields
\begin{eqnarray}
\omega_{H,T}(\alpha t)
= \mu_H + \int_{C_T(\alpha t)} \mathrm{d}G_H.
\end{eqnarray}
Using \eqref{eq:scalingtakenaka}, we inherit a self similarity property from Proposition \ref{prop:takenakaselfsimilarity} for the process $\omega_{H,T}$:
\begin{eqnarray}
\left(\omega_{H,\alpha T}(\alpha t)\right)_{t}\overset{\mathcal{L}}=\left(\omega_{H,T}(t)\right)_{t},
\end{eqnarray}
which shows that rescaling time amounts to modifying the decorrelation scale, while preserving the statistical structure of the process.

The same phenomenon appears at the level of the associated random measure. By definition,
\begin{eqnarray}
M_{H,T}(\alpha t)
= \int_0^{\alpha t} e^{\omega_{H,T}(s)}\,\mathrm{d}s
= \alpha \int_0^{t} e^{\omega_{H,T}(\alpha u)}\,\mathrm{d}u.
\end{eqnarray}
Combining this identity with the previous scaling relation yields:
\begin{eqnarray}
M_{H,\alpha T}(\alpha t)
\overset{\mathcal{L}}=
\alpha M_{H,T}(t).
\end{eqnarray}
Equivalently,
\begin{eqnarray}
M_{H,\alpha T}(\alpha t)
= \alpha \int_0^t e^{\omega_{H,T}(s)}\,\mathrm{d}s.
\end{eqnarray}
Thus, defining a new random measure for arbitrary $\alpha>0$:
\[
\mathrm{d}P^{\alpha}_{H,T}(t):=\mathrm{d}M_{H,\alpha T}(\alpha t),
\]
one obtains the homogeneous self similar relation:
\begin{eqnarray}
\mathrm{d}P^{\alpha}_{H,T}(t)=\alpha\,\mathrm{d}M_{H,T}(t),
\end{eqnarray}
which highlights the intrinsic scale-consistent nature of the construction.\\

The failure of exact self similarity of $\omega_{H,T}$  originates from the affine dependence of the covariance on $\tau^{2H}$ (see Eq.~\eqref{eq:covsfbm}), in contrast with the pure power-law behavior of fractional Brownian motion. Nonetheless, appropriate rescaling of the decorrelation parameter restores invariance properties, which are summarized below.

\begin{prop}
\label{prop:propertislawsfbm}
The following claims hold:
\begin{enumerate}

\item For any fixed $(t,\tau)\in\mathbb{R}_+^2$, the increments of $\omega_{H,T}$ marginally  satisfy:
\begin{eqnarray}
\omega_{H,T}(t+\tau)-\omega_{H,T}(t)
\overset{\mathcal{L}}=
\frac{\lambda}{\sqrt{H(1-2H)}}\,B_{(\tau/T)^{2H}},
\end{eqnarray}
where $(B_u)_{u\geq 0}$ is a standard Brownian motion.

\item For any $(x,r)\in\mathbb{R}_+^2$,
\begin{eqnarray}
\omega_{H,T}(x+rt)=\omega_{r,H,T}(t)+\Omega_r,
\end{eqnarray}
where $\Omega_r$ is a centered Gaussian variable with variance
$\frac{\nu^2}{2}(1-r^{2H})$ and $\omega_{r,H,T}$ is an S-fBM with intermittency coefficient $\lambda r^H$ independent of $\Omega_r$.
\end{enumerate}
\end{prop}
The proof of Proposition \ref{prop:propertislawsfbm} is in Appendix~\ref{app:Scaling properties}.
\\
Beyond the latter properties, the S-fBM displays aggregation and increment behaviors that resemble those of self-similar Gaussian processes, even though exact self-similarity does not hold. Consider $d$ independent S-fBM processes
$\left(\omega^i_{H,T}\right)_{i\in\llbracket 1,d \rrbracket}$
sharing the same Hurst exponent $H$ and correlation scale $T$, but possibly distinct intermittency coefficients $\lambda_i$.
\begin{remark}
Define the process:
\begin{eqnarray}
\omega^d_{H,T}:=\sum_{i=1}^d  \omega^i_{H,T}
\end{eqnarray}
A direct computation shows that, for $t\in \R_+$,
\begin{eqnarray}
\E\big[\omega^d_{H,T}(t)\big]
= -\frac{1}{4H(1-2H)}\sum_{i=1}^d \lambda_i^2,
\end{eqnarray}
and that for $\tau=|t-s|$,
\begin{eqnarray}
\mathrm{cov}\big(\omega^d_{H,T}(t),\omega^d_{H,T}(s)\big)
= \frac{\sum_{i=1}^d \lambda_i^2}{2H(1-2H)}
\left(1-\left(\frac{\tau}{T}\right)^{2H}\right)\mathds{1}_{\tau\leq T}.
\end{eqnarray}
Consequently,
\begin{eqnarray}
\omega^d_{H,T}\overset{\mathcal{L}}=\tilde{\omega}_{H,T},
\end{eqnarray}
where $\tilde{\omega}_{H,T}$ is an S-fBM with the same Hurst exponent and decorrelation scale, but with an effective intermittency coefficient $\sqrt{\sum_{i=1}^d \lambda_i^2}$. This aggregation property reflects a form of Gaussian stability while preserving the specific S-fBM covariance structure.\\
\end{remark}

We conclude by addressing a fundamental limitation of the S-fBM framework: its incompatibility with the semimartingale setting. This feature is shared by a large class of Gaussian Volterra processes (see \cite{Sottinenfredholm}), including the fractional Brownian motion. As shown in \cite{Rogers}, fractional Brownian motion is a semimartingale only in the Brownian case $H=\tfrac{1}{2}$, a consequence of its long memory, non-Markovian structure and path irregularity.

\begin{prop}[\textbf{Non semimartingality of the S-fBM}]
\label{prop:sfbmnotsemimartingale}
For any $H\in\left]0,\frac{1}{2}\right[$ and $T>0$, the process $\omega_{H,T}$ is not a semimartingale.
\end{prop}

The proof of both Proposition is done  in Appendix~\ref{app:sfbmnotsemimartingale}. This result is consistent with the fact that the S-fBM can serve as a building block for fractional Brownian motion. Indeed, as shown in Appendix~A.2 of \cite{wu2022rough}, defining
\begin{eqnarray}
B^H_t=\omega_{H,T}^+(t)-\omega_{H,T}^-(t),
\end{eqnarray}
with
\begin{eqnarray}
\omega_{H,T}^{\pm}(t)=\int_{D^{\pm}(t)} \mathrm{d}G_H,
\qquad
D^{\pm}(t)=C(t)\pm C(0),
\end{eqnarray}
yields a centered Gaussian process whose covariance function is
\begin{eqnarray}
\mathrm{cov}(B^H_t,B^H_s)
=\frac{\nu^2}{2}\left(t^{2H}+s^{2H}-|t-s|^{2H}\right),
\end{eqnarray}
thereby recovering fractional Brownian motion and reinforcing the interpretation of the S-fBM as a foundational object underlying both rough and multifractal stochastic models.

\section{Deviation inequalities}

This section is devoted to the derivation of deviation inequalities for the S-fBM process. Exploiting its Gaussian structure, we obtain explicit and non-asymptotic upper bounds for the tail distributions of several functionals of interest, most notably the supremum of the process evaluated along finite collections of time points. Such inequalities play a crucial role in quantifying extreme fluctuations of the log-volatility modeled by the S-fBM and provide probabilistic guarantees on rare events. Before addressing deviation bounds directly, we first establish preliminary uniform boundedness properties for exponential functionals of the S-fBM, which serve as technical ingredients in the subsequent analysis.\\
\\
For any $x\in \mathbb{R}$, we denote by $B(x,r)$ the closed ball in $\mathbb{R}$ with center $x$ and radius $r>0$. Let $n\in \mathbb{N}$ and consider a set $D$ covered by a countable union of closed balls, ie: $D\subset \cup_{i=1}^n B(x_i,r_i)$, where $\left(x_i\right)_{i\in \llbracket 1,n \rrbracket}\in \mathbb{R}^n$. For any vector $x\in \mathbb{R}^n$, we introduce the notations
\[
x^{*}:=\underset{i\in \llbracket 1,n \rrbracket}{\sup}\{x_i\},
\qquad
x_{*}:=\underset{i\in \llbracket 1,n \rrbracket}{\inf}\{x_i\}.
\]
These quantities will naturally appear in the control of the supremum of the process over $D$.\\\\
We consider the maximum fluctuation process:
\begin{eqnarray}
\forall u \in \mathbb{R}_{+},\quad
\omega^{x}_{H,T}(u):=\sup_{i\in \llbracket 1,n \rrbracket}\omega^{x_i}_{H,T}(u),
\end{eqnarray}
where $\omega^{x_i}_{H,T}(u):=\omega_{H,T}(u)-\omega_{H,T}(x_i)$.\\\\
This construction isolates the relative fluctuations of the S-fBM around a finite set of reference points and allows us to control local extrema of the process on arbitrary compact domains.\\\\
We have the following preliminary result.
\begin{prop}[\textbf{Uniform boundedness}]
\label{prop:Expectationbound}
There exists a positive constant $C>0$ such that for any $\varepsilon>0$, the following holds:
\begin{eqnarray}
\E\left(\sup_{u\in D}e^{\omega^{x}_{H,T}(u)} \right)
\leq
2n\sqrt{\frac{\pi}{\eta_T(H,\lambda^2)}}
\exp\left(
\frac{(\eta_{T}(H,\lambda^2) a_{*}(\varepsilon)+1)^2}{\eta_{T}(H,\lambda^2)}+\eta_{T}(H,\lambda^2)C^2 a_{*}(\varepsilon)^2
\right),
\end{eqnarray}
where:
\begin{itemize}
    \item $\eta_{T}(h,l)=\frac{T^{2h}h(1-2h)}{2nl r^{*\,2h}}$
    \item $a_{*}(\varepsilon)=\varepsilon \sqrt{\ln\!\left(
\frac{r_{*}}{\sqrt{T}}
\left(\frac{\lambda}{2\varepsilon}\right)^{\frac{1}{2H}}
\big(H(1-2H)\big)^{-\frac{1}{4H}}
\right)}$
\end{itemize}

\end{prop}

The proof is given in Appendix~\ref{app:uniformboundedness}.\\\\
This proposition establishes that, despite the rough nature of the S-fBM paths, exponential functionals of its localized supremum remain integrable. The bound highlights the combined influence of the Hurst exponent $H$, the intermittency parameter $\lambda$, the correlation limit $T$ and the geometric complexity of the domain through $n$ and $r^*$. In particular, rougher processes (smaller $H$) or stronger intermittency lead to weaker concentration properties.\\

At this stage, we now turn to deviation inequalities for the S-fBM process itself. The following theorem provides explicit tail probability bounds for the supremum of the process evaluated at finitely many time points. The proof, based on Gaussian concentration arguments, is deferred to Appendix~\ref{app:deviationineq}.

\begin{theo}[\textbf{Tail bounds}]
\label{theo:upperbounddevSfbm}
For an arbitrary tuple $(t_j)_{1\leq j \le m} \in \mathbb{R}_{+}^m$, the following claims hold:
\begin{enumerate}
\item 
\begin{eqnarray}
\forall \delta>0,\quad
\Prob\left(\underset{1 \leq j \leq m }{\sup}|\omega_{H,T}(t_j)|\geq \delta \right)
\leq
2m
\exp\left(
-\frac{\left(\delta-\frac{\lambda^2}{4H(1-2H)} \right)^2H(1-2H)}{\lambda^2}
\right)
\left(e^{\delta}+1 \right).\nonumber \\
\end{eqnarray}

\item 
\begin{eqnarray}
\limsup_{x \to \infty}
x^{-2}\ln\left(
\Prob\left(\underset{1 \leq j \leq m}{\sup}\left\lvert\omega_{H,T}(t_j)\right\rvert \geq x \right)
\right)
\leq
-\frac{1}{\nu^2\left(4+\frac{\nu^2}{8}\right)}.
\end{eqnarray}
\end{enumerate}
\end{theo}

The first bound provides a non-asymptotic concentration inequality, valid for any finite sampling of the process, while the second result characterizes the asymptotic Gaussian decay of the tail distribution. Together, they show that extreme fluctuations of the S-fBM remain exponentially unlikely, with a decay rate controlled by the effective variance parameter $\nu^2$.

We can extend the previous result to $d$ independent S-fBM processes $\left(|\omega_{H_i,T}(\cdot)|\right)_{i\in \llbracket 1,d \rrbracket}$ by defining
\[
\left\lvert\omega^{\max}_t\right\rvert
:=\underset{i\in \llbracket 1,d \rrbracket}{\sup}|\omega_{H_i,T}(t)|.
\]
Using the union bound, we obtain
\begin{eqnarray}
\Prob\left(\underset{1 \leq l \leq m }{\sup}\left\lvert\omega^{\max}_{t_l}\right\rvert\geq \delta \right)
\leq
2m\sum_{i=1}^d
\exp\left(
-\frac{\left(\delta-\frac{\lambda_i^2}{4H_i(1-2H_i)} \right)^2H_i(1-2H_i)}{\lambda_i^2}
\right)
\left(e^{\delta}+1 \right).
\end{eqnarray}
A natural uniform upper bound then reads
\begin{eqnarray}
\Prob\left(\underset{1 \leq l \leq m }{\sup}\left\lvert\omega^{\max}_{t_l}\right\rvert\geq \delta \right)
\leq
2md
\exp\left(
-\frac{\left(\delta-\frac{\lambda_{*}^2}{4H^{*}(1-2H^{*})} \right)^2H^{*}(1-2H^{*})}{\lambda^{*2}}
\right)
\left(e^{\delta}+1 \right).
\end{eqnarray}
This bound clearly reveals a curse of dimensionality, as the prefactor grows linearly with the number of independent components.

Exactly as in Theorem~\ref{theo:upperbounddevSfbm}, the previous arguments extend to a $d$-dimensional S-fBM with independent marginals. Following the computations in the proof of Theorem~\ref{theo:upperbounddevSfbm}, the union bound yields
\begin{eqnarray}
\forall i\in \llbracket 1,d \rrbracket,\quad
\Prob\left(\underset{0\leq j \leq m}{\sup}\left\lvert\omega_{H_i,T}(t_j)\right\rvert\geq x \right)
\leq
e^{\frac{-x^2}{\nu_i^2\left(4+\frac{\nu_i^2}{8}\right)} }
\left(
1+\frac{1}{2}\left(e^{m^2\frac{\nu_i^2}{4+\frac{\nu_i^2}{8}}}-1\right)
+\frac{1}{\frac{1+\frac{\nu_i^2}{32}}{m^2}-1}
\right).
\end{eqnarray}
This expression suggests that the survival function of
$\underset{0\leq j \leq m}{\sup}\left\lvert\omega^{\max}_{t_j}\right\rvert$
exhibits exponential decay. More precisely,
\begin{eqnarray}
\exists C_d^m>0,\ \forall x\in \mathbb{R}_+,\quad
\Prob\left(\underset{0\leq j \leq m}{\sup}\left\lvert\omega^{\max}_{t_j}\right\rvert\geq x \right)
\leq
C_d^m e^{-A_{\nu}x^2},
\end{eqnarray}
where
\[
\begin{cases}
\displaystyle
C_d^m=d\underset{i\in \llbracket 1,d \rrbracket}{\sup}
\left\{
1+\frac{1}{2}\left(e^{m^2\frac{\nu_i^2}{4+\frac{\nu_i^2}{8}}}-1\right)
+\frac{1}{\frac{1+\frac{\nu_i^2}{32}}{m^2}-1}
\right\},\\[1.2ex]
\displaystyle
A_{\nu}=\frac{1}{\nu_{*}^2\left(4+\frac{\nu_{*}^2}{8}\right)}.
\end{cases}
\]
On the other hand,
\begin{eqnarray}
\ln\left(
\Prob\left(\underset{0\leq j \leq m}{\sup}\left\lvert\omega^{\max}_{t_j}\right\rvert\geq x \right)
\right)
\leq
-A_{\nu}x^2
+\ln(C_d^m),
\end{eqnarray}
which implies the following tail deviation inequality:
\begin{eqnarray}
\limsup_{x\rightarrow \infty}
x^{-2}
\ln\left(
\Prob\left(\underset{0\leq j \leq m}{\sup}\left\lvert\omega^{\max}_{t_j}\right\rvert \geq x \right)
\right)
\leq
-A_{\nu}.
\end{eqnarray}
Similar large deviation bounds can be derived at the level of the Log S-fBM random measure. In particular, it quantifies the concentration of the Log S-fBM random measure on an arbitrary compact as well as around the integrated S-fBM defined as:
\begin{eqnarray}
    \forall I\subset K,\tab \Omega_{H,T}(I):=\frac{1}{\lambda}\int_I \left(\omega_{H,T}(u)-\mu_{H} \right) du
\end{eqnarray}
where $K\subset \R_+$ such that $\left|K \right|\leq T $.
In fact, $\Omega_{H,T}$ is of a particularly involved in the small intermittency approximation, see Proposition 4 in \cite{wu2022rough}. This proposition describes the behaviour of the joint moments of the logarithm of the normalized random measure when the intermittency parameter $\lambda$ goes to $0$. More precisely, the joint moments of the Log S-fBM random measure coincide with those of the Gaussian field $\Omega_{H,T,\Delta}$, up to a multiplicative factor $\lambda^{n}$. Consequently, in the small intermittency regime, the logarithm of the normalized measure behaves, at the level of its finite-dimensional moments, like $\Omega_{H,T,\Delta}$. Related results about the small intermittency approximation will be referred to in Section \ref{sec:scaleinvarianceSIA}.

\begin{theo}[\textbf{Tail bounds}]
 \label{thm:deviationineqmrm}
 For any compact $D\subset \R_+$, the following holds:
    \begin{enumerate}
        \item \begin{eqnarray}
\limsup_{x\to\infty}
x^{-2}
\ln
\Prob\left(
\left|
\ln\left(\frac{M_{H,T}(D)}{|D|}\right)\right| \ge x
\right)
\le
-\frac{H(1-2H)}{\lambda^2}.\nonumber
\end{eqnarray}
 \item\begin{eqnarray}
\limsup_{x\to\infty}
x^{-2}
\ln\left(
\Prob\left(
\left|
\ln\left(\frac{M_{H,T}(D)}{|D|}\right)
- \lambda \frac{\Omega_{H,T}(D)}{|D|}
\right|\right)
\ge x
\right)
\le
-\frac{H(1-2H)}
{2\lambda^2\,}\left(\frac{T}{\operatorname{diam}(D)}\right)^{2H}.
\nonumber
\end{eqnarray}
Here, $\operatorname{diam}(D):=\underset{(t,s)\in D^2}\sup \left|
t-s
\right|$ the diameter of $D$.
    \end{enumerate}
\end{theo}
The proof of Theorem \ref{thm:deviationineqmrm} is given in Appendix \ref{app:deviationineqmrm}.\\

These deviation inequalities describing the tail behaviour of the Log S-fBM random measure over a compact set $D$. The first bound shows that the random variable 
$\ln\left(\frac{M_{H,T}(D)}{|D|}\right)$ has sub-Gaussian tails, with an exponential decay of order $e^{-\alpha x^{2}}$ as $x\to\infty$, where the rate constant depends explicitly on the Hurst parameter $H$ and the intermittency parameter $\lambda$. This highlights the respective roles of these parameters: the intermittency parameter controls the amplitude of extreme fluctuations, while the factor $H(1-2H)$ reflects the roughness of the underlying fractional structure. The second inequality refines this result by quantifying the deviations of the Log S-fBM random measure from its Gaussian approximation $\lambda\,\frac{\Omega_{H,T}(D)}{|D|}$. It shows that this difference also satisfies a sub-Gaussian concentration bound, with a rate that additionally depends on the observation scale through $\left(\frac{T}{\operatorname{diam}(D)}\right)^{2H}$, revealing the multiscale nature of the model: smaller sets yield stronger concentration. Overall, these estimates provide a quantitative description of the Gaussian-type concentration properties of the logarithmic fluctuations and justify the Gaussian approximation of the log-measure at large deviation scales.\\

These results complement the small intermittency approximation framework by establishing non-asymptotic probabilistic bounds on the fluctuations of $\ln\left(\frac{M_{H,T}(.)}{|.|}\right)$. Such concentration inequalities provide a theoretical justification for the stability of empirical covariance estimators used in calibration procedures. In the regime where the intermittency parameter $\lambda^2$ is small, the fluctuations of the Log S-fBM random measure remain tightly controlled, which supports the validity of the small intermittency approximation employed in \cite{wu2022rough,bacry2008log}. The theorem above formalizes this behavior through Gaussian-type large deviation bounds.

\section{Hypothesis Testing: Multifractal versus Rough Volatility}
\label{sec:HypothesisTesting}
One of the central questions in the statistical analysis of volatility concerns the identification of its dependence structure. While classical multifractal models describe the log-volatility field through logarithmic correlations, more recent evidence has suggested that log volatility may instead exhibit a rougher behavior, characterized by a non-zero Hurst exponent. Distinguishing between these two regimes is therefore of fundamental importance, both for understanding the underlying stochastic dynamics and for selecting an appropriate model for forecasting, calibration, and risk management.\\

Within the Log S-fBM framework, these two behaviors are naturally unified through the Hurst parameter \(H\). The critical case \(H=0\) recovers the logarithmically correlated (multifractal) regime, whereas \(H\neq0\) corresponds to the generalized rough volatility model. \\

This naturally leads to the statistical problem of testing:

\[
\hyperref[sec:HypothesisTesting]{\mathcal{H}_0}:H=0
\qquad\text{versus}\qquad
\hyperref[sec:HypothesisTesting]{\mathcal{H}_1}:H\neq0,
\]
thereby assessing whether the observed data provide statistical evidence in favor of genuine roughness beyond the multifractal model.

The asymptotic theory developed in the previous sections provides an explicit solution to this problem. In particular, the direct log volatility statistic
\[
\Sigma_n=\sum_{j=1}^{n}\log\!\left(\frac{M_{H,T,\Delta}((j-1)\Delta)}{\Delta}\right)
\]
satisfies an asymptotic Gaussian limit after normalization by the intermittency parameter \(\lambda\), with a limiting variance that admits a closed-form expression depending explicitly on \(H\). This characterization enables the construction of a simple hypothesis test whose null distribution is analytically available and whose asymptotic power can be derived explicitly. Unlike threshold-based procedures, the proposed approach exploits the complete information contained in the observed log volatilities yielding a direct and interpretable testing procedure for discriminating between the multifractal and rough volatility regimes.

\subsection{A central limit theorem in the small intermittency regime}
The aim of this section is to establish a central limit theorem for aggregated statistics of the
log-volatility field in the small intermittency regime, more details about this approximation is given in the upcoming Section \ref{sec:scaleinvarianceSIA}). Although the underlying Gaussian field
allows for exact finite-sample representations, the weak intermittency limit provides a natural
framework to study the asymptotic behavior of nonlinear functionals and their dependence
structure. By exploiting the explicit covariance structure and a perturbative expansion in
$\lambda^2$, we derive the Gaussian limiting behavior of the statistic while preserving the
effects of the temporal correlations induced by the model.\\

To start with, for every integer $n\in \N$, we denote
\[
L_n=n\Delta\leq T,
\]
we observe that the aggregated statistic:
\begin{eqnarray}
    S_n = \sum_{j=1}^n
\Omega_{H,T,\Delta}((j-1)\Delta)
\end{eqnarray}
satisfies the exact identity

\begin{eqnarray}
\label{eq:gaussianityofSn}
    \lambda S_n
\sim
\mathcal N
\left(
0,
V_n(H)
\right)
\end{eqnarray}
where:
\begin{eqnarray}
    V_n(H) \;=\; \frac{L_n^2}{2H(1-2H)}\left(T^{2H}-\frac{2L_n^{2H}}{(2H+1)(2H+2)}\right)
\end{eqnarray}

thanks to the covariance structure of Eq.~\eqref{eq:autocovOmegaDelta}.
\\
For the sake of simplicity, let $Y_j := \ln\big(M_{H,T,\Delta}((j-1)\Delta)/\Delta\big)$, $j=1,\dots,n$, and
$\Sigma_n:=\sum_{j=1}^nY_j$ the actual observable variable. The following result, proved in Appendix \ref{app:proofpropCLTsmallinterm}, holds.
\begin{prop}
\label{prop:CLTsmallinterm}
Fix $n\in\mathbb N$, $\Delta>0$, and a fixed $T>0$ with $n\Delta\le T$:
\begin{eqnarray}
    \frac1\lambda\,\Sigma_n \;\xrightarrow[\lambda\to 0]{\ d\ }\; \mathcal N(0,V_n(H))
\end{eqnarray}
\end{prop}

The following Corollary, whose prove is in Appendix \ref{app:proofcorrolaryV02asH0}, characterizes the behavior of the variance term in the limiting
regime $H\to0$, which plays an important role in the construction of statistical tests
discriminating multifractal and rough volatility dynamics. 
\begin{Corollary}[Limit of $V_n(H)$ as $H\to0$]\label{prop:V02asH0}
\begin{eqnarray}
V_n(0) := \lim_{H\to0}V_n(H) = \frac{L_n^2}{2}\Big[2\ln(T/L_n)+3\Big],
\end{eqnarray}
\end{Corollary}
In particular, the small-$H$
limit corresponds to the boundary between different scaling behaviors and is required to
properly define the null and alternative distributions involved in the testing procedure.
Although the closed-form expression of $V_n(H)$ contains an apparent singularity through the
factor $1/H$, this singularity is removable: the vanishing of the numerator exactly compensates
the diverging prefactor, yielding a finite and explicit limiting variance. This result provides
the necessary regularity of the variance structure for the asymptotic analysis of the proposed
multifractal versus rough volatility hypothesis test.
\subsection{The testing statistic}

The asymptotic distribution derived above provides a direct testing procedure for
discriminating the multifractal critical regime from rough volatility alternatives. The null
hypothesis corresponds to the boundary case $H=0$, while non-zero values of $H$ describe
departures from this critical scaling regime.\\

We consider the following statistic:

\begin{eqnarray}
\label{eq:Zstatistic}
    Z_n(H):=\frac{\Sigma_n}{\lambda\sqrt{V_n(H)}}
\end{eqnarray}
that is Gaussian in the small intermittency limit as advocated in Proposition \ref{prop:CLTsmallinterm}.

Under \hyperref[sec:HypothesisTesting]{$\mathcal{H}_0$}, the limiting variance is explicitly given by $V_n(0)$, which motivates the
standardized statistic:

\[
Z_n(0):=\frac{\Sigma_n}{\lambda\sqrt{V_n(0)}} .
\]

We denote the variance-ratio:

\begin{eqnarray}
    R^2(H):=\frac{V_n(H)}{V_n(0)}
\end{eqnarray}
\begin{mydeff}[Test power]
For $\alpha\in $, consider the rejection region $\mathcal R_\alpha:=\{|Z_n|>z_{1-\alpha/2}\}$,
$z_{1-\alpha/2}:=\Phi^{-1}(1-\alpha/2)$ and $\Phi$ is the standard Gaussian cumulative distribution function. The \textit{power function} of the test, evaluated at a true parameter value $H^\ast\ne0$, is
\begin{eqnarray}
    \pi(H^\ast) \;:=\; \mathbb P_{H^\ast}\big(Z_n\in\mathcal R_\alpha\big)
\;=\; \mathbb P_{H^\ast}\big(|Z_n|>z_{1-\alpha/2}\big)
\end{eqnarray}
which corresponds to the probability of rejecting \hyperref[sec:HypothesisTesting]{$\mathcal{H}_0$}.
\end{mydeff}
\begin{remark}
Under the model's own CLT,
$Z_n\to_d\mathcal N(0,1)$ at $H^\ast=0$ and $Z_n\to_d\mathcal N(0,R^2(H^\ast))$ at $H^\ast\ne0$, giving the closed form
$\pi(H^\ast)=1-\Phi\left(\frac{z_{1-\alpha/2}}{R(H^\ast)}\right)+\Phi\left(-\frac{z_{1-\alpha/2}}{R(H^\ast)}\right)$ used for the theoretical
predictions. 
\end{remark}
$Z_n$ as defined in Eq.~\eqref{eq:Zstatistic} requires $\lambda$. In practice, $\lambda$ is unknown and
must be estimated separately, since neither statistic can identify $\lambda$ and $H$ separately at
this order (both confound them only through $\lambda^2V_n(H)$, up to a multiplicative constant
factor). Substituting a consistent $\hat\lambda$, for instance the one from the GMM method used by Wu \textit{et al.} in \cite{wu2022rough}, preserves the stated asymptotic distribution.

\section{The scale invariance property}
\subsection{Overview}
It is well known so far that the volatility exhibits rich scaling behavior across time scales. Empirical studies have shown that the variance of returns scales with time as a power law, i.e., $\mathrm{Var}[r(\tau)] \propto \tau^{\alpha}$ with $\alpha \in [0,1]$, deviating from the Brownian motion assumption where $\alpha = 1$ (see \cite{Cont, dacorogna2001introduction}). This scaling law is accompanied by stylized facts such as volatility clustering, long memory in squared returns and heavy tails. To account for these phenomena, several modeling approaches have been developed.\\

Multifractal models of asset returns are designed to capture the heterogeneous scaling behavior observed in financial time series, particularly in volatility. Unlike models that assume constant or homogenous scaling (e.g., Brownian motion), multifractal processes allow different moments of return distributions to scale with time in different ways. Empirically, the $q$-th moment of asset returns satisfies the relation $\mathbb{E}[|r(\tau)|^q] \sim \tau^{\zeta(q)}$, where $\zeta(q)$ is a non-linear concave function of $q$, indicating multifractal behavior (see \cite{bacry2001multifractal, calvet2002multifractality}). The foundational work of Mandelbrot et al.\ \cite{mandelbrot1997multifractal} introduced the Multifractal Model of Asset Returns (MMAR), in which asset price dynamics are governed by a Brownian motion subordinated to a multifractal trading time. This approach captures both heavy tails and long-range dependence in volatility. Bacry et al.\ \cite{bacry2001multifractal} later proposed the MRW that preserves analytical tractability while reproducing empirical scaling laws. Calvet and Fisher in \cite{calvet2002multifractality} introduced the Markov-Switching Multifractal (MSM) model, which uses a discrete-time multiplicative cascade mechanism based on Markov chains. This model offers a parsimonious yet powerful structure for capturing multifractality and long memory in volatility and has been successfully applied in volatility forecasting and risk management.\\

Empirical studies of high-frequency financial data have shown that the log-volatility process displays H\"older regularity significantly lower than that of Brownian motion, suggesting that it is better modeled by a fractional Brownian motion (fBM) with Hurst exponent $H < 0.5$, see \cite{gatheral2018volatility, bennedsen2017hybrid}. This motivates the development of rough volatility models, in which the log-volatility process satisfies a scaling law of the form:
\begin{eqnarray}
\label{eq:scaleinvariancesquare}
    \mathbb{E}[|\log \sigma_{t+\tau} - \log \sigma_t|^2] \sim \tau^{2H},
\end{eqnarray}

indicating that the increments of log-volatility exhibit power-law decay with respect to the time lag $\tau$, where $H$ typically lies between $0.1$ and $0.2$ in empirical studies.

The prototypical example is the Rough Fractional Stochastic Volatility (RFSV) model, introduced by Gatheral et al.\ \cite{gatheral2018volatility}, where the log-volatility is driven by a fractional Brownian motion with low Hurst exponent:
\[
\log \sigma_t = \mu + \nu W_t^H,
\]
with $W_t^H$ denoting fBM. This model captures the fine-scale irregularity of volatility paths, consistent with the observed roughness in realized volatility measures. Importantly, rough volatility models maintain tractability in option pricing and calibration, while improving empirical fit, especially at high frequencies (see \cite{el2019characteristic}). The roughness of volatility implies that volatility is not a semimartingale and its pathwise properties invalidate standard Itô calculus for volatility dynamics. Nevertheless, rough models preserve key stylized facts, such as volatility clustering and long memory in the log-volatility process, through a parsimonious and structurally coherent framework based on fractional scaling.\\

Together, these models provide complementary perspectives on the scaling properties of volatility, with rough volatility capturing fine-scale irregularities and multifractal models capturing broader heterogeneity in scaling behavior.\\\\
The aim of the upcoming section is to investigate the scaling properties of the Log S-fBM model.
\subsection{Scale invariance of the Log S-fBM model via the small intermittency approximation}
\label{sec:scaleinvarianceSIA}
The small intermittency approximation is a technique that consists in computing generalized moments of log volatilities in the limit where $\lambda^2$ goes to 0. This was used in \cite{wu2022rough} in order to calibrate a Log S-fBM model to observed data and earlier in \cite{bacry2008log} in the context of the multifractal random walk model as well as the Log S-fBM model (in Section 4 in \cite{wu2022rough}) thanks to Proposition 4 in \cite{wu2022rough}. In few words, the approach relies on examining the multifractal random measure over a time interval of length $\Delta$, denoted $M_{H,T,\Delta}$, or alternatively its logarithmic form $\ln(M_{H,T,\Delta})$. 
The idea is to match the theoretical autocovariance structure given in closed form with the empirical counterparts of either 
$\tau \mapsto \mathrm{Cov}\!\left(M_{H,T,\Delta}(t),\, M_{H,T}(t+\tau)\right)$
or in the logarithmic case $\tau \mapsto \mathrm{Cov}\!\left(\ln(M_{H,T,\Delta}(t)),\, \ln(M_{H,T}(t+\tau))\right)$.
In practice (see \cite{wu2022rough}), the logarithmic version is typically employed, since the closed-form autocovariance can be obtained under a small intermittency approximation. 
These theoretical expressions are then compared with those computed from either simulated data or observed time series in order to perform the estimation.\\

Here in, we are interested in the scale invariance properties of the $p$ order moments of the Log S-fBM model. In the rough volatility setting, Gatheral \textit{et al.} in \cite{gatheral2018volatility} demonstrated that when the $p$-order moments of the log volatility is proportional to $\tau^{qH}$ when $\tau$ is the time scale considered, which is the generalization of Eq.~\eqref{eq:scaleinvariancesquare} . This was reported when the log volatility is a fractional Ornstein Uhlenbeck process (see Corollary 3.1 in \cite{gatheral2018volatility}) as well as for $S\&P500$ and $NASDAQ$ market data (see figure (2.4) and (2.5)  again in \cite{gatheral2018volatility}). On the other side, the multifractal random walk model has shown to have similar behaviour. In particular, in \cite{Kozhemyak}, the log density of MRW increment showed a dependence ot the scale (see figure 3.5, section 3.4 in \cite{Kozhemyak}). Consequently, a natural concern is to investigate the scale dependence of the $p$-order moment of the Log S-fBM model. Particularly, we present closed form formulas in the small intermittency regime.\\

For any process $x$, we introduce the increment operator:
\begin{eqnarray}
    \forall t,\tau>0, \tab \delta_{\tau}x_t:=x_{t+\tau}-x_t \nonumber
\end{eqnarray}
to stand for the increment process. In particular, for the Log S-fBM process $X$ with $\tau\geq 0$:
\begin{eqnarray}
    \delta_{\tau}X_t = X_{t+\tau} - X_t \footnotemark
    \label{eq:deltaXtau}
\end{eqnarray}

For $p\in \N$, we consider for any fixed $t>0$:

\begin{eqnarray}
\label{eq:empiricalqordermoments}
\begin{cases*}
     m^{M}_{H,T,\tau}(p):=\E\left( \left|\delta_{\tau}M_{H,T}(t)\right|^p\right)\\
     m^{X}_{H,T,\tau}(p):=\E\left( \left|\delta_{\tau}X_t\right|^p\right)\\
     m^{log}_{H,T,\tau}(p):=\E\left( \left|\delta_{\tau}\ln\left(M_{H,T}(t)\right)\right|^p\right)\\
     m^{\omega}_{H,T,\tau}(p):=\E\left( \left|\delta_{\tau}\omega_{H,T}(t)\right|^p\right)
\end{cases*}
\end{eqnarray}
We can derive in close form small intermittency approximations of the first and second $q^{th}$ order moments.
\begin{theo}[\textbf{scale invariance}]
\label{theo:scaleinvariance}
For any $(t,\tau)\in \R_{+}^2$ and $p\in \N$:
\begin{eqnarray}
    \tab          \begin{cases}
        \ln\left( m^{M}_{H,T,\tau}(p)\right)\underset{\lambda^2\rightarrow 0}\sim p\ln\left(\tau\right)+\frac{\lambda^2}{2\tau^2}\left(p^2-p\right)\left(\frac{\tau^2}{2H (1-2H)}-\frac{2|\tau|^{2H+2}}{4HT^{2H}(1-4H^2)(H+1)}\right)\\
        \ln\left( m^{X}_{H,T,\tau}(p)\right)\underset{\lambda^2\rightarrow 0}\sim \frac{p}{2}\ln\left(\tau\right)+\frac{\lambda^2}{\tau^2}\frac{p}{2}\left(\frac{p}{2}-1\right)\left(\frac{\tau^2}{2H (1-2H)}-\frac{2|\tau|^{2H+2}}{4HT^{2H}(1-4H^2)(H+1)}\right)+\frac{p}{2}\ln(2)+\ln\left(\frac{\Gamma\left(\frac{p+1}{2}\right)}{\sqrt{\pi}} \right)
    \end{cases}\nonumber\\
\end{eqnarray}
where $\Gamma(x) = \int_{0}^{+\infty} t^{x-1} e^{-t}\, dt,
\tab x>0$.

\end{theo}
The proof is in Appendix \ref{app:scaleinvarianceproof}.

\section{Numerical experiments}

As mentioned earlier in the introduction, a central objective of the Log S-fBM framework is to provide a unified description of volatility dynamics that encompasses both multifractal and rough volatility behaviors. While this unification is supported by theoretical results on covariance structures and asymptotic limits, it is also essential to investigate how these properties manifest at the level of finite-sample distributions. 
\subsection{Hypothesis testing}
This section evaluates the finite-sample performance of the proposed hypothesis test through Monte Carlo simulations. While the asymptotic theory establishes the limiting distribution of the test statistic under the null and characterizes its asymptotic power under fixed alternatives, practical applications rely on finite observation windows. The objective of these experiments is therefore to quantify the empirical rejection probability of the test as a function of the sample size and the true Hurst parameter.\\

We consider the testing problem as in Section \ref{sec:HypothesisTesting}:
\[
\hyperref[sec:HypothesisTesting]{H_0}:H=0
\qquad\text{versus}\qquad
\hyperref[sec:HypothesisTesting]{H_1}:H\neq0,
\]
We sample a Log S-fBM process with a non vanishing Hurst exponent and evaluate the Hypothesis testing accuracy via estimating the power function by Monte Carlo simulations.\\
First and foremost, we evaluate in Appendix \ref{app:Hypothesistestingappendix} the convergence in distribution in the small intermittency regime of Eq.~\eqref{eq:convgcelambdasmalllogOmega}. and see that it is consistent with the small intermittency regime described in \cite{wu2022rough}. Now we evaluate the Central limit theorem of Proposition \ref{prop:CLTsmallinterm}.
\begin{figure}[H]
    \centering
    
    \begin{subfigure}[b]{0.85\textwidth}
        \centering
        \includegraphics[width=\textwidth]{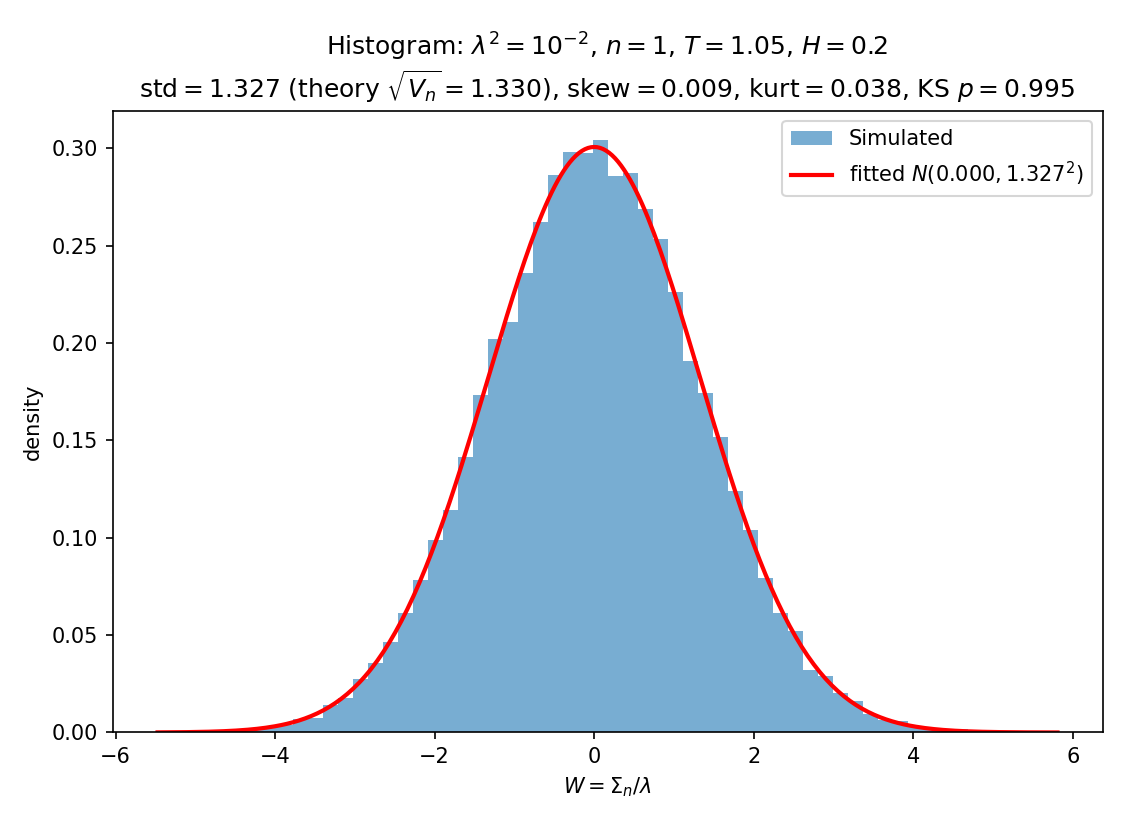}
        \caption{Empirical distribution}
        \label{fig:LogdensityH0005}
    \end{subfigure}
    
    \vskip\baselineskip
    
    \begin{subfigure}[b]{0.85\textwidth}
        \centering
        \includegraphics[width=\textwidth]{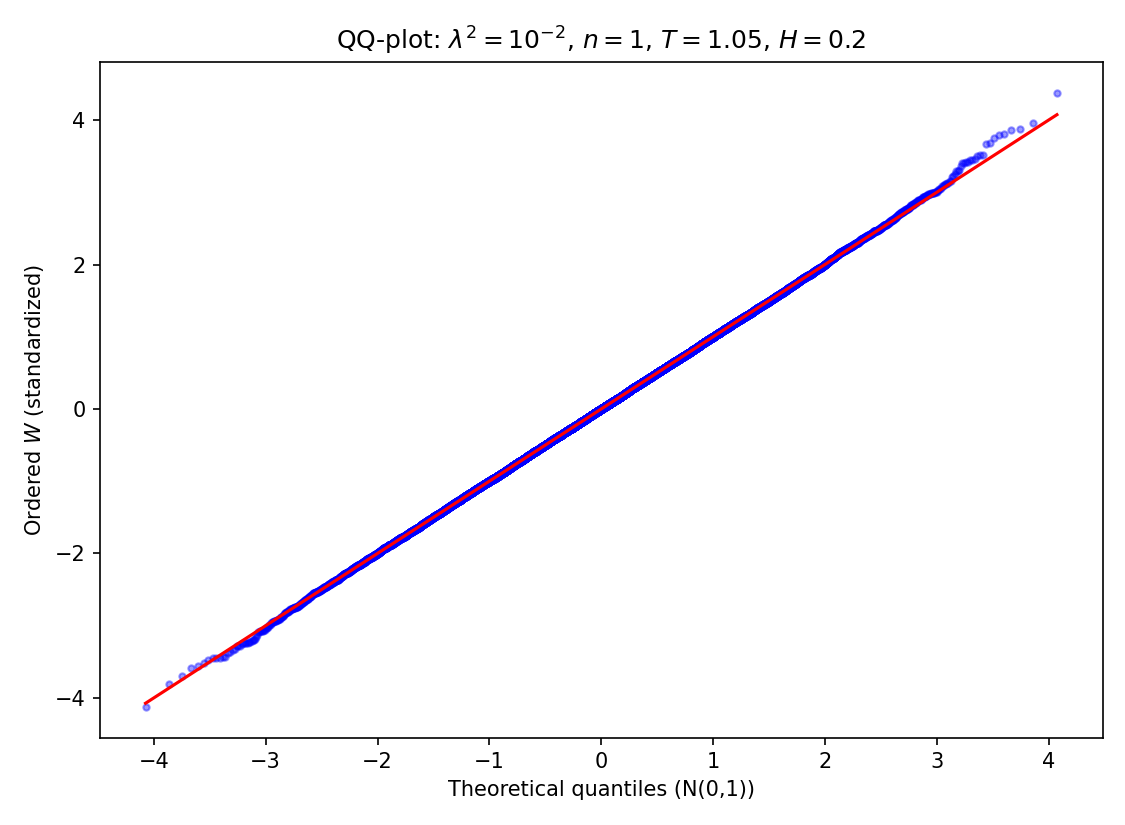}
        \caption{Quantile-Quantile plot}
        \label{fig:LogdensityH005}
    \end{subfigure}
    
    \caption[Short caption for table of figures]{%
    Numerical evidence of the Central limit theorem in the small intermittency regime (top) with the Quantile-Quantile plot of $\frac{\Sigma_n}{\lambda\sqrt{V_n(H)}}$ against the standard Gaussian benchmark.}
    
    \label{fig:scaleinvariancedensities}
\end{figure}
Now, we are interested in the proper Hypothesis testing task. For that sake, we simulate the Log S-fBM process and test the null Hurst exponent Hypothesis. the test statistic given for a predefined risk $\alpha \in \left]0,1\right[$ by:
\begin{eqnarray}
    Z_n = \frac{\Sigma_n}{\lambda\sqrt{V_n(0)}}
\end{eqnarray}
and we reject  \hyperref[sec:HypothesisTesting]{$\mathcal{H}_0$} when $|Z_n|>z_{1-\alpha/2}$ (Type I test).\\

We evaluate the empirical power via the Monte Carlo estimate:

\begin{eqnarray}
    \hat\pi(H^\ast) \;:=\; \frac1N\sum_{i=1}^N \mathbf 1_{\{|Z_n^{(i)}|>z_{1-\alpha/2}\big\}}
\end{eqnarray}
which corresponds to the Monte Carlo proxy of the rejection probability. One has the following results for multiple Hurst exponents.
\begin{table}[ht]
\centering

\begin{tabular}{rrrrr}
\toprule
$H$ & $V_n(H)$ & $R=\sqrt{V_n(H)/V_n(0)}$ & empirical power & theoretical power $\pi(H)$\\
\midrule
0.05 & 59.49 & 1.0000 & 0.0505 & 0.0500\\
0.10 & 73.51 & 1.1116 & 0.0797 & 0.0779\\
0.15 & 93.10 & 1.2510 & 0.1130 & 0.1172\\
0.20 & 121.28 & 1.4278 & 0.1696 & 0.1698\\
0.25 & 163.59 & 1.6583 & 0.2407 & 0.2372\\
0.30 & 231.19 & 1.9714 & 0.3251 & 0.3201\\
0.35 & 350.33 & 2.4267 & 0.4170 & 0.4193\\
0.40 & 600.00 & 3.1758 & 0.5344 & 0.5371\\
0.45 & 1375.82 & 4.8091 & 0.6884 & 0.6836\\
\bottomrule
\end{tabular}
\caption{Empirical and theoretical power of the test for different values of the Hurst exponent $H$. 
The empirical power (obtained by Monte Carlo simulation) is compared with the theoretical power
$\pi(H)$. Here $n=5$ and $\Delta=1$}
\label{tab:power_comparison}
\end{table}

The results in Table~\ref{tab:power_comparison} show a close agreement between the empirical and theoretical power across all values of $H$, with absolute deviations not exceeding $0.005$, confirming the validity of the theoretical power formula $\pi(H)$ even at the small sample size $n=5$. Both $V_n(H)$ and the standardized ratio $R$ increase monotonically and non-linearly with $H$, reflecting the growing discriminating power of the test as the alternative moves away from the null ($H=0.05$). The power itself rises from the nominal level $\alpha \approx 0.05$ to nearly $0.69$ at $H=0.45$, with no systematic bias between empirical and theoretical values, the residual discrepancies being consistent with Monte Carlo sampling error.\\

In what follows, we evaluate the impact on the number of observations $n$. We consider the representative values $H=0.01$ and $H=0.1$, corresponding respectively to a nearly critical regime and a more pronounced departure from the null hypothesis.

\begin{figure}[H]
    \centering

    \begin{subfigure}[b]{0.75\textwidth}
        \centering
        \includegraphics[width=\textwidth]{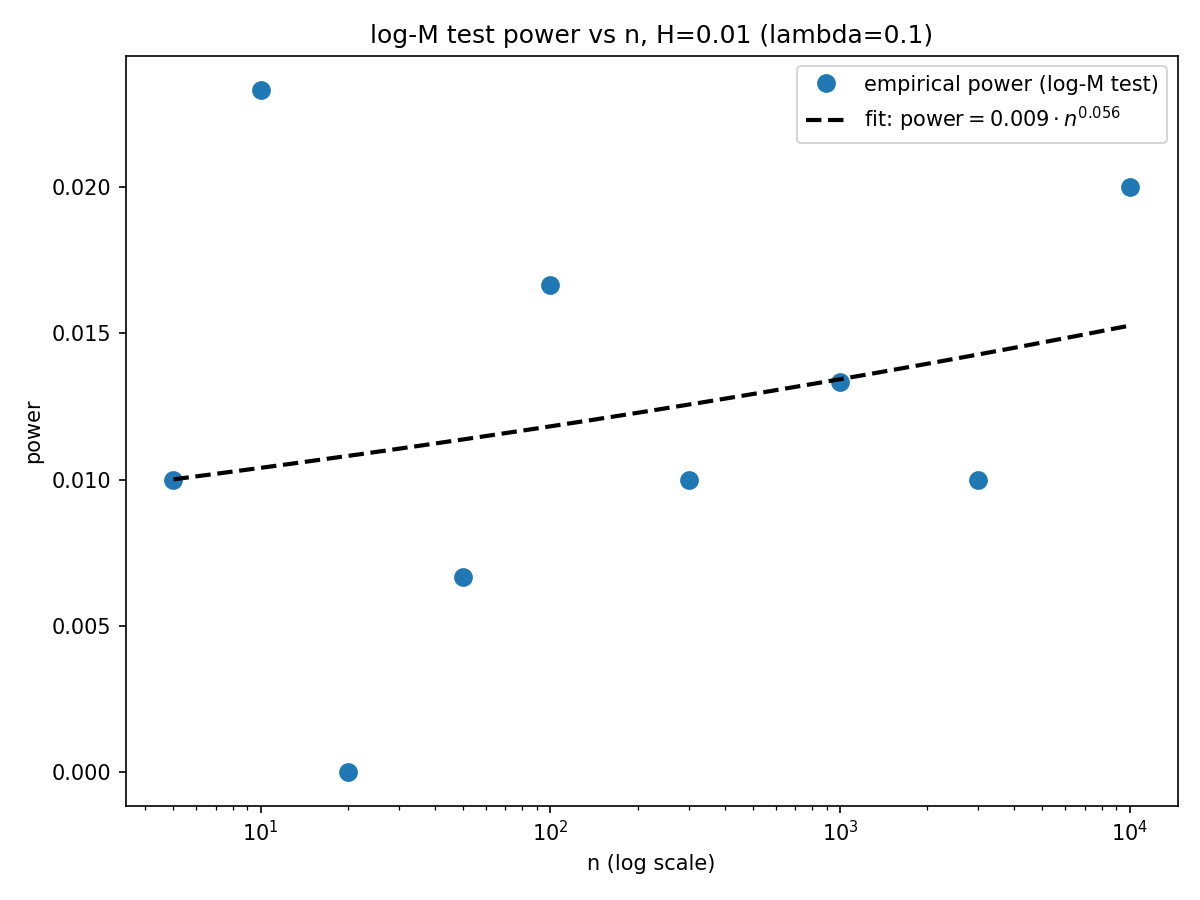}
        \caption{$H=0.01$.}
        \label{fig:power_H001}
    \end{subfigure}

    \vspace{0.8cm}

    \begin{subfigure}[b]{0.75\textwidth}
        \centering
        \includegraphics[width=\textwidth]{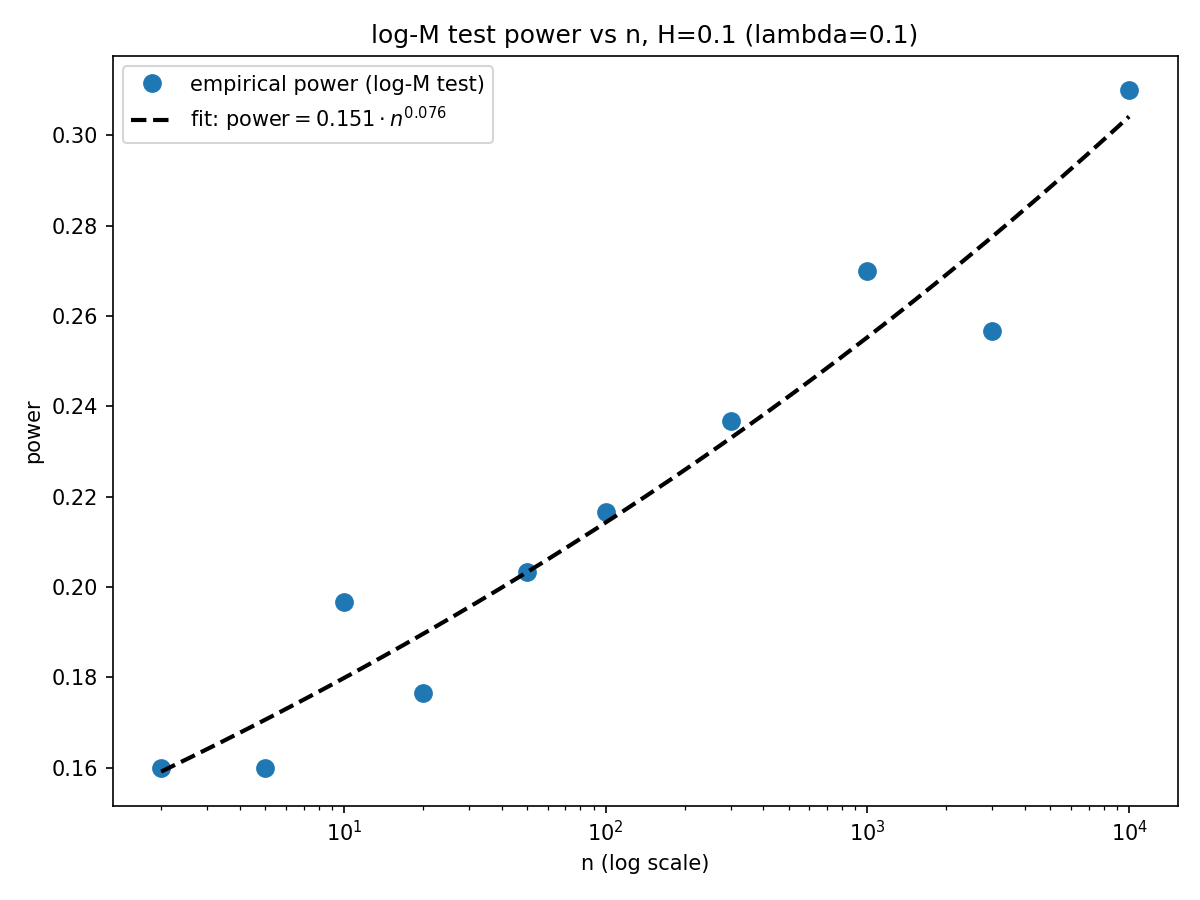}
        \caption{$H=0.10$.}
        \label{fig:power_H01}
    \end{subfigure}

    \caption{Empirical power as a function of the sample size $n$ (displayed on a logarithmic scale) for two values of the Hurst exponent. The dashed line corresponds to the least-squares fit of a power-law relationship between the empirical power and $n$. Panel (a) illustrates the near-null case $H=0.01$, where the power increases slowly with the sample size, whereas panel (b) corresponds to $H=0.10$, for which the convergence towards unit power is substantially faster.}
    \label{fig:power_vs_n}
\end{figure}

The experiments investigate how rapidly the power increases with the observation length and provide a quantitative assessment of the sample sizes required to distinguish these alternatives from the critical case.

\subsection{Scale invariance}
The behavior of return increment distributions across temporal aggregation scales plays a crucial role in understanding the empirical relevance of the model.

The distributional scale invariance of financial returns has been extensively documented in the multifractal volatility literature. Early works by Mandelbrot and subsequent developments introduced multiplicative cascade models to explain the observed self-similarity and heavy-tailed nature of asset returns across temporal aggregation scales \cite{mandelbrot1963variation, mandelbrot1997multifractal}. The multifractal random walk and multifractal random measure frameworks developed by Bacry, Delour and Muzy provided a continuous-time formulation in which return distributions preserve their shape across scales, up to a deterministic normalization, a property directly inherited from the underlying multiplicative cascade structure \cite{bacry2001multifractal, bacry2003log}. These models reproduce both empirical scaling laws of moments and the approximate invariance of full increment distributions, which are regarded as key stylized facts of financial time series.

Rough volatility models focus primarily on the irregularity of volatility sample paths and the short memory of volatility increments, as formalized through fractional Brownian motion with Hurst exponent smaller than one half \cite{gatheral2018volatility}. The numerical and empirical studies suggest that rough volatility models may exhibit approximate stability of return distributions over limited ranges of aggregation scales, particularly when combined with stochastic volatility dynamics \cite{bayer2016pricing, el2019microstructural}. These observations indicate that scale-invariant distributional features are not exclusive to multifractal models and motivate the investigation of intermediate frameworks, such as the Log S-fBM model, that bridge multifractal and rough volatility regimes.

The Log S-fBM model is designed to interpolate between these two regimes, raising the natural question of whether scale-invariant distributional features persist throughout the entire range of Hurst exponents. To address this question, we focus on the empirical distribution of increments of the Log S-fBM process at multiple aggregation scales. Rather than examining low-order moments alone, which may obscure important distributional features, we study the full increment density using nonparametric estimation techniques. This approach allows us to assess both the stability of distributional shapes across scales and the presence of non-Gaussian features such as heavy tails.

Beyond distributional properties, scale invariance in financial time series has been extensively investigated through the behavior of $q^{\text{th}}$ order moments of returns, volatility and log-volatility across temporal aggregation scales. In both multifractal and rough volatility frameworks, the approximate linear dependence of logarithmic moments on the logarithm of the scale is regarded as a key empirical signature of scale invariance and self-similar dynamics. Such moment-based analyses provide a complementary perspective to density-based approaches and play a central role in the empirical validation of stochastic volatility models.

In the multifractal literature, early studies on multiplicative cascade models and multifractal random walks established that the $q^{\text{th}}$ order moments of returns and volatility obey precise scaling laws over a wide range of scales, both in market data and in theoretical models \cite{muzy2000modelling,bacry2001modelling,bacry2003log, bacry2008log}. In \cite{muzy2000modelling}, Muzy\textit{ et al.} conducted similar investigations ending up showing scale invariance of $q^{th}$ order moments and log densities on both market data and the multifractal random walk model (see respectively figure 1,3 and 4 in \cite{muzy2000modelling}) These results highlighted the ability of multifractal models to reproduce scale invariance not only at the level of second-order moments but across an entire spectrum of orders $q$, reflecting the presence of intermittency and multiscaling effects.

Afterwards, the rough volatility literature has revisited moment scaling properties from a different perspective. As mentionned earlier, Gatheral \textit{et al.} showed that the logarithm of the $q^{\text{th}}$ order moments of realized volatility exhibits an approximately linear behavior as a function of the logarithmic scale, with a slope closely related to the Hurst parameter, for several highly liquid assets including DAX and Bund futures \cite{gatheral2018volatility}. Similar results were subsequently reported for major equity indices such as the S\&P~500 and NASDAQ. The scale invariance property was also observed using alternative volatility estimators, including the uncertainty zone estimator and realized variance (see Robert\textit{ et al.} \cite{Robert})\cite{Robert} and further discussed in related theoretical and empirical studies \cite{fukasawa2019volatility}. Comparable scaling behaviors were identified using range-based volatility estimators, such as the Garman--Klass and Parkinson estimators on equity indices where the scale invariance property was demonstrated by Mouti in \cite{mouti2023rough}, where those range based volatility estimator were investigated on S\&P100 and IBEX35 (see figures 3,4,5 and 6 in \cite{mouti2023rough}).

Motivated by these findings, we now assess the ability of the Log S-fBM model to reproduce numerically the scale invariance at the level of $q^{\text{th}}$ order moments. In particular, we numerically evaluate the goodness of fit between theoretical logarithmic moments derived under the small intermittency approximation and their empirical counterparts computed from simulated sample paths. This analysis provides a stringent benchmark for evaluating whether the Log S-fBM model preserves the moment-scaling properties observed in both multifractal and rough volatility frameworks, thereby offering further evidence of its interpretation as an interpolation between these two regimes.\\

The purpose of the following numerical experiment is therefore twofold. First, we aim to assess whether the increment path density of the Log S-fBM exhibits scale invariance across a wide range of Hurst exponents, from the multifractal regime to the rough volatility regime. Second, we examine whether the model generates fat-tailed increment distributions that are robust across aggregation scales, in line with well-documented stylized facts of financial return data (see as an example the empirical study \cite{Cont}).The results presented below provide empirical evidence that the Log S-fBM model preserves scale-invariant distributional features across a wide range of Hurst exponents. These findings complement the theoretical analysis and further support the interpretation of the Log S-fBM as a flexible and empirically meaningful bridge between multifractal and rough volatility models. In what follows, we display the obtained results in Figure \ref{fig:scaleinvariancedensities}.

\newpage

\newpage

\begin{center}
\begin{figure}[H]
\hspace*{-1cm}
\includegraphics[width=105mm]{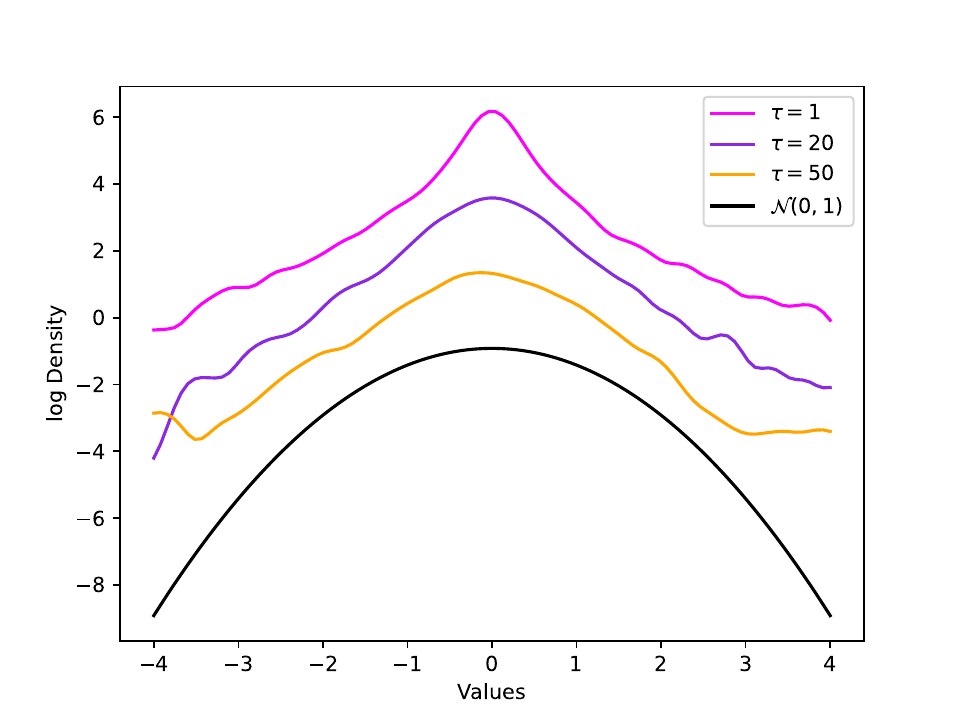}%
\hspace*{-1cm}
\includegraphics[width=105mm]{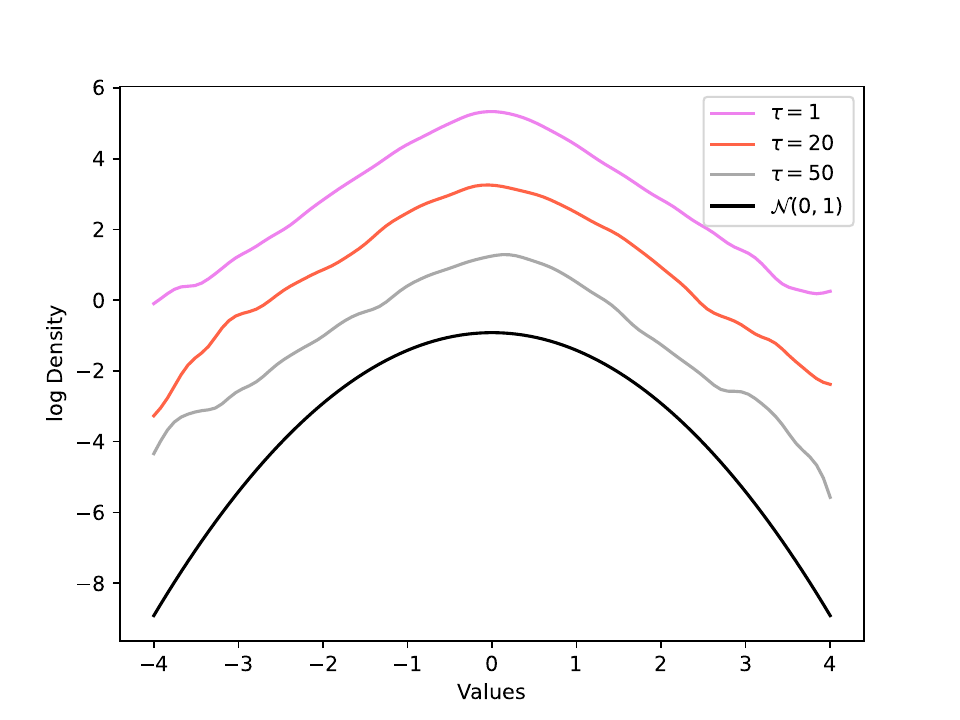}%
\begin{center}
(a) $H=0.005$ \hspace*{8cm} (b) $H=0.05$ 
\end{center}
\hspace*{-0.5cm}\\
\hspace*{3cm}\includegraphics[width=105mm]{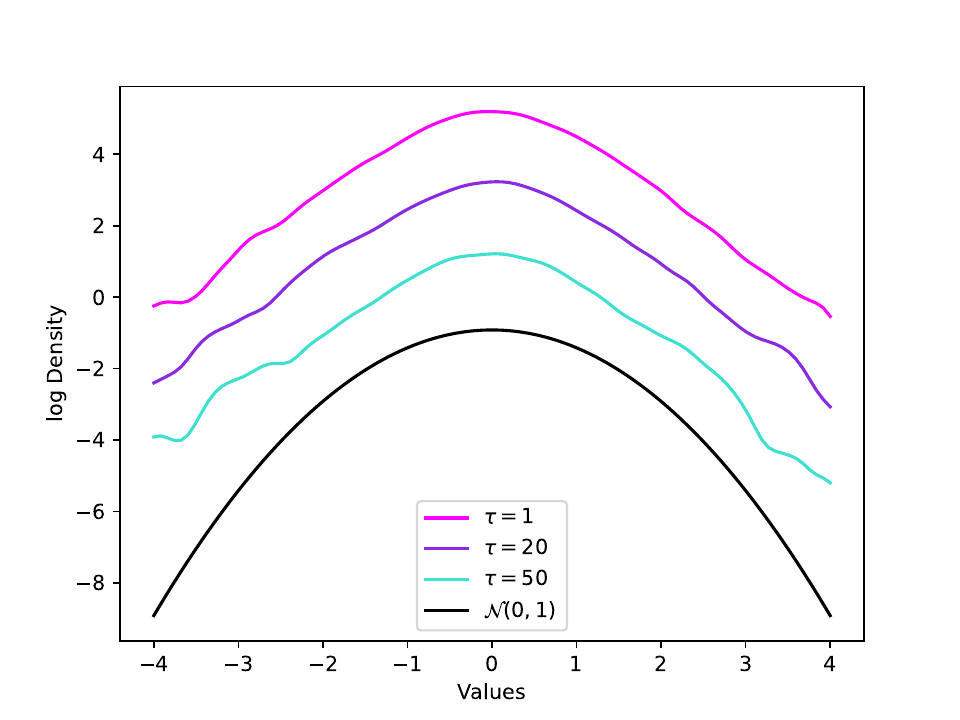}
\\[-0.7cm]

\begin{center}
 (c) $H=0.1$
\end{center}

\caption[Short caption for table of figures]{log kernel density estimator (using the gaussian kernel) of a sample path of the process $\left(\delta_{\tau}X_t\right)_{t\in \R}$ across multiple scales for different hurst indices and arbitrarily decayed from one another . (a) and (c) represents respectively the multifractal setting and the rough volatility setting. $L=2^{12}$, $T=2^{12}$, $\lambda^2=0.02$.}
\label{fig:scaleinvariancedensities}
\end{figure}
\end{center}
We notice that the scale invariance of the log density of the increments is conserved through a continuum of H values. Which gives an additional arguments in favour of the idea that the Log S-fBM model interpolates between the multifractal world and rough volatility world.\\\\
This experiments shows also the fat-tailedness of the return increments generated from a log S-fBM model, which is inline with the stylized fact on observed market data.\\\\

Finally, we numerically evaluate the goodness of fit of the theoretical log $q^{th}$ order moments in Eq.~\eqref{eq:empiricalqordermoments} as functions of the log scale in small intermittency (see in theorem \ref{theo:scaleinvariance}) with their empirical counterparts defined  considered as a benchmark:
\begin{eqnarray}
\begin{cases*}
     \widehat{m}^{M}_{H,T,\tau}(p):=\frac{1}{N}\sum_{k=1}^N\left|\delta_{k\tau}M_{H,T}(k)\right|^p\\
     \widehat{m}^{X}_{H,T,\tau}(p):=\frac{1}{N}\sum_{k=1}^N \left|\delta_{k\tau}X_k\right|^p\\
     \widehat{m}^{log}_{H,T,\tau}(p):=\frac{1}{N}\sum_{k=1}^N \left|\delta_{k\tau}\ln\left(M_{H,T}(k)\right)\right|^p\\
     \widehat{m}^{\omega}_{H,T,\tau}(p):=\frac{1}{N}\sum_{k=1}^N \left|\delta_{k\tau}\omega_{H,T}(k)\right|^p
\end{cases*}
\nonumber
\end{eqnarray}
where $N=\frac{L}{\tau}$ and $L$ the length of the sample path.\\\\
Herein, we are going to benchmark them with their small intermittency approximations counterparts. To do so, we proceeding similarly to the previous experiment, we fix $L$ as well as the the Log S-fBM hyper parameters $H$, $\lambda^2$ and $T$. Then, we compute the empirical log moments as functions of the log scale $\ln(\tau)$ and draw the small intermittency approximation counterparts. The results are in figure \ref{fig:logqmomentscale}. The plots is represents the $q^{th}$ order moments of respectively the log volatility, volatility, returns $X$ and S-fBM increments seen as functions of the log scale $\ln\left(\tau\right)$ for three values of $q$ together with their small intermittency counterparts (straight lines). Concerning the S-fBM $q^{th}$ order moments, we have their closed form formula from \cite{wu2022rough}:

\begin{equation}
\label{eq:qthorderSfbm}
\E\left(\left|\delta_{\tau}\omega_{H,T}(t)\right|^q\right)
= \nu^q\left(\frac{\tau}{T}\right)^{qH} 2^{\frac{q}{2}}\frac{\Gamma\left(\frac{q+1}{2}\right)}{\sqrt{\pi}}.
\end{equation}
The results are displayed in Figure \ref{fig:logqmomentscale} bellow.

\begin{center}
\begin{figure}[H]
\hspace*{-0.5cm}
\includegraphics[width=102mm]{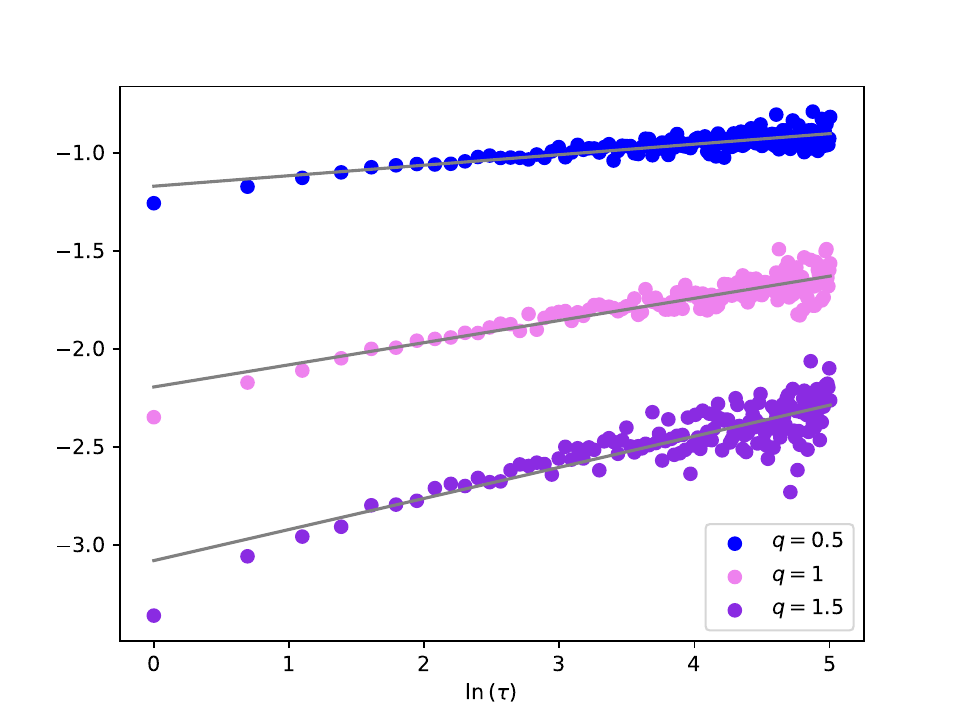}%
\hspace*{-0.5cm}
\includegraphics[width=102mm]{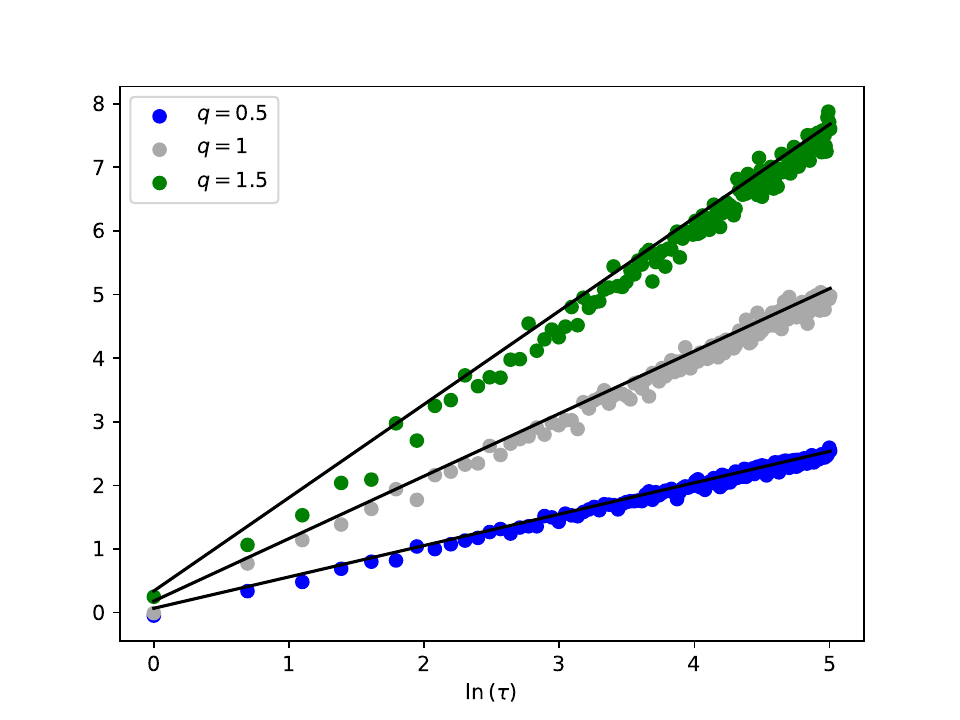}%
\begin{center}
(a) $ \ln\left(m^{log}_{H,T,\tau}\right)$ \hspace*{6cm} (b) $\ln\left(m^{M}_{H,T,\tau}\right)$ 
\end{center}
\hspace*{-0.5cm}
\includegraphics[width=102mm]{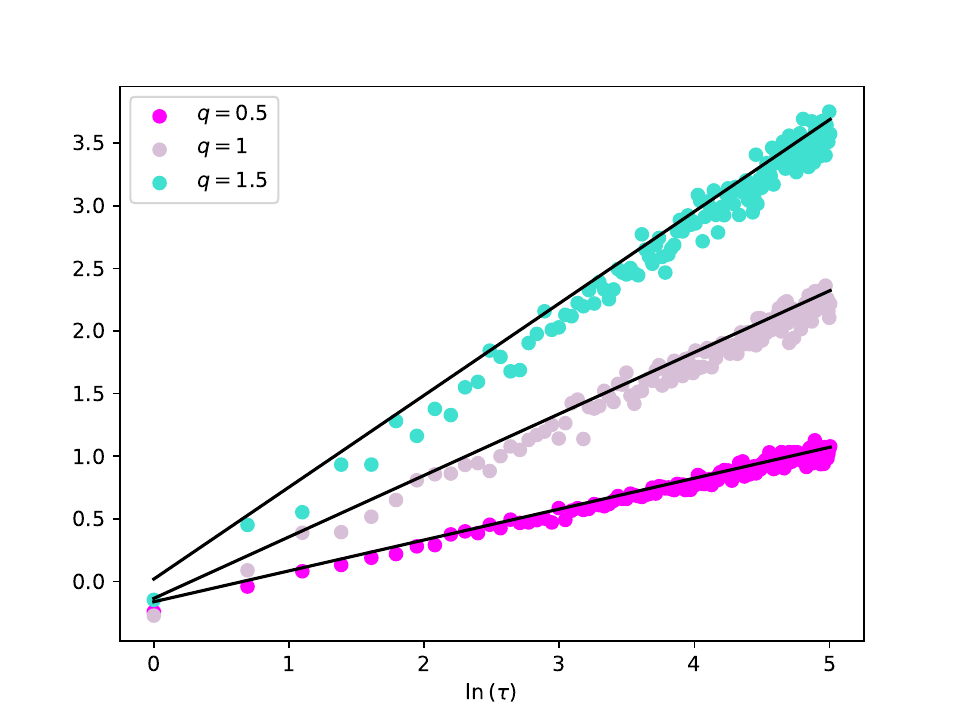}%
\hspace*{-0.5cm}
\includegraphics[width=102mm]{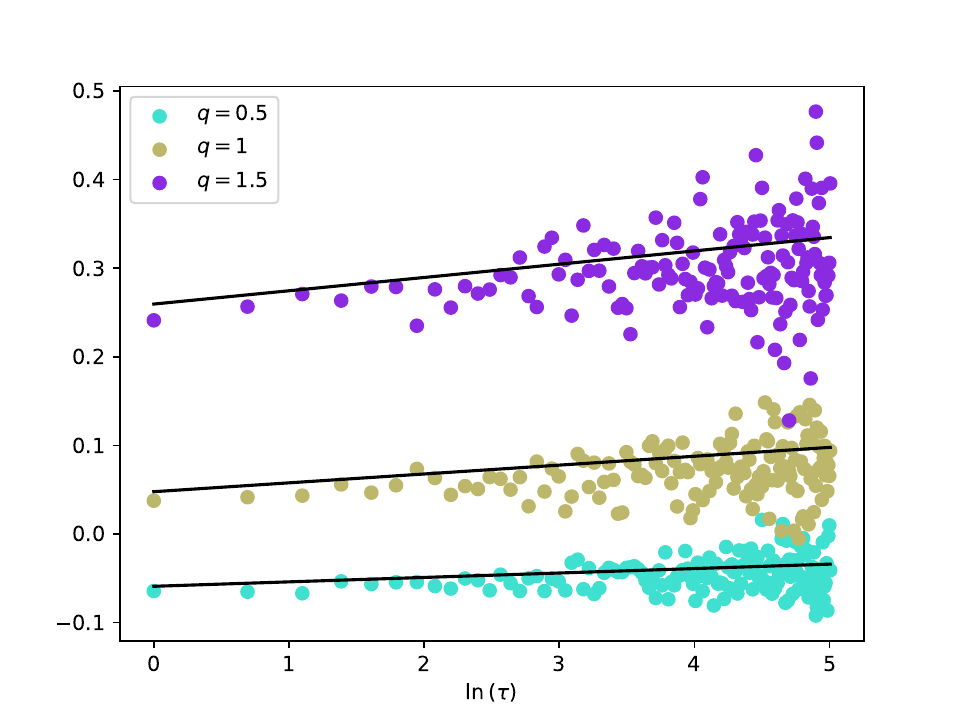}%
\\[-0.7cm]

\begin{center}
(c) $ \ln\left(m^{X}_{H,T,\tau}\right)$ \hspace*{6cm} (d) $\ln\left(m^{\omega}_{H,T,\tau}\right)$ 
\end{center}
\caption[Short caption for table of figures]{Scale invariance of the $q^{th}$ order moments represented by 150 dots in all the figures. For (b), (c), the black line represents the small intermittency closed formulas of theorem \ref{theo:scaleinvariance} as a function of $\ln\left(\tau\right)$. For (d) the black line represents the log $q^{th}$ order moments of Eq.~\eqref{eq:qthorderSfbm}. The grey line in (a) is the linear regression curve calibrated on the plotted data points. Its closed formula is already given in Eq.(28) in  \hypersetup{citecolor=blue}\cite{wu2022rough}. $L=2^{12}$, $T=2^{12}$, $H=0.01$ and $\lambda^2=0.02$.}
\label{fig:logqmomentscale}
\end{figure}
\end{center}
Indeed, the scale invariance properties of both Figure \ref{fig:scaleinvariancedensities} and \ref{fig:logqmomentscale} are similar to the ones showed in both rough volatility model and multifractal models. Furthermore, thanks to Proposition 2 in \cite{wu2022rough}, one can consider to some extent that the Log S-fBM model is an interpolation between the multifractal setting and the rough volatility setting.

\newpage
\section{Conclusion and prospects}

The results presented in this chapter constitute an extension and refinement of the theoretical framework underlying the Log S-fBM model introduced by Wu \textit{et al.} in \cite{wu2022rough}. Beyond recalling its original construction, we have provided a detailed analysis of the intrinsic properties of the S-fBM process, with a particular emphasis on its probabilistic structure and its implications from stochastic volatility modeling perspective.\\
From a distributional perspective, the S-fBM exhibits several distinctive features that place it at the crossroads between classical Brownian motion and fractional Brownian motion. While remaining a centered Gaussian process, its covariance structure induces long-range dependence and irregular sample paths that depart from the semimartingale framework. In this respect, the non-semimartingale nature of the S-fBM, established in this work, echoes the seminal result of Rogers \cite{Rogers} for fractional Brownian motion and confirms that S-fBM belongs to the broader class of rough Gaussian processes.\\ 
A central aspect of the S-fBM construction lies in its definition through the Takenaka time scale domain. This geometric viewpoint induces self similarity properties that propagate from the domain itself to the process and to the associated volatility measure. We have shown that these self similar relations provide a coherent interpretation of how temporal rescaling can be compensated by adjustments of the decorrelation horizon, thereby yielding  invariance properties. Such results highlight the structural consistency of the Log S-fBM model across time scales, a feature that is crucial for applications in finance where volatility dynamics exhibit strong multi scale behavior.\\
Another concern was the systematic exploitation of the Gaussian nature of the S-fBM together with the uniform boundedness of its sample paths to derive deviation inequalities with explicit and non-asymptotic upper bounds. These inequalities quantify the probability of extreme fluctuations of the S-fBM process and provide precise control in terms of the model parameters, namely the Hurst exponent, the intermittency coefficient and the decorrelation limit. From a practical standpoint, these results offer rigorous probabilistic guarantees and understanding the tail behavior of the S-fBM process emphasizing on the tail behavior in extreme parameter configuration such as the super rough regime.\\
Finally, we have investigated the scale invariance properties of the Log S-fBM model in the small intermittency regime. By leveraging explicit asymptotic formulas, we have shown that the model reproduces scaling behaviors that are consistent with both rough and multifractal volatility paradigms. These theoretical findings are in agreement with empirical observations, as illustrated by the numerical experiments and provide a unifying perspective that bridges multifractal random measures and rough volatility models within a single framework, as advocated by Wu \textit{et al.} in \cite{wu2022rough}. Altogether, the results paper contribute to a better understanding of the mathematical foundations of the Log S-fBM model and reinforce its relevance as a flexible and theoretically grounded tool for modeling volatility across multiple time scales.

\vfill
\textbf{Acknowledgment:} This work has received support from the French government, managed by the National Research Agency (ANR), under the "France 2030" program with reference "ANR-23-IACL-0008.

\newpage

\newpage

\begin{appendices}

\section{Definitions and preliminary results}
\label{app:definitions}
First, here is the expectation of a log $\chi^2(1)$ random variable, a result from \cite{Pav}.
\begin{lem}
\label{lem:expectationlogxi2}
    Given $X=\ln(G^2)$ where $G \sim \mathcal{N}(0,1)$, we have:
    \begin{eqnarray}
        \E\left(X\right)=\ln\left(2\Gamma\left(\frac{1}{2}\right)\right)
    \end{eqnarray}
    Where $\Gamma(x) = \int_0^\infty t^{x-1}e^{-t} \, dt,x>0$.
    \end{lem}
This means that for any $\alpha>0$ with $X_{\alpha}:=\ln(|G|^{\alpha})$, 
\begin{eqnarray}
\label{eq:consequencelemmaexpectationlogxi2}
    \E\left(X_{\alpha}\right)=\frac{\alpha}{2}\ln\left(2\Gamma\left(\frac{1}{2}\right)\right)
\end{eqnarray}

We recall the definition of a slowly decaying function.
\begin{mydeff}
\label{def:slowlydecaying}
Consider $I\subset \R$.\\
   A function $f:I\mapsto \R$ is slowly decaying if for any $t\in I$:
   \begin{eqnarray}
       \frac{f(tx)}{f(x)}\underset{x\rightarrow \infty}\longrightarrow 1
   \end{eqnarray}

\end{mydeff}

We expose useful deviation inequalities for the gaussian random variable.
\begin{lem}
\label{lem:tailgaussianupperbound}
    If $X\sim \mathcal{N}(0,\sigma^2)$, then :
    \begin{equation}
        \forall t>0,\tab \Prob\left(|X|>t \right)\leq 2e^{-\frac{t^2}{2\sigma^2}}
    \end{equation}
\end{lem}
\textit{\textbf{Proof}}\\
    Given a $t>0$, we have using Markov's inequality:
    \begin{eqnarray}
        \forall s>0,\Prob\left(X>t \right) \leq \E\left(e^{sX}\right)e^{-st}=e^{\frac{\sigma^2s^2}{2}-st}\nonumber
    \end{eqnarray}
    After straightforward computations, we have:
    \begin{eqnarray}
        \underset{s>0}{\inf}\left\{\frac{\sigma^2s^2}{2}-st \right\}=-\frac{t^2}{2\sigma^2}\nonumber
    \end{eqnarray}
    Which leads to:
    \begin{eqnarray}
        \Prob\left(|X|>t \right)\leq e^{-\frac{t^2}{2\sigma^2}} \nonumber
    \end{eqnarray}
    On the other hand, as $X\sim \mathcal{N}(0,\sigma^2)$:
    \begin{eqnarray}
        \Prob\left(X<-t \right) =\Prob\left(-X>t \right) = \Prob\left(X>t \right) \nonumber
    \end{eqnarray}
    We conclude then that:
    \begin{eqnarray}
        \Prob\left(|X|>t \right)=\Prob\left(X>t \right)+\Prob\left(X<-t \right)\leq 2e^{-\frac{t^2}{2\sigma^2}} \nonumber
    \end{eqnarray}
    \\
\tab\tab\tab\tab\tab\tab\tab\tab\tab\tab\tab\tab\tab\tab\tab\tab\tab\tab\tab\tab\tab\tab\tab\tab\tab\tab\tab\tab\tab\tab\tab\tab\tab\tab $\blacksquare$

\begin{lem}
\label{lem:upperboundmomentgaussian}
    Given $X\sim \mathcal{N}\left(0,\sigma^2\right)$, we have:
    \begin{equation}
        \forall k\in \N,\tab  \E\left(|X|^k\right)\leq \left(2\sigma^2\right)^{\frac{k}{2}} \Gamma\left(\frac{k}{2} \right) 
    \end{equation}
    Where $\Gamma\left(x\right) :=\int_0^{+\infty}e^{-u}u^{x-1}du$, $x>0 $
\end{lem}
\textit{\textbf{Proof}}\\
Given $k\in \N$, we have:
\begin{eqnarray}
    \E\left(|X|^k\right) = \int_0^{+\infty}\Prob\left(|X|^k>t\right) dt=\int_0^{+\infty}\Prob\left(|X|>t^{\frac{1}{k}}\right) dt \nonumber
\end{eqnarray}
Using lemma \ref{lem:tailgaussianupperbound}, we obtain:
\begin{eqnarray}
    \E\left(|X|^k\right) \leq 2\int_0^{+\infty}e^{\frac{-t^{\frac{2}{k}}}{2\sigma^2}}dt \nonumber
\end{eqnarray}
By the change of variable $u:=\frac{t^{\frac{2}{k}}}{2\sigma^2}$, one has:
\begin{eqnarray}
    \E\left(|X|^k\right) \leq \left(2\sigma^2\right)^{\frac{k}{2}}k\int_0^{+\infty}e^{-u}u^{\frac{k}{2}-1}du \nonumber 
\end{eqnarray}
Which means that:
\begin{eqnarray}
    \E\left(|X|^k\right) \leq \left(2\sigma^2\right)^{\frac{k}{2}}k\Gamma\left(\frac{k}{2}\right) \nonumber
\end{eqnarray}
\\
\tab\tab\tab\tab\tab\tab\tab\tab\tab\tab\tab\tab\tab\tab\tab\tab\tab\tab\tab\tab\tab\tab\tab\tab\tab\tab\tab\tab\tab\tab\tab\tab\tab\tab $\blacksquare$

Last but not least, we recall the \textit{Borel-TIS} theorem.
\begin{theo}(\textit{Borel-TIS})\\
\label{theo:boreltis}
    We consider $\mathcal{T}$ a topological space and a gaussian field $\left(X(x)\right)_{x\in \mathcal{T}}$ that is continuous almost surely and uniformly bounded in $\mathcal{L}^2$.
\\
Then, it holds that:
\begin{eqnarray}
    \forall t\geq 0,\Prob\left(\left|\sup_{x\in \mathcal{T}}X(x)-\E\left(\sup_{x\in \mathcal{T}}X(x)\right)> t  \right|\leq  \right)\leq 2e^{-\frac{t^2}{2\sigma_\mathcal{T}^2}}
\end{eqnarray}
Where $\sigma_\mathcal{T}^2=\sup_{x\in \mathcal{T}}\E\left(\left |X(x)\right|^2\right)$.
\end{theo}
    
\section{The S-fBM process: Miscellaneous results}
\label{app:appendixS-fBM}
We commence this appendix by computing the autocovariance function of $\Omega_{H,T,\Delta}$ for any arbitrary $\Delta>0$.
\begin{prop}
\label{prop:PropositioncovOmega}
   The process 
   $\Omega_{H,T,\Delta}$ is a stationary centered gaussian process, with the following covariance structure:
  
    \begin{eqnarray}
   \label{eq:autocovOmegaDelta}
    \forall (t,s)\in\R_{+}^2, \quad 
    \text{Cov}\left(\Omega_{H,T,\Delta}(t),\Omega_{H,T,\Delta}(s) \right)=\frac{\Delta^2}{2H(1-2H)}-\Delta^2\left(\frac{\tau}{T}\right)^{2H}g\left(\frac{\Delta}{\tau}\right),
    \end{eqnarray}
   where $\tau=|t-s|$ and $g\left(z\right):=\frac{|1+z|^{2H+2}+|1-z|^{2H+2}-2}{2Hz^2(1-(2H)^2)(2H+2)}$.
\end{prop}
\textit{\textbf{Proof}}\\
We fix $\Delta>0$.\\
By definition, $ \left(\omega_{H,T}(t)-\mu_{H}\right)_t$ is a centered gaussian process. Consequently, $\Omega_{H,T,\Delta}$ is a centered gaussian process too.\\\\
We suppose that $t>s$ as they play symmetric roles. One has:
    \begin{eqnarray}
        \text{Cov}\left(\Omega_{H,T,\Delta}(t),\Omega_{H,T,\Delta}(s) \right)= \frac{1}{\lambda^2}\int_t^{t+\Delta}\int_s^{s+\Delta}\text{Cov}\left(\omega_{H,T,\Delta}(t),\omega_{H,T,\Delta}(s) \right)dudv \nonumber
    \end{eqnarray}
    By definition, we have :
    \begin{eqnarray}
        \frac{1}{\lambda^2}\int_t^{t+\Delta}\int_s^{s+\Delta}\text{Cov}\left(\omega_{H,T,\Delta}(t),\omega_{H,T,\Delta}(s) \right)dudv= \frac{1}{2H(1-2H)}\int_t^{t+\Delta}\int_s^{s+\Delta}\left(1-\left(\frac{|u-v|}{T}\right)^{2H}dudv\right) \nonumber
    \end{eqnarray}
    After some algebra, one has:
    \begin{eqnarray}
        \frac{1}{\lambda^2}\int_t^{t+\Delta}\int_s^{s+\Delta}\text{Cov}\left(\omega_{H,T,\Delta}(t),\omega_{H,T,\Delta}(s) \right)dudv=\frac{\Delta^2}{2H(1-2H)}-\frac{|\tau+\Delta|^{2H+2}- 
        \tau^{2H+2}}{2HT^{2H}(1-(2H)^2)(2H+2)}\nonumber \\
        -\frac{1}{2HT^{2H}(1-(2H)^2)}\left(\int_t^{s+\Delta}|u-s-\Delta|^{2H}(u-s-\Delta) du+\int_{s+\Delta}^{t+\Delta}|u-s-\Delta|^{2H}(u-s-\Delta) du\right) \nonumber
    \end{eqnarray}
    \\
    By integrating the remaining term, we obtain:
    \begin{eqnarray}
        \frac{1}{\lambda^2}\int_t^{t+\Delta}\int_s^{s+\Delta}\text{Cov}\left(\omega_{H,T,\Delta}(t),\omega_{H,T,\Delta}(s) \right)dudv=\frac{\Delta^2}{2H(1-2H)}-
        \frac{|\tau+\Delta|^{2H+2}-\tau^{2H+2}+|\tau-\Delta|^{2H+2}-\tau^{2H+2}}{2HT^{2H}(1-(2H)^2)(2H+2)}\nonumber
    \end{eqnarray}
    Which leads to:
    \begin{eqnarray}
        \frac{1}{\lambda^2}\int_t^{t+\Delta}\int_s^{s+\Delta}\text{Cov}\left(\omega_{H,T,\Delta}(t),\omega_{H,T,\Delta}(s) \right)dudv=\frac{\Delta^2}{2H(1-2H)}-
        \frac{|\tau+\Delta|^{2H+2}-2\tau^{2H+2}+|\tau-\Delta|^{2H+2}}{2HT^{2H}(1-(2H)^2)(2H+2)}\nonumber
    \end{eqnarray}
   
Thus, $\Omega_{H,T,\Delta}(.)$ is a stationary gaussian process with the claimed autocovariance structure. \\
    \tab\tab\tab\tab\tab\tab\tab\tab\tab\tab\tab\tab\tab\tab\tab\tab\tab\tab\tab\tab\tab\tab\tab\tab\tab\tab\tab\tab\tab\tab\tab\tab\tab\tab $\blacksquare$

We follow up by recalling Maruyama's theorem.\\
\begin{theo}
\label{theo:maruyama}
(Maruyama)\\
    Let $X$ be a centered stationary Gaussian process valued on $\R^d$ with a continuous covariance kernel $\gamma$ and spectral measure $\mu$. Then $X$ is ergodic if and only if $\mu$ has no atom.
\end{theo}

In fact, the S-fBM process admits a spectral density in closed form.

\begin{prop}
\label{prop:spectraldensitysfbm}
     Given $H\in \left]0,\frac{1}{2}\right[$ and $T>0$.\\
   The spectral density of $\omega_{H,T}$ exists and is expressed as follows:
   \begin{eqnarray}
    \forall u\in \R,\tab  f_{H,T}(u)=\frac{\lambda^2}{\pi H(1-2H)}\left( \frac{sin(uT)}{u}-T\underset{\substack{k\geq 0 \\k\hspace{0.1cm} even}}\sum\frac{\left(Tu\right)^k\left(-1\right)^{\lfloor k/2 \rfloor}}{\left(2H+k+1\right)k!}\right)
\end{eqnarray}
\end{prop}
\textit{\textbf{Proof}}\\
    $\omega_{H,T}$ is weakly stationary and its autocovariance function is in $\mathcal{L}^1([-T,T])$. Thus, it admits a spectral density denoted $f_{H,T}$.\\\\
    We have:
    \begin{eqnarray}
        f_{H,T}(u)=\frac{1}{2\pi}\int_{\R}C_{\omega_{H,T}}(h)e^{-ihu}dh
        \nonumber
    \end{eqnarray}
Using the expression of $C_{\omega_{H,T}}(.)$, we end up with:
\begin{eqnarray}
        f_{H,T}(u)=\frac{\lambda^2}{2\pi H(1-2H)}\left( 2\frac{sin(uT)}{u}-\int_{-T}^T\left(\frac{|h|}{T}\right)^{2H} e^{-ihu}dh\right)
        \nonumber
\end{eqnarray}
Furthermore, we have by straightforward computations:
$$\overline{\int_{-T}^T|h|^{2H} e^{-ihu}dh}=-\int_{-T}^T|h|^{2H} e^{-ihu}(-dh)=\int_{-T}^T|h|^{2H} e^{-ihu}dh$$
Which means that the previous function is $\R$ valued.\\
Besides, 
\begin{eqnarray}
    \int_{-T}^T|h|^{2H} e^{-ihu}dh=2\int_{0}^T h^{2H} cos\left(hu\right)dh=2\mathcal{R}\left(\int_{0}^T h^{2H} e^{ihu}dh \right)\nonumber
\end{eqnarray}
Using the series expansion, one has:
\begin{eqnarray}
    \int_{0}^T h^{2H} e^{ihu}dh =\sum_{k=0}^{\infty}\frac{\left(iu\right)^k}{k!}\int_0^Th^{2H+k} dh= T^{2H+1}\sum_{k=0}^{\infty}\frac{\left(iTu\right)^k}{\left(2H+k+1\right)k!}\nonumber
\end{eqnarray}
After some algebra, one has:
\begin{eqnarray}
    \int_{-T}^T|h|^{2H} e^{-ihu}dh=2T^{2H+1}\sum_{k=0}^{\infty}\frac{\left(Tu\right)^kcos\left(\frac{k\pi}{2} \right)}{\left(2H+k+1\right)k!}\nonumber
\end{eqnarray}
Thus, 
\begin{eqnarray}
    \int_{-T}^T|h|^{2H} e^{-ihu}dh=2T^{2H+1}\underset{\substack{k\geq 0 \\k\hspace{0.1cm} even}}\sum\frac{\left(Tu\right)^k\left(-1\right)^{\lfloor k/2 \rfloor}}{\left(2H+k+1\right)k!}\nonumber
\end{eqnarray}
Which gives the claimed result.

\tab\tab\tab\tab\tab\tab\tab\tab\tab\tab\tab\tab\tab\tab\tab\tab\tab\tab\tab\tab\tab\tab\tab\tab\tab\tab\tab\tab\tab\tab\tab\tab\tab\tab $\blacksquare$
    
\section{The S-fBM process: Self similarity properties and  non semimartingality proofs}
\label{app:Scaling properties, non semimartingality and deviation inequality proofs}
\subsection{Self similarity properties}
\label{app:Scaling properties}
We begin by the Takenaka self similarity property (Proposition \ref{prop:takenakaselfsimilarity}) proof.\\\\
\textit{\textbf{Proof}}\\
\label{proof:ScalingpropertiesTakenaka}
    For arbitrary $\alpha>0$, one has:
\begin{eqnarray}
    C_T\left(\alpha s\right)=\left\{(t,h)/h>0,|t-\alpha s|<\frac{\min\left(h,\frac{T}{2}\right)}{2}  \right\} \nonumber
\end{eqnarray}
Which means that:
\begin{eqnarray}
    C_T\left(\alpha s\right)=\left\{(t,h)/h>0, s-\frac{\min\left(h,\frac{T}{2}\right)}{2\alpha}<t< s+\frac{\min\left(h,\frac{T}{2}\right)}{2\alpha}  \right\}\nonumber
\end{eqnarray}
After some algebra, we get:
\begin{eqnarray}
    C_T\left(\alpha s\right)=\left\{(t,h)/h>0, |t-s|<\frac{\min\left(h,\frac{T}{2\alpha}\right)}{2} \right\}\nonumber
\end{eqnarray}
Meaning that:
\begin{eqnarray}
 C_T\left(\alpha s\right)=\left\{(t,h)/h>0, s-\frac{\min\left(h,\frac{T}{2\alpha}\right)}{2}<t< s+\frac{\min\left(h,\frac{T}{2\alpha}\right)}{2} \right\}\nonumber
\end{eqnarray}
As a result, $C_T\left(\alpha s\right)=C_{\frac{T}{\alpha}}\left( s\right)$ leading to the claimed result.\\\\

     Given $T_1<T_2$ and $t\in[0,T_1]$, we have that:
     
\begin{align*}
\min\left(h,\dfrac{T_1}{2}\right)&=\min\left(h,\dfrac{T_2}{2}\right)=h, &&\text{if } h\in[0,T_1],\\
\min\left(h,\dfrac{T_1}{2}\right)&=T_1,\quad \min\left(h,\dfrac{T_2}{2}\right)=h, &&\text{if } h\in[T_1,T_2],\\
\min\left(h,\dfrac{T_1}{2}\right)&=T_1,\quad \min\left(h,\dfrac{T_2}{2}\right)=T_2, &&\text{if } h>T_2.
\end{align*}
     Geometrically we have:
     \begin{center}
\begin{figure}[H]
     \centering
     \includegraphics[width=6.8cm ,height=6.8cm]{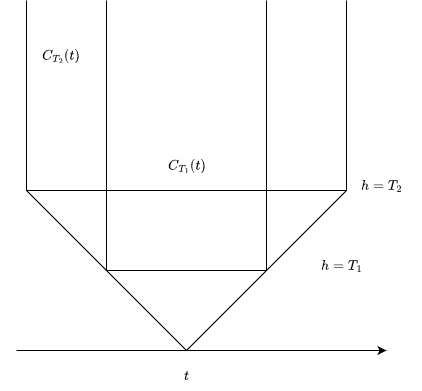}
     \label{fig:inclusiontakenaka}
     \caption{Inclusion representation of the Takenaka domain.}
 \end{figure} 
\end{center}
Which gives the inclusion $C_{T_1}(t) \subset C_{T_2}(t) $.
    \\
\tab\tab\tab\tab\tab\tab\tab\tab\tab\tab\tab\tab\tab\tab\tab\tab\tab\tab\tab\tab\tab\tab\tab\tab\tab\tab\tab\tab\tab\tab\tab\tab\tab\tab $\blacksquare$\\\\
Here comes the proof of Proposition \ref{prop:propertislawsfbm}.\\
\textit{\textbf{Proof}}\\
\begin{enumerate}

    \item 
     Given $(H,T)\in \left]0,\frac{1}{2}\right[\times \R_{+}$.\\
    $\omega_{H,T}$ is a gaussian process. As a result, $\omega_{H,T}(t)-\omega_{H,T}(s)$ is a gaussian random variable for any $(t,s)\in\R_+^2$. Moreover, $\omega_{H,T}$ is stationary. Thus, $\omega_{H,T}(t)-\omega_{H,T}(s)$ is centered. \\
    We have for any $\alpha \in \R$:
    \begin{eqnarray}
        \E\left(e^{i\alpha \left(\omega_{H,T}(t)-\omega_{H,T}(s) \right)} \right)=e^{\frac{-\alpha^2Var\left(\omega_{H,T}(t)-\omega_{H,T}(s)\right)}{2}}
        \nonumber
    \end{eqnarray}
By direct computations, we have:
\begin{eqnarray}
    Var\left(\omega_{H,T}(t)-\omega_{H,T}(s)\right) = \frac{\lambda^2}{H(1-2H)}\left(\frac{|t-s|}{T}\right)^{2H} \nonumber
\end{eqnarray}
Which gives:
\begin{eqnarray}
    \E\left(e^{i\alpha \left(\omega_{H,T}(t)-\omega_{H,T}(s) \right)} \right)=e^{\frac{-\alpha^2}{2}\frac{\lambda^2}{H(1-2H)}\left(\frac{|t-s|}{T}\right)^{2H} }
        \nonumber
\end{eqnarray}
This leads to:
\begin{eqnarray}
    \E\left(e^{i\alpha \left(\omega_{H,T}(t)-\omega_{H,T}(s) \right)} \right)=\E\left(e^{i\alpha \left(\frac{\lambda}{\sqrt{H(1-2H)}}B_{\left(\frac{|t-s|}{T}\right)^{2H}} \right)} \right)
        \nonumber
\end{eqnarray}
\item One has for any $(s,t)\in \R_{+}^2$:
\begin{eqnarray}
   \text{cov}(\omega_{H,T}(x+rt), \omega_{H,T}(x+rs)) =\frac{\nu^2}{2}\left(1-\left(\frac{r}{T}\right)^{2H}\left|t-s\right|^{2H}\right)\nonumber
\end{eqnarray}
After some algebra, one has:
\begin{eqnarray}
   \text{cov}(\omega_{H,T}(x+rt), \omega_{H,T}(x+rs)) =\frac{r^{2H}\nu^2}{2}\left(1-\left(\frac{\left|t-s\right|}{T}\right)^{2H}\right)+C(r,H)\nonumber
\end{eqnarray}
where $C(r,H)=\frac{\nu^2}{2}\left(1-r^{2H} \right)$.\\
\end{enumerate}

Which concludes the proof.\\
\tab\tab\tab\tab\tab\tab\tab\tab\tab\tab\tab\tab\tab\tab\tab\tab\tab\tab\tab\tab\tab\tab\tab\tab\tab\tab\tab\tab\tab\tab\tab\tab\tab\tab $\blacksquare$\\\\
\subsection{Non semimartingality}
\label{app:sfbmnotsemimartingale}
We follow up by exposing the proof of Proposition  \ref{prop:sfbmnotsemimartingale}.\\\\
\textit{\textbf{Proof}}\\
\label{proof:sfbmnotsemimartingale}
Given $n\in \N$.\\
    We consider the following dyadic subdivision $t_k:=\frac{k}{2^n},k\in \llbracket 1, 2^{n} \rrbracket$ of $[0,1]$ and denote the discrete time increment process $\left(Y_k:=\omega_{H,T}(t_k)-\omega_{H,T}(t_{k-1}),k\in\llbracket 1, 2^{n} \rrbracket \right)$.\\\\
    Given that $\omega_{H,T}$ is a gaussian process, $Y$ is a dicrete time gaussian process and centered ($\omega_{H,T}$ is weakly stationary). Moreover, we have for any $(k,k^{'})\in \llbracket 1, 2^{n} \rrbracket^2$ and $m\in \llbracket 1, 2^{n} \rrbracket$:
    \begin{eqnarray}
        \text{Cov}\left(Y_{k+m},Y_{k^{'}+m} \right) =  \text{Cov}\left(\omega_{H,T}(t_{k+m}),\omega_{H,T}(t_{k^{'}+m}) \right)- \text{Cov}\left(\omega_{H,T}(t_{k+m}),\omega_{H,T}(t_{k^{'}+m-1}) \right)- \nonumber\\ \text{Cov}\left(\omega_{H,T}(t_{k+m-1}),\omega_{H,T}(t_{k^{'}+m}) \right)+ \text{Cov}\left(\omega_{H,T}(t_{k+m-1}),\omega_{H,T}(t_{k^{'}+m-1}) \right) \nonumber
    \end{eqnarray}
As $\omega_{H,T}$ is stationary, we obtain:
\begin{eqnarray}
        \text{Cov}\left(Y_{k+m},Y_{k^{'}+m} \right) =  \text{Cov}\left(\omega_{H,T}(t_{k}),\omega_{H,T}(t_{k^{'}}) \right)- \text{Cov}\left(\omega_{H,T}(t_{k}),\omega_{H,T}(t_{k^{'}-1}) \right)- \nonumber\\ \text{Cov}\left(\omega_{H,T}(t_{k-1}),\omega_{H,T}(t_{k^{'}}) \right)+ \text{Cov}\left(\omega_{H,T}(t_{k-1}),\omega_{H,T}(t_{k^{'}-1}) \right) \nonumber
    \end{eqnarray}
Thus, 
\begin{eqnarray}
        \text{Cov}\left(Y_{k+m},Y_{k^{'}+m} \right) =  \text{Cov}\left(Y_{k},Y_{k^{'}} \right)\nonumber
\end{eqnarray}
Meaning that $Y$ is covariance stationary as well.\\
As a result, we have for any tuple $(p_j)_{j\in \llbracket 1, 2^{n} \rrbracket} \in \R^n$ that:
\begin{eqnarray}
        \E\left(e^{\sum_{i=1}^{2^{n}}p_iY_{i+m}} \right)= \E\left(e^{\sum_{i=1}^{2^{n}}p_iY_{i}} \right)
        \nonumber
\end{eqnarray}
Consequently, $Y$ is stationary.\\\\
On the other hand, as the spectral density of $\omega_{H,T}$ exists (see Proposition \ref{prop:spectraldensitysfbm}), then its spectral measure has no atoms. Using the previous computations, we conclude the same for $Y$. According to \textit{Maruyama}'s theorem (theorem \ref{theo:maruyama}), $Y$ is ergodic. Then, the ergodic theorem leads to:
\begin{eqnarray}
    \frac{1}{2^n}\sum_{k=1}^{2^n}|Y_k|^p \underset{n \rightarrow \infty}\longrightarrow \E\left(|Y_1|^p \right),\tab a.s,\mathcal{L}^1
    \nonumber
\end{eqnarray}
$Y_1$ is a centered gaussian random variable with variance is $\frac{\lambda^2}{2^{2nH}T^{2H}H(1-2H)}$.\\
Using Eq.(18) in \cite{Winkelbauer} we have:
\begin{eqnarray}
    \E\left(|Y_1|^p \right)=\frac{\lambda^p}{2^{pnH}T^{2H}\left(H(1-2H)\right)^{\frac{p}{2}}}2^{\frac{p}{2}}\frac{\Gamma\left(\frac{p+1}{2}\right)}{\sqrt{\pi}} \nonumber
\end{eqnarray}
Where $\Gamma(x) = \int_0^\infty t^{x-1}e^{-t} \, dt,x>0$.\\\\
Consequently, we obtain:
\begin{eqnarray}
    \sum_{k=1}^{2^n}|Y_k|^p \underset{n \rightarrow \infty}\sim C_p 2^{n\left(1-pH \right)},\tab a.s \nonumber
\end{eqnarray}
Where $C_p:=\frac{\left(\sqrt{2}\lambda\right)^p}{T^{2H}\left(H(1-2H)\right)^{\frac{p}{2}}}\frac{\Gamma\left(\frac{p+1}{2}\right)}{\sqrt{\pi}}$.\\
As a result, we have while $H\neq \frac{1}{2}$ using the same arguments of section 2 in \cite{Rogers}:
\begin{eqnarray}
    \lim_{{n \to \infty}} \sum_{k=1}^{2^n}|Y_k|^p \underset{n \rightarrow \infty} \in \{0,\infty \} \nonumber
\end{eqnarray}
Which concludes the proof.\\
\tab\tab\tab\tab\tab\tab\tab\tab\tab\tab\tab\tab\tab\tab\tab\tab\tab\tab\tab\tab\tab\tab\tab\tab\tab\tab\tab\tab\tab\tab\tab\tab\tab\tab $\blacksquare$

\section{Deviation inequalities}
\subsection{Preliminaries}
We start by recalling the definition of respectively the covering and packing numbers, useful for the coming proofs.\\

In the following, we consider $D$ and arbitrary subset of $\R_+$.
\begin{mydeff}[Covering number]
    Let $(D, d)$ be a metric space and $A \subset D$.\\
    For $\varepsilon > 0$, the \textit{covering number} of $A$ with respect to $d$ is defined as:
\[
N(A, d, \varepsilon) := \min \left\{ N \in \mathbb{N} : \exists t_1, \dots, t_N \in D \text{ such that } A \subset \bigcup_{i=1}^N B_d(t_i, \varepsilon) \right\},
\]
where 
\[
B_d(t_i, \varepsilon) := \{ y \in D : d(t_i, y) \le \varepsilon \}
\]
denotes the closed ball of radius $\varepsilon$ with center $x_i$.
\end{mydeff}
\begin{mydeff}[Packing number]
Let $(D,d)$ be a metric space and $\varepsilon>0$. The \textit{packing number} is
\[
M(D,d,\varepsilon) = \max \Big\{ m : \exists t_1,\dots,t_m \in D \text{ such that } d(t_i,t_j)\ge \varepsilon \text{ for all } i\neq j \Big\}.
\]
Points are mutually separated by distance at least $\varepsilon$.\\

In the following, we consider for a stochastic process $X$ its associated canonical metric defined as: 
\begin{eqnarray}
    d(s,t) := \sqrt{\mathbb{E}\big[(X_t-X_s)^2\big]}, \qquad s,t\in D.
\end{eqnarray}

\end{mydeff}
\begin{lem}
\label{lem:coveringlowbd}
    We consider the process $(X_t)_{t\in D}$ satisfying a non-degeneracy condition on the increments: there exist constants $c>0$ and $\alpha \in ]0,1[$ such that:
    \begin{eqnarray}
\label{eq:lowerboundcocovering}
\bm{\bm{\mathcal{H}}}:\quad 
\forall (s,t)\in D^2, \tab d(s,t)^2 \ge c |t-s|^{2\alpha}
\end{eqnarray}
Then:
\begin{eqnarray}
    N(D,d,\varepsilon) \ge diam\left(D,|.|\right)\left(\frac{\sqrt{c}}{2\varepsilon}\right)^{1/\alpha}
\end{eqnarray}
where \(diam\left(D,|.|\right) = \sup_{(s,t)\in D} |t-s|\) the diameter of $D$.
\end{lem}
\textit{\textbf{Proof}}\\
We observe for any $\varepsilon>0$ that: 
\begin{eqnarray}
    N(D,d,\varepsilon) \ge M(D,d,2\varepsilon)  \nonumber
\end{eqnarray}
From the assumption, the canonical metric satisfies
\[
d(s,t)\ge \sqrt{c}\, |t-s|^{\alpha}.
\]

To obtain $d(s,t)\ge \varepsilon$ for $(t,s)\in D$, it suffices that
\[
|t-s| \ge \left(\frac{\varepsilon}{\sqrt{c}}\right)^{1/\alpha}.
\]
On the other hand, given arbitrary points $t_1 < t_2 < \dots < t_m \in D$
the minimal spacing is 
\[
\delta_{\min} := \min_{i\in \llbracket 1,m \rrbracket} (t_{i+1} - t_i) \ge \delta\left(\varepsilon\right):=\left(\frac{\varepsilon}{\sqrt{c}}\right)^{1/\alpha}.
\]

Then the number of points \(m\) satisfies:
\[
m \le \frac{diam\left(D,|.|\right)}{\delta_{\min}} + 1 \le \frac{diam\left(D,|.|\right)}{\delta\left(\varepsilon\right)} + 1,
\]
 The latter upper bound is reached for $(t_i)_{1\leq i \leq m}$ being a linear subdivision of $D$. as a result:
\begin{eqnarray}
    N(D,d,\varepsilon) \ge \frac{diam\left(D,|.|\right)}{\delta\left(2\varepsilon\right)}
\end{eqnarray}
which leads to the claimed result.\\
\tab\tab\tab\tab\tab\tab\tab\tab\tab\tab\tab\tab\tab\tab\tab\tab\tab\tab\tab\tab\tab\tab\tab\tab\tab\tab\tab\tab\tab\tab\tab\tab\tab\tab $\blacksquare$
We end this section by presenting the Sudakov's lower bound theorem that is going to be useful in the upcoming proofs.
\begin{theo}[Sudakov minoration]
\label{thm:sudakov}
Let $X$ be a Gaussian process and and $(\mathcal{T},d)$ a metric space with the canonical metric:
\[
d(s,t) := \sqrt{\mathbb{E}\big[(X_t-X_s)^2\big]}, \qquad (s,t)\in \mathcal{T}^2.
\]
Then there exists a universal constant $C>0$ such that for every $\varepsilon>0$,
\begin{eqnarray}
    \mathbb{E}\left[\underset{t\in \mathcal{T}}\sup X_t\right]
\;\ge\;
C\,\varepsilon \sqrt{\ln \left(N(\mathcal{T},d,\varepsilon)\right)}
\end{eqnarray}
\end{theo}

\subsection{Proof of Proposition \ref{prop:Expectationbound}}
\label{app:uniformboundedness}
    For any $i\in \llbracket 1,n \rrbracket$, the process $\omega^{x_i}_{H,T}(.)$ is uniformly bounded in $\mathcal{L}^2$ as it's weakly stationary with a finite second order moment. \\\\
Thus, using theorem \ref{theo:boreltis}, we have for any $t\in \R^{*}_{+}$:

\begin{eqnarray}
    \Prob\left(\left|\sup_{u\in D}\omega^{x_i}_{H,T}(u)-\E\left(\sup_{u\in D}\omega^{x_i}_{H,T}(u)\right)\right|>t \right) \leq e^{-\frac{t^2}{2\sigma_{x,r}^2}} \nonumber
\end{eqnarray}
Where $\sigma_{x_i}^2 =  \sup_{u\in D}\E\left(\omega^{x_i}_{H,T}(u)^2 \right)$.\\\\
On the other hand, we have for any $t\geq 0$:
\begin{eqnarray}
    \Prob\left(\sup_{u\in D}\omega^{x_i}_{H,T}(u)>t+\E\left(\sup_{u\in D}\omega^{x_i}_{H,T}(u)\right) \right) \leq \Prob\left(\left|\sup_{u\in D}\omega^{x_i}_{H,T}(u)-\E\left(\sup_{u\in D}\omega^{x_i}_{H,T}(u)\right)\right|>t \right) \nonumber
\end{eqnarray}
Which leads to 
    \begin{eqnarray}
    \Prob\left(\sup_{u\in D}\omega^{x_i}_{H,T}(u)>t \right) \leq e^{-\frac{\left(\E\left(\sup_{u\in D}\omega^{x_i}_{H,T}(u)\right)-t\right)^2}{2\sigma_{x_i}^2}} \nonumber
\end{eqnarray}
Besides, Using the autocovariance function of $\omega_{H,T}$, we get for any $i\in \llbracket 1,n\rrbracket$:
     \begin{eqnarray}
       \E\left(\left(\omega_{H,T}^{x_i}(u)\right)^2 \right)=\frac{\lambda^2}{H(1-2H)}\left(\frac{|u-x_i|}{T}\right)^{2H} 
        \nonumber
    \end{eqnarray}
    Given that $u\in B(x,r)$, we obtain:
      \begin{eqnarray}
       \sigma_{x_i}^2\leq\frac{\lambda^2}{H(1-2H)}r_i^{2H} 
        \nonumber
    \end{eqnarray}
    Using the union bound, we obtain:
    \begin{eqnarray}
        \sup_{u\in D}\E\left(\omega^{x_i}_{H,T}(u)^2 \right)\leq  \frac{n\lambda^2}{T^{2H}H(1-2H)}\sum_{i=1}^nr_i^{2H} 
        \nonumber
    \end{eqnarray}
Consequently, this leads to:
\begin{eqnarray}
    \Prob\left(\sup_{u\in D}\omega^{x_i}_{H,T}(u)>t \right) \leq e^{-\frac{T^{2H}H(1-2H)\left(\E\left(\sup_{u\in D}\omega^{x_i}_{H,T}(u)\right)-t\right)^2}{2n\lambda^2\left(r^{*}\right)^{2H}}} \nonumber
\end{eqnarray}
By straightforward computations, the S-fBM process satisfies for any $i\in \llbracket1,n \rrbracket$:
\begin{eqnarray}
    \E \left(\left|\omega^{x_i}_{H,T}(t)-\omega^{x_i}_{H,T}(s)\right|^2 \right)=\E \left(\left|\omega_{H,T}(s)-\omega_{H,T}(t)\right|^2 \right)=\frac{\lambda^2}{H(1-2H)T^{2H}}\left|t-s\right|^{2H} \nonumber
\end{eqnarray}
which means that Hypothesis~\hyperref[eq:lowerboundcocovering]{\(\bm{\mathcal H}\)} is satisfied for respectively $\alpha=2H$ and $c=\frac{\lambda^2}{H(1-2H)T^{2H}}$. Thanks to Sudakov's minoration (Theorem \ref{thm:sudakov}) together with Lemma \ref{lem:coveringlowbd}, one has:
\begin{eqnarray}
    \Prob\left(\sup_{u\in B(x,r)}\omega^{x_i}_{H,T}(u)>t \right) \leq 2e^{-\frac{T^{2H}H(1-2H)\left(a(r)-t\right)^2}{2n\lambda^2\left(r^{*}\right)^{2H}}}\nonumber
\end{eqnarray}
where $a(r)=C\,\varepsilon \sqrt{\ln\!\left(
\frac{r}{\sqrt{T}}
\left(\frac{\lambda}{2\varepsilon}\right)^{\frac{1}{2H}}
\big(H(1-2H)\big)^{-\frac{1}{4H}}
\right)}$
This can be simplified as:
\begin{eqnarray}
    \Prob\left(\sup_{u\in B(x,r)}\omega^{x_i}_{H,T}(u)>t \right) \leq 2e^{-\frac{T^{2H}H(1-2H)\left(a_{*}-t\right)^2}{2n\lambda^2\left(r^{*}\right)^{2H}}}\nonumber
\end{eqnarray}
where $a_{*}=C\,\varepsilon \sqrt{\ln\!\left(
\frac{r_{*}}{\sqrt{T}}
\left(\frac{\lambda}{2\varepsilon}\right)^{\frac{1}{2H}}
\big(H(1-2H)\big)^{-\frac{1}{4H}}
\right)}$ for any $\varepsilon>0$.\\
Using the union bound, we obtain:
\begin{eqnarray}
    \Prob\left(\sup_{u\in D}\omega^{x_i}_{H,T}(u)>t \right) \leq 2ne^{-\frac{T^{2H}H(1-2H)\left(a_{*}-t\right)^2}{2n\lambda^2\left(r^{*}\right)^{2H}}}
\end{eqnarray}
Besides,
\begin{eqnarray}
    \E\left(e^{\sup_{u\in D}\omega^{x}_{H,T}(u)} \right) =\int_0^{+\infty}\Prob\left(e^{\sup_{u\in D}\omega^{x}_{H,T}(u)}>t \right)dt  =\int_0^{+\infty}\Prob\left(\sup_{u\in D}\omega^{x}_{H,T}(u)>\ln(t) \right)dt  \nonumber
\end{eqnarray}
Which gives the following:
\begin{eqnarray}
    \E\left(e^{\sup_{u\in D}\omega^{x}_{H,T}(u)} \right) \leq 2n\int_0^{+\infty}e^{-\eta_T(H,\lambda^2)\left(a_{*}-\ln(t)\right)^2}dt\nonumber
\end{eqnarray}
where $\eta_T(h,l)=\frac{T^{2h}h(1-2h)}{2nl\left(r^{*}\right)^{2h}}$.\\\\
After some algebra, we have:
\begin{eqnarray}
    \E\left(e^{\sup_{u\in D}\omega^{x}_{H,T}(u)} \right) \leq 2n\sqrt{\frac{\pi}{\eta_T(H,\lambda^2)}}e^{\left(\frac{\eta_T(H,\lambda^2) a_{*}+1}{\sqrt{\eta_T(H,\lambda^2)}} \right)^2+\eta_T(H,\lambda^2) a_{*}^2} \nonumber
\end{eqnarray}
Which concludes the proof.\\
\tab\tab\tab\tab\tab\tab\tab\tab\tab\tab\tab\tab\tab\tab\tab\tab\tab\tab\tab\tab\tab\tab\tab\tab\tab\tab\tab\tab\tab\tab\tab\tab\tab\tab $\blacksquare$
\subsection{Proof of Theorem \ref{theo:upperbounddevSfbm}}
\label{app:deviationineq}
    We consider the following tuple $(t_j)_{1\leq j \le m} \in \R_{+}^m$ and a $\delta>0$.\\\\
\begin{enumerate}
    \item  We have using the union bound: 
    \begin{eqnarray}
        \Prob\left(\underset{1 \leq j \leq m }{\sup}\omega_{H,T}(t_j)\geq \delta \right) \leq \sum_{j=1}^m\Prob\left(\omega_{H,T}(t_j)\geq \delta \right) \nonumber
    \end{eqnarray}
    Using Markov's inequality, we have:
    \begin{eqnarray}
        \forall j\in \llbracket 1,m \rrbracket,\forall x>0,\tab  \Prob\left(\omega_{H,T}(t_j)\geq \delta \right) \leq \E\left(e^{x\omega_{H,T}(t_j)}\right)e^{-x\delta} \nonumber
    \end{eqnarray}
    Which means that:
    \begin{eqnarray}
        \forall j\in \llbracket 1,m \rrbracket,\forall x>0,\tab  \ln\left(\Prob\left(\omega_{H,T}(t_j)\geq \delta \right)\right) \leq \ln\left(\E\left(e^{x\omega_{H,T}(t_j)}\right)\right)-x\delta \nonumber
    \end{eqnarray}
    In particular:
    \begin{eqnarray}
        \forall j\in \llbracket 1,m \rrbracket, \tab \ln\left(\Prob\left(\omega_{H,T}(t_j)\geq \delta \right)\right) \leq \underset{x>0}{\inf} \left\{\ln\left(\E\left(e^{x\omega_{H,T}(t_j)}\right)\right)-x\delta\right\} \nonumber
    \end{eqnarray}
  As   $ \forall j\in \llbracket 1,m \rrbracket, \tab \omega_{H,T}(t_j) \sim \mathcal{N}\left(-\frac{\lambda^2}{4H(1-2H)},\frac{\lambda^2}{2H(1-2H)} \right)$, thus:
   \begin{eqnarray}
        \forall j\in \llbracket 1,n \rrbracket,\forall x>0,\tab   \E\left(e^{x\omega_{H,T}(t_j)}\right)=e^{\frac{\lambda^2}{4H(1-2H)}\left(x^2-x \right)} \nonumber
    \end{eqnarray}
    Meaning that:
    \begin{eqnarray}
        \forall j\in \llbracket 1,m \rrbracket, \tab \ln\left(\Prob\left(\omega_{H,T}(t_j)\geq \delta \right)\right) \leq \underset{x>0}{\inf} \left\{\frac{\lambda^2}{4H(1-2H)}x^2-\left(\delta+\frac{\lambda^2}{4H(1-2H)} \right)x \right\} \nonumber
    \end{eqnarray}
After some algebra, we have the right hand side is equal to: $-\frac{\left(\delta+\frac{\lambda^2}{4H(1-2H)} \right)^2H(1-2H)}{\lambda^2}$.\\\\
We conclude then that:
\begin{eqnarray}
        \forall j\in \llbracket 1,m \rrbracket,\tab \Prob\left(\omega_{H,T}(t_j)\geq \delta \right) \leq e^{-\frac{\left(\delta+\frac{\lambda^2}{4H(1-2H)} \right)^2H(1-2H)}{\lambda^2}} \nonumber
    \end{eqnarray}
    On the other hand, we have:
\begin{eqnarray}
        \forall j\in \llbracket 1,m \rrbracket,\tab \Prob\left(\omega_{H,T}(t_j)\leq -\delta \right) =\Prob\left(-\omega_{H,T}(t_j)\geq \delta \right) \nonumber
    \end{eqnarray}
As $\omega_{H,T}(.)$ is a gaussian process, performing the same computations  leads to:
\begin{eqnarray}
        \forall j\in \llbracket 1,m \rrbracket,\tab \Prob\left(-\omega_{H,T}(t_j)\geq \delta \right) \leq e^{-\frac{\left(\delta-\frac{\lambda^2}{4H(1-2H)} \right)^2H(1-2H)}{\lambda^2}} \nonumber
    \end{eqnarray}
Moreover, one has:
\begin{eqnarray}
        \forall j\in \llbracket 1,m \rrbracket,\tab \Prob\left(|\omega_{H,T}(t_j)|\geq \delta \right)=\Prob\left(\omega_{H,T}(t_j)\leq -\delta \right) +\Prob\left(\omega_{H,T}(t_j)\geq \delta \right) \nonumber
    \end{eqnarray}
Consequently, we get:

\begin{eqnarray}
        \Prob\left(\underset{1 \leq j \leq m }{\sup}|\omega_{H,T}(t_j)|\geq \delta \right) \leq 2me^{-\frac{\left(\delta-\frac{\lambda^2}{4H(1-2H)} \right)^2H(1-2H)}{\lambda^2}} \left(e^{\frac{\frac{\delta\lambda^2}{H(1-2H)} H(1-2H)}{\lambda^2}}+1 \right)\nonumber
\end{eqnarray}
  Which leads to the given result.
\item
   For any $\eta>0$ and $\theta \in \N^*$, we have:
    $$\E\left(e^{\eta\left\lvert\underset{1 \leq j \leq m}{\sup}\omega_{H,T}(t_j)\right\rvert^{\theta} } \right)=\E\left(\sum_{k\geq 0}\frac{\eta^k}{k!}\left\lvert\underset{1 \leq j \leq m}{\sup}\omega_{H,T}(t_j)\right\rvert ^{k\theta}\right)$$
Using Lebesgue's monotone convergence theorem, we get:
$$\E\left(e^{\eta\left\lvert\underset{1 \leq j \leq m}{\sup}\omega_{H,T}(t_j)\right\rvert^{\theta} } \right)=\sum_{k\geq 0}\frac{\eta^k}{k!}\E\left(\left\lvert\underset{1 \leq j \leq m}{\sup}\omega_{H,T}(t_j)\right\rvert ^{k\theta}\right)$$
Which means that:
$$\E\left(e^{\eta\left\lvert\underset{1 \leq j \leq m}{\sup}\omega_{H,T}(t_j)\right\rvert^{\theta} } \right)\leq1+\sum_{k\geq 1}\frac{\eta^k}{k!}\E\left(\left(\underset{1 \leq j \leq m}{\sup}\left\lvert\omega_{H,T}(t_j)\right\rvert \right)^{k\theta}\right)$$
Besides, for any $p\in \N^*\cup\{+\infty\}$, we have:
$$\E\left(\left(\underset{1 \leq j \leq m}{\sup}\left\lvert\omega_{H,T}(t_j)\right\rvert \right)^{p}\right) \leq \E\left(\left(\sum_{j=1}^n\left\lvert\omega_{H,T}(t_j)\right\rvert \right)^{p}\right)$$
Thus,
$$\E\left(\left(\underset{1 \leq j \leq m}{\sup}\left\lvert\omega_{H,T}(t_j)\right\rvert \right)^{p}\right) \leq n^p\E\left(\left(\sum_{j=1}^n\frac{1}{n}\left\lvert\omega_{H,T}(t_j)\right\rvert \right)^{p}\right)$$
Using Jensen's inequality, we obtain:
$$\E\left(\left(\underset{1 \leq j \leq m}{\sup}\left\lvert\omega_{H,T}(t_j)\right\rvert \right)^{p}\right) \leq n^{p-1}\sum_{j=1}^n\E\left(\left\lvert\omega_{H,T}(t_j)\right\rvert ^{p}\right)$$

Besides, for any $  j\in \llbracket 1,n \rrbracket$, $\omega_{H,T}(t_j) \sim \mathcal{N}\left(-\frac{\lambda^2}{4H(1-2H)},\frac{\lambda^2}{2H(1-2H)} \right)$.\\
For simplicity, we denote here: $\mu_0:=-\frac{\lambda^2}{4H(1-2H)}$ and $\sigma_0^2:=\frac{\lambda^2}{2H(1-2H)}$.\\

Using Jensen's inequality again, we have:
 $$\E\left(\left(\underset{1 \leq j \leq m}{\sup}\left\lvert\omega_{H,T}(t_j)\right\rvert \right)^{p}\right) \leq \left(2m\right)^{p-1}\sum_{j=1}^n\left(\E\left(\left\lvert\omega_{H,T}(t_j)-\mu_0\right\rvert ^{p}\right)+\left\lvert\mu_0\right\rvert ^{p}\right)$$
Using lemma \ref{lem:upperboundmomentgaussian}, we get:
 $$\E\left(\left(\underset{1 \leq j \leq m}{\sup}\left\lvert\omega_{H,T}(t_j)\right\rvert \right)^{p}\right) \leq \left(2m\right)^{p-1}\sum_{j=1}^n\left(\left(2(\sigma_0)^2\right)^{\frac{p}{2}}p\Gamma\left(\frac{p}{2}\right)+\left\lvert\mu_0\right\rvert ^{p}\right)$$
By upper bounding the right hand side, we obtain:
 $$\E\left(\left(\underset{1 \leq j \leq m}{\sup}\left\lvert\omega_{H,T}(t_j)\right\rvert \right)^{p}\right) \leq \left(2m\right)^{p-1}m\left(\nu^{p}p\Gamma\left(\frac{p}{2}\right)+\left(\frac{\nu^2}{4}\right) ^{p}\right)$$
As a result, we have:
$$\E\left(e^{\eta\left\lvert\underset{1 \leq j \leq m}{\sup}\omega_{H,T}(t_j)\right\rvert^{\theta} } \right)\leq1+\sum_{k\geq 1}\frac{\eta^k}{k!} \left(2m\right)^{k\theta-1}m\left(\nu^{k\theta}k\theta\Gamma\left(\frac{k\theta}{2}\right)+\left(\frac{\nu^2}{4}\right) ^{k\theta}\right)$$
We set $\theta=2$. Thus,
$$\E\left(e^{\eta\left\lvert\underset{1 \leq j \leq m}{\sup}\omega_{H,T}(t_j)\right\rvert^{2} } \right)\leq1+\frac{1}{2}\sum_{k\geq 1}\frac{1}{k!} \left(m^2\frac{\eta\nu^4}{4}\right)^{k}+\sum_{k\geq 1}\frac{\eta^k}{(k-1)!}(4m^2)^{k}\nu^{2k}\Gamma\left(k\right)$$
Which means that:
 $$ \E\left(e^{\eta\left\lvert\underset{1 \leq j \leq m}{\sup}\omega_{H,T}(t_j)\right\rvert^{2} } \right)\leq1+\frac{1}{2}\left(e^{m^2\frac{\eta\nu^4}{4}}-1\right)+\sum_{k\geq 1}(\eta 4m^2\nu^2)^{k}$$
 Thus, we get:
 $$ \E\left(e^{\eta\left\lvert\underset{1 \leq j \leq m}{\sup}\omega_{H,T}(t_j)\right\rvert^{
 2
 } } \right)\leq 1+\frac{1}{2}\left(e^{m^2\frac{\eta\nu^4}{4}}-1\right)+\frac{\eta 4m^2\nu^2}{1-\eta 4m^2\nu^2}$$
 Thanks to Markov's inequality, we have for any $x>0$:
\begin{eqnarray}
    \Prob\left(\underset{1 \leq j \leq m}{\sup}\left\lvert\omega_{H,T}(t_j)\right\rvert \geq x \right)\leq e^{-\eta x^{2}}\left(1+\frac{1}{2}\left(e^{m^2\frac{\eta\nu^4}{4}}-1\right)+\frac{\eta 4m^2\nu^2}{1-\eta 4m^2\nu^2} \right) \nonumber
\end{eqnarray}
Which lead to:
\begin{eqnarray}
    \ln\left(\Prob\left(\underset{1 \leq j \leq m}{\sup}\left\lvert\omega_{H,T}(t_j)\right\rvert \geq x \right)\right)\leq -\eta x^{2}+\ln\left(1+\frac{1}{2}\left(e^{m^2\frac{\eta\nu^4}{4}}-1\right)+\frac{\eta 4m^2\nu^2}{1-\eta 4m^2\nu^2} \right) \nonumber
\end{eqnarray}
Furthermore, in small intermittency, we have:
\begin{eqnarray}
    \ln\left(1+\frac{1}{2}\left(e^{m^2\frac{\eta\nu^4}{4}}-1\right)+\frac{\eta 4m^2\nu^2}{1-\eta 4m^2\nu^2} \right) = \eta m^2\nu^2\left(4+\frac{\nu^2}{8}\right)+ o\left(\lambda^2\right)
    \nonumber
\end{eqnarray}
We choose $\eta=\frac{1}{\nu^2\left(4+\frac{\nu^2}{8}\right)}$.
Thus, we get:
\begin{eqnarray}
    \ln\left(\Prob\left(\underset{1 \leq j \leq m}{\sup}\left\lvert\omega_{H,T}(t_j)\right\rvert \geq x \right)\right)\leq -\frac{x^{2}}{\nu^2\left(4+\frac{\nu^2}{8}\right)} +\ln\left(1+\frac{1}{2}\left(e^{m^2\frac{\nu^2}{4+\frac{\nu^2}{8}}}-1\right)+\frac{1}{\frac{1+\frac{\nu^2}{32}}{m^2}-1} \right) \nonumber
\end{eqnarray}

From which the claimed result follows.
\end{enumerate}
     
\tab\tab\tab\tab\tab\tab\tab\tab\tab\tab\tab\tab\tab\tab\tab\tab\tab\tab\tab\tab\tab\tab\tab\tab\tab\tab\tab\tab\tab\tab\tab\tab\tab\tab $\blacksquare$

\subsection{Proof of Theorem \ref{thm:deviationineqmrm}}
We start by introductory dyadic-chaining bound.
\begin{lem}[Gaussian maximum bound]\label{lem:gaussmax}
If $\xi_1,\dots,\xi_N$ are centered Gaussian (not necessarily independent) with
$\text{Var}(\xi_i)\le\sigma^2$ for all $i$, then
\[
\E\Big[\max_{1\le i\le N}\xi_i\Big] \;\le\; \sigma\sqrt{2\ln N}.
\]
\end{lem}

\textit{\textbf{Proof}}\\
For any $\theta>0$, by Jensen's inequality applied to $\exp(\theta\cdot)$,
\[
\exp\Big(\theta\,\E\big[\max_i\xi_i\big]\Big) \le \E\Big[\exp\big(\theta\max_i\xi_i\big)\Big] = \E\Big[\max_i e^{\theta\xi_i}\Big] \le \sum_{i=1}^N \E\big[e^{\theta\xi_i}\big] \le N\,e^{\theta^2\sigma^2/2},
\]
using $\E[e^{\theta\xi_i}]=e^{\theta^2\text{Var}(\xi_i)/2}\le e^{\theta^2\sigma^2/2}$. which leads to,
$\E[\max_i\xi_i] \le \frac{\ln N}{\theta}+\frac{\theta\sigma^2}{2}$; maximizing over $\theta$ (set
$\theta=\sqrt{2\ln N}/\sigma$) gives $\E[\max_i\xi_i]\le\sigma\sqrt{2\ln N}$.\\
\tab\tab\tab\tab\tab\tab\tab\tab\tab\tab\tab\tab\tab\tab\tab\tab\tab\tab\tab\tab\tab\tab\tab\tab\tab\tab\tab\tab\tab\tab\tab\tab\tab\tab $\blacksquare$

\begin{prop}[Dyadic chaining bound]\label{prop:chaining}
For $\delta>0$, let $(X_t)_{t\in[0,\delta]}$ be centered Gaussian process with:
\begin{eqnarray}
   \E[(X_t-X_s)^2]\le L^2|t-s|^{2H} 
\end{eqnarray}
 for some
$H\in(0,1)$ and $L>0$. Then:
\begin{eqnarray}
    \E\Big[\sup_{t\in[0,\delta]}X_t\Big] \;\le\; C(H)\,L\,\delta^H, \qquad
C(H) := \sqrt{2\ln 2}\sum_{n=1}^\infty (n+1)^{1/2}\,2^{-nH} < \infty.
\end{eqnarray}
\end{prop}

\textit{\textbf{Proof}}\\
For $n\ge 0$ let $\mathcal D_n:=\{k\delta/2^n : k=0,1,\dots,2^n\}$ be the $n$-th dyadic grid on
$[0,\delta]$, and for $t\in[0,\delta]$ let $\pi_n(t)\in\mathcal D_n$ be the nearest dyadic point at
level $n$, so $\pi_n(t)\to t$ as $n\to\infty$.
Telescoping (with $\pi_0(t)\equiv 0$) gives
\[
X_t - X_0 = \sum_{n=1}^\infty \big(X_{\pi_n(t)}-X_{\pi_{n-1}(t)}\big),
\]
so
\[
\sup_{t\in[0,\delta]} X_t \;\le\; X_0 \;+\; \sum_{n=1}^\infty \max_{s\in\mathcal D_n,\,s'\in\mathcal D_{n-1},\,|s-s'|\le \delta 2^{-n}} \big(X_s - X_{s'}\big).
\]
At level $n$, the number of distinct pairs $(s,s')$ arising this way is at most $2^{n+1}$. Each difference $X_s-X_{s'}$ is centered Gaussian with variance
$\le L^2(\delta/2^n)^{2H}$. By Lemma~\ref{lem:gaussmax},
\[
\E\Big[\max_{n\text{-level pairs}}\big(X_s-X_{s'}\big)\Big] \;\le\; L\Big(\frac{\delta}{2^n}\Big)^{H}\sqrt{2\ln(2^{n+1})}
= L\,\delta^H\,2^{-nH}\sqrt{2(n+1)\ln 2}.
\]
Since $X_0$ is centered, $\E[X_0]=0$. Summing over $n\ge 1$ (justified by monotone
convergence, since the series converges absolutely):
\[
\E\Big[\sup_{[0,\delta]}X_t\Big] \;\le\; \sum_{n=1}^\infty L\,\delta^H\,2^{-nH}\sqrt{2(n+1)\ln2}
= L\,\delta^H\sqrt{2\ln2}\sum_{n=1}^\infty (n+1)^{1/2}2^{-nH}.
\]
The series $\sum_n (n+1)^{1/2}2^{-nH}$ converges for every $H>0$ by the ratio test (geometric decay
$2^{-nH}$ dominates the polynomial factor $(n+1)^{1/2}$), defining the finite constant $C(H)$.
\\
\tab\tab\tab\tab\tab\tab\tab\tab\tab\tab\tab\tab\tab\tab\tab\tab\tab\tab\tab\tab\tab\tab\tab\tab\tab\tab\tab\tab\tab\tab\tab\tab\tab\tab $\blacksquare$

Leveraging on the previous results, one has the following Corollary.

\begin{Corollary}[Bound for $\tilde{\omega}_{H,T}:=\omega_{H,T}-\mu_H$]\label{coro:upperboundexpsupomegatilde} The following claim holds:
\begin{eqnarray}
    \Big|\E\Big[\sup_{t\in D}\tilde\omega_{H,T}(t)\Big]\Big| \;=\; \E\Big[\sup_{t\in D}\tilde\omega_{H,T}(t)\Big] \;\le\; C(H)\,L\,\operatorname{diam}(D)^H.
\end{eqnarray}
where $L^2=\lambda^2/(H(1-2H)T^{2H})$.\\
The equality $|\E[\sup_D\tilde\omega_{H,T}]|=\E[\sup_D\tilde\omega_{H,T}]$ holds thanks to Sudakov's inequality (Theorem \ref{thm:sudakov}).\\
In the sequel, we denote $\varsigma_{H}(D):=C(H)\,L\,\operatorname{diam}(D)^H$.
\end{Corollary}
\label{app:deviationineqmrm}
Now, we are ready for the proof of  Theorem \ref{thm:deviationineqmrm}.  
\begin{enumerate}
\item  By definition, one has the two-sided enclosure:
\begin{eqnarray}
\inf_{t\in D} \omega_{H,T}(t) \le  \ln \left(\frac{M_{H,T}(D)}{|D|}\right) \le \sup_{t\in D} \omega_{H,T}(t). \nonumber
\end{eqnarray}

Consequently,
\begin{eqnarray}
\left| \ln \left(\frac{M_{H,T}(D)}{|D|}\right) \right| 
\le \max \left( \left|\inf_{t\in D} \omega_{H,T}(t)\right| , \left|\sup_{t\in D} \omega_{H,T}(t)\right| \right)
 \nonumber
\end{eqnarray}

Therefore, for any $x>0$,
\begin{eqnarray}
\Prob\left( \left| \ln \left(\frac{M_{H,T}(D)}{|D|} \right) \right| \ge x \right)
\le \Prob\Big( \max \left( \left|\inf_{t\in D} \omega_{H,T}(t)\right| , \left|\sup_{t\in D} \omega_{H,T}(t)\right| \right) \ge x \Big). \nonumber
\end{eqnarray}
which means that:
\begin{eqnarray}
\Prob\left( \left| \ln \left(\frac{M_{H,T}(D)}{|D|} \right) \right| \ge x \right)
\le \Prob\Big( \max \left( \left|\inf_{t\in D} \omega_{H,T}(t)-\mu_H\right| , \left|\sup_{t\in D} \omega_{H,T}(t)-\mu_H\right| \right) \ge x-\mu_H \Big). \nonumber
\end{eqnarray}
As $\tilde{\omega}_{H,T}:=\omega_{H,T}-\mu_H$ is a centered Gaussian process
\begin{eqnarray}
\Prob\left( \left| \ln \left(\frac{M_{H,T}(D)}{|D|} \right) \right| \ge x \right)
\le 2\Prob\Big( \left|\sup_{t\in D} \tilde{\omega}_{H,T}\right| \ge x-\mu_H \Big). \nonumber
\end{eqnarray}
Besides, we have:
\begin{eqnarray}
 \Prob\left( \left|\sup_{t\in D} \tilde{\omega}_{H,T}\right| \ge x \right) \leq 2\Prob\left(\left|\sup_{t\in D} \tilde{\omega}_{H,T}-\mathbb{E}\left[\sup_{t\in D} \tilde{\omega}_{H,T}\right]\right| \ge x-\mu_H-\left|\mathbb{E}\left[\underset{t\in D}\sup\, \tilde{\omega}_{H,T}\right]\right|  \right) 
. \nonumber
\end{eqnarray}

applying Borel-TIS inequality (Theorem \ref{theo:boreltis}) together with Corollary \ref{coro:upperboundexpsupomegatilde} lead to:
\begin{eqnarray}
 \forall x>\Tilde{\mu},\quad \Prob\left( \left|\sup_{t\in D} \tilde{\omega}_{H,T}(t)\right| \ge x \right) \leq  2 \exp\left( - \frac{\left(x-\Tilde{\mu}_H\right)^2}{2 \sigma^2_{H,T}} \right). \nonumber
\end{eqnarray}

where
\begin{eqnarray}
\sigma^2_{H,T} = \frac{\lambda^2}{2H(1-2H)},\quad  \Tilde{\mu}_H=\mu_H+\varsigma_H(D) \nonumber
\end{eqnarray}
Meaning that:
\begin{eqnarray}
 \forall x>\Tilde{\mu},\quad \Prob\left(\left|
\ln\left(\frac{M_{H,T}(D)}{|D|}\right)\right| \ge x
\right)
\le
2 \exp\left(- \frac{\left(x-\Tilde{\mu}_H\right)^2}{\frac{\lambda^2}{H(1-2H)}} \right).\nonumber
\end{eqnarray}
Thus,
\begin{eqnarray}
\limsup_{x\to\infty}
x^{-2}
\ln
\Prob\left(
\left|
\ln\left(\frac{M_{H,T}(D)}{|D|}\right)\right| \ge x
\right)
\le
-\frac{H(1-2H)}{\lambda^2}.\nonumber
\end{eqnarray}
\item 

Similar computations as before lead to:
\begin{eqnarray}
\inf_{t\in D} \omega_{H,T}(t) \le  \ln \left(\frac{M_{H,T}(D)}{|D|}\right) \le \sup_{t\in D} \omega_{H,T}(t). \nonumber
\end{eqnarray}

Consequently,
\begin{eqnarray}
\left| \ln \left(\frac{M_{H,T}(D)}{|D|}\right)- \lambda \frac{\Omega_{H,T}\left(D\right)}{|D|} \right| 
\le  \max \left( \left|\frac{1}{|D|}\int_D\inf_{t\in D} Y_{H,T}^u(t) du+\mu_H\right| , \left|\frac{1}{|D|}\int_D\sup_{t\in D} Y_{H,T}^u(t) du+\mu_H\right| \right)
 \nonumber
\end{eqnarray}
where: $Y_{H,T}^u(t)=\omega_{H,T}(t)-\omega_{H,T}(u)$. 
Thus, 
\begin{eqnarray}
\left| \ln \left(\frac{M_{H,T}(D)}{|D|}\right)- \lambda \frac{\Omega_{H,T}\left(D\right)}{|D|} \right| 
\le  \max \left( \left|\frac{1}{|D|}\int_D\inf_{t\in D} Y_{H,T}^u(t) du\right| , \left|\frac{1}{|D|}\int_D\sup_{t\in D} Y_{H,T}^u(t) du\right| \right)+\mu_H
 \nonumber
\end{eqnarray}
Therefore, for any $x\in \R$,
\begin{eqnarray}
&& \Prob\left(
\left|
\ln\left(\frac{M_{H,T}(D)}{|D|}\right)
- \lambda \frac{\Omega(D)}{|D|}
\right|
\ge x
\right)\le
\nonumber \\
&& 
\Prob\Bigg(
\max \Bigg\{
\left|\frac{1}{|D|}\int_D \inf_{t\in D} Y_{H,T}^u(t)\,du\right|,
\left|\frac{1}{|D|}\int_D \sup_{t\in D} Y_{H,T}^u(t)\,du\right|
\Bigg\}
\ge x-\mu_H
\Bigg).
\nonumber
\end{eqnarray}
$\omega_{H,T}$ is a Gaussian process so is $Y_{H,T}^u$ for any fixed $u\in D$, then:
\begin{eqnarray}
&& \Prob\left(
\left|
\ln\left(\frac{M_{H,T}(D)}{|D|}\right)
- \lambda \frac{\Omega(D)}{|D|}
\right|
\ge x
\right)\le
\Prob\Bigg(\left|\frac{1}{|D|}\int_D \sup_{t\in D} Y_{H,T}^u(t)\,du\right|
\ge x-\mu_H
\Bigg).
\nonumber
\end{eqnarray}
Thus, 
\begin{eqnarray}
&& \Prob\left(
\left|
\ln\left(\frac{M_{H,T}(D)}{|D|}\right)
- \lambda \frac{\Omega(D)}{|D|}
\right|
\ge x
\right)\le
\Prob\Bigg( \left|\sup_{(t,u)\in D^2} Y_{H,T}^u(t)\right|
\ge x-\mu_H
\Bigg).
\nonumber
\end{eqnarray}
The process $\left(Y_{H,T}^u(t)\right)_{(t,u)\in \R^2}$ is gaussian by construction. Besides,
\begin{eqnarray}
&& \Prob\left( \left|\sup_{(t,u)\in D^2} Y_{H,T}^u(t)\right| \ge 2x-\mu_H \right) \leq\nonumber \\
&&  \Prob\left(\left|\sup_{(t,u)\in D^2} Y_{H,T}^u(t)-\mathbb{E}\left[\sup_{(t,u)\in D^2} Y_{H,T}^u(t)\right]\right| \ge x-\mu_H-\left|\mathbb{E}\left[\sup_{(t,u)\in D^2} Y_{H,T}^u(t)\right]\right|  \right) 
. \nonumber
\end{eqnarray}

Applying Borel-TIS inequality (Theorem \ref{theo:boreltis}) leads to:
\begin{eqnarray}
 \Prob\left( \left|\sup_{(t,u)\in D^2} Y_{H,T}^u(t)\right| \ge x-\mu_H \right)\leq  2 \exp\left( - \frac{\left(x-\mu_H-\left|\mathbb{E}\left[\sup_{(t,u)\in D^2} Y_{H,T}^u(t)\right]\right|\right)^2}{2 \sigma^2_{H,T}} \right). \nonumber
\end{eqnarray}

where
\begin{eqnarray}
\sigma^2_{H,T} = \underset{(t,u)\in D^2}\sup\E\left(Y_{H,T}^u(t)^2\right)\leq \frac{\lambda^2}{H(1-2H)T^{2H}} diam(D)^{2H}. \nonumber
\end{eqnarray}
Following the same arguments as the first point,
\begin{eqnarray}
\limsup_{x\to\infty}
x^{-2}
\ln\left(
\Prob\left(
\left|
\ln\left(\frac{M_{H,T}(D)}{|D|}\right)
- \lambda \frac{\Omega(D)}{|D|}
\right|\right)
\ge x
\right)
\le
-\frac{H(1-2H)T^{2H}}
{2\lambda^2\,\operatorname{diam}(D)^{2H}}.
\nonumber
\end{eqnarray}
\end{enumerate}
\tab\tab\tab\tab\tab\tab\tab\tab\tab\tab\tab\tab\tab\tab\tab\tab\tab\tab\tab\tab\tab\tab\tab\tab\tab\tab\tab\tab\tab\tab\tab\tab\tab\tab $\blacksquare$

\section{Hypothesis testing Appendix}
\label{app:Hypothesistestingappendix}
\subsection{Convergence in distribution as $\lambda\to 0$}

\begin{figure}[H]
    \centering

    \begin{subfigure}[b]{0.45\textwidth}
        \centering
        \includegraphics[width=\textwidth]{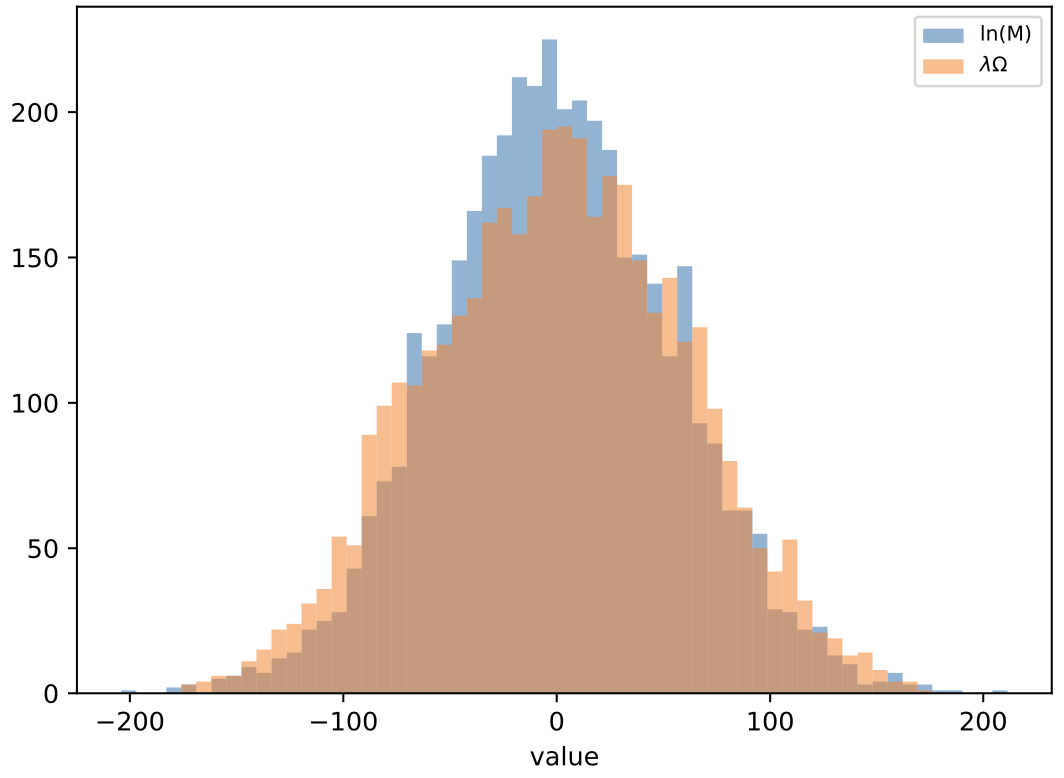}
        \caption{$\lambda^2=1e+2$}
        \label{fig:plot2}
    \end{subfigure}
    \hfill
    \begin{subfigure}[b]{0.45\textwidth}
        \centering
        \includegraphics[width=\textwidth]{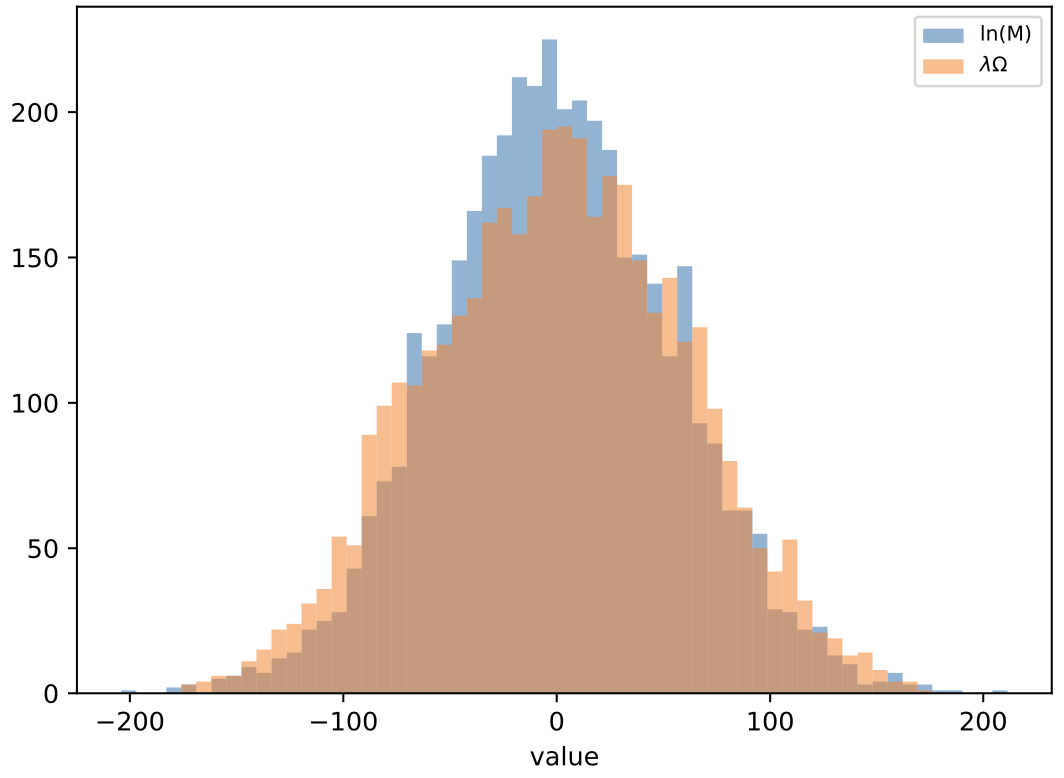}
        \caption{$\lambda^2=1e-2$}
        \label{fig:plot3}
    \end{subfigure}

    \vskip\baselineskip

    \begin{subfigure}[b]{0.45\textwidth}
        \centering
        \includegraphics[width=\textwidth]{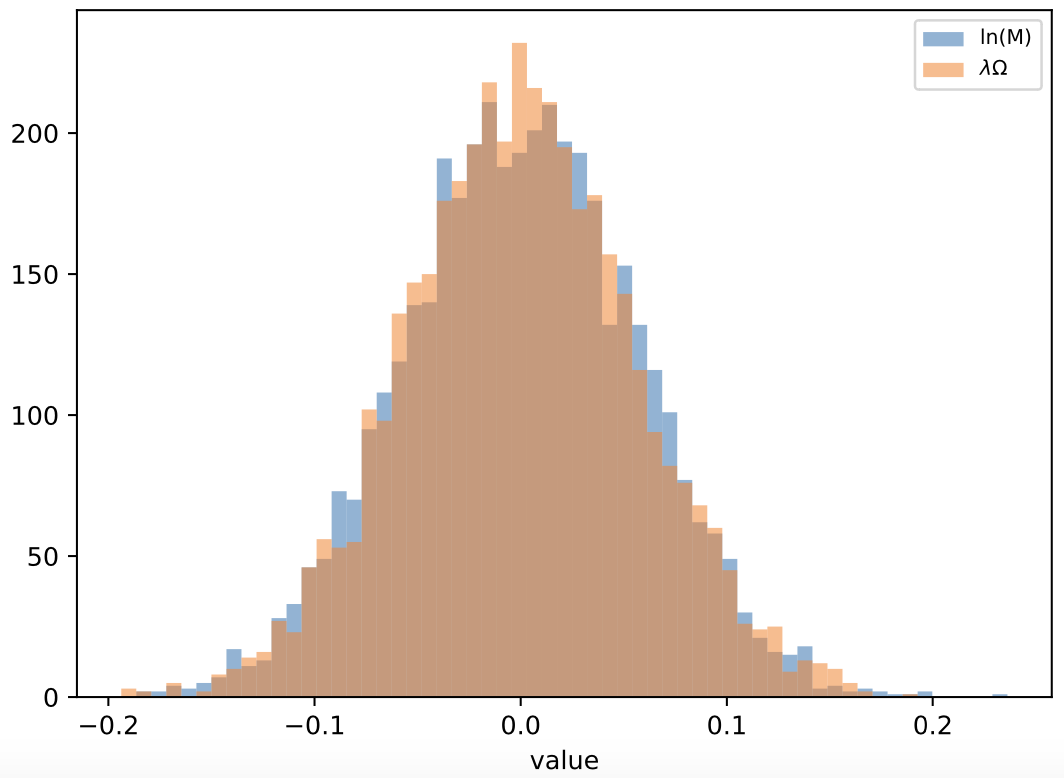}
        \caption{$\lambda^2=1e-4$}
        \label{fig:plot4}
    \end{subfigure}

    \caption{Empirical distributions of the normalized statistic $\ln(M)$ for three values of the intermittency parameter $\lambda^2$. Each histogram is obtained from independent simulations and is shown together with the corresponding Gaussian fit. The empirical standard deviation and the correlation with the reference Gaussian variable are reported for each value of $\lambda^2$.}
    \label{fig:three_plots_gaussianlimitlambda}
\end{figure}

The numerical experiments provide strong evidence for the Gaussian approximation predicted by the central limit theorem as well as the small intermittency approximation. Figure~\ref{fig:three_plots_gaussianlimitlambda} shows that, after normalization, the empirical distributions are centered and exhibit an approximately bell-shaped profile over a wide range of intermittency parameters. As expected, the approximation is most accurate in the small-intermittency regime, where the perturbative assumptions underlying the theoretical analysis are best satisfied. Even for larger values of $\lambda^2$, the empirical distributions remain reasonably symmetric, although deviations from the Gaussian reference become progressively more visible.\\

\begin{figure}[H]
\centering

    \begin{subfigure}[b]{0.45\textwidth}
        \centering
        \includegraphics[width=\textwidth]{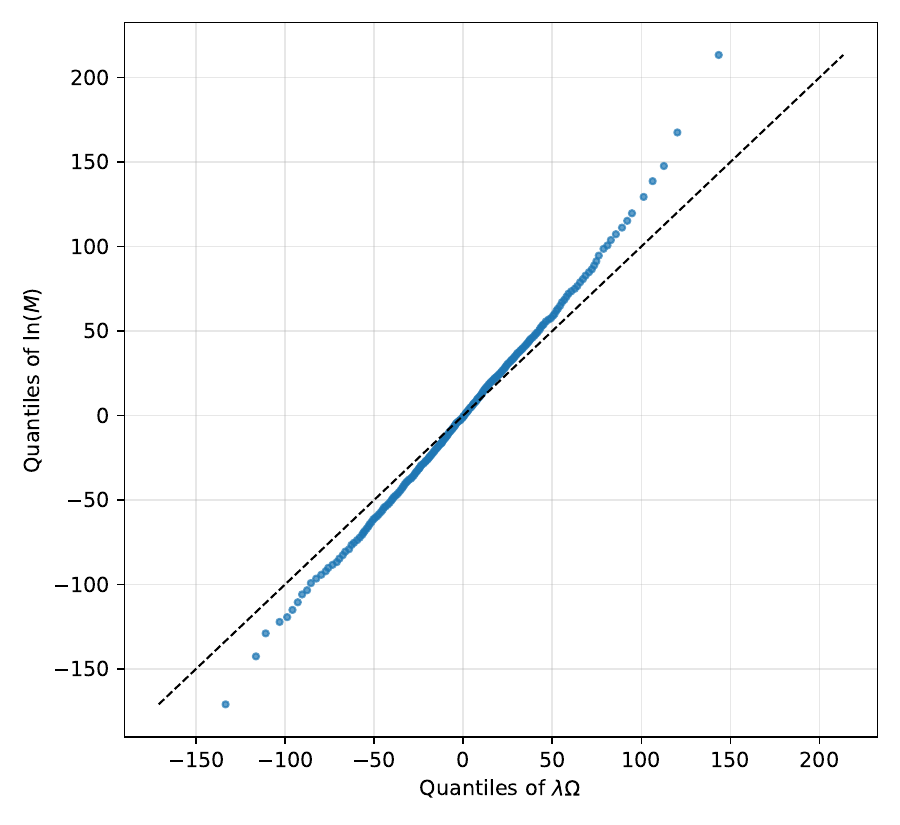}
        \caption{$\lambda^2=1e+2$}
        \label{fig:plot2}
    \end{subfigure}
    \hfill
    \begin{subfigure}[b]{0.45\textwidth}
        \centering
        \includegraphics[width=\textwidth]{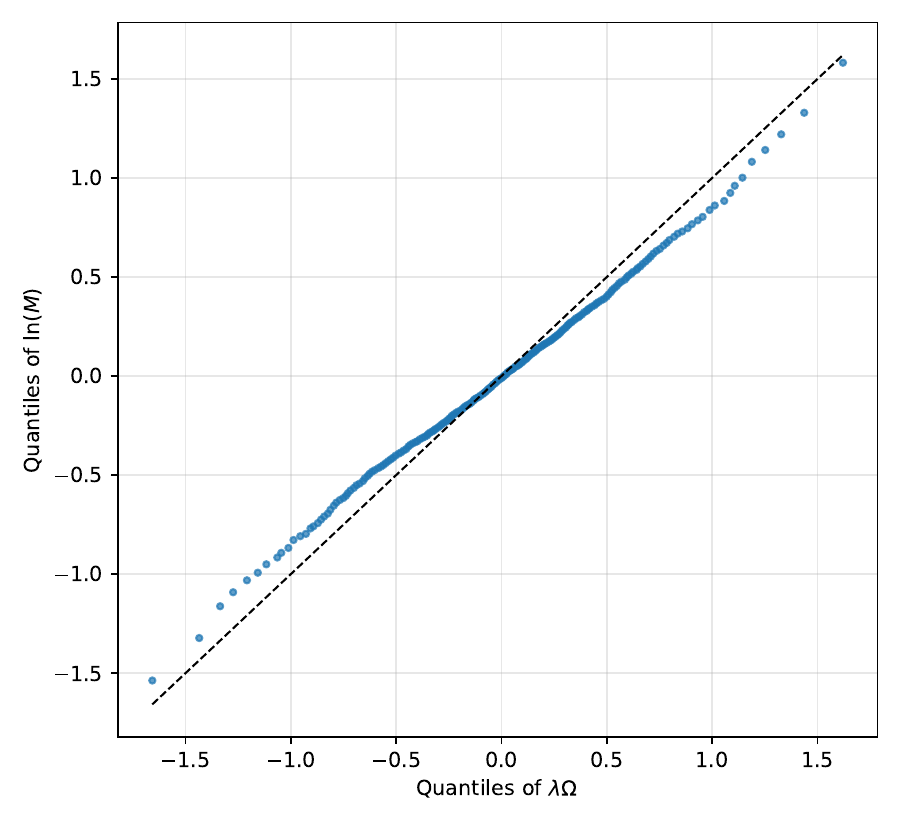}
        \caption{$\lambda^2=1e-2$}
        \label{fig:plot3}
    \end{subfigure}

    \vskip\baselineskip

    \begin{subfigure}[b]{0.45\textwidth}
        \centering
        \includegraphics[width=\textwidth]{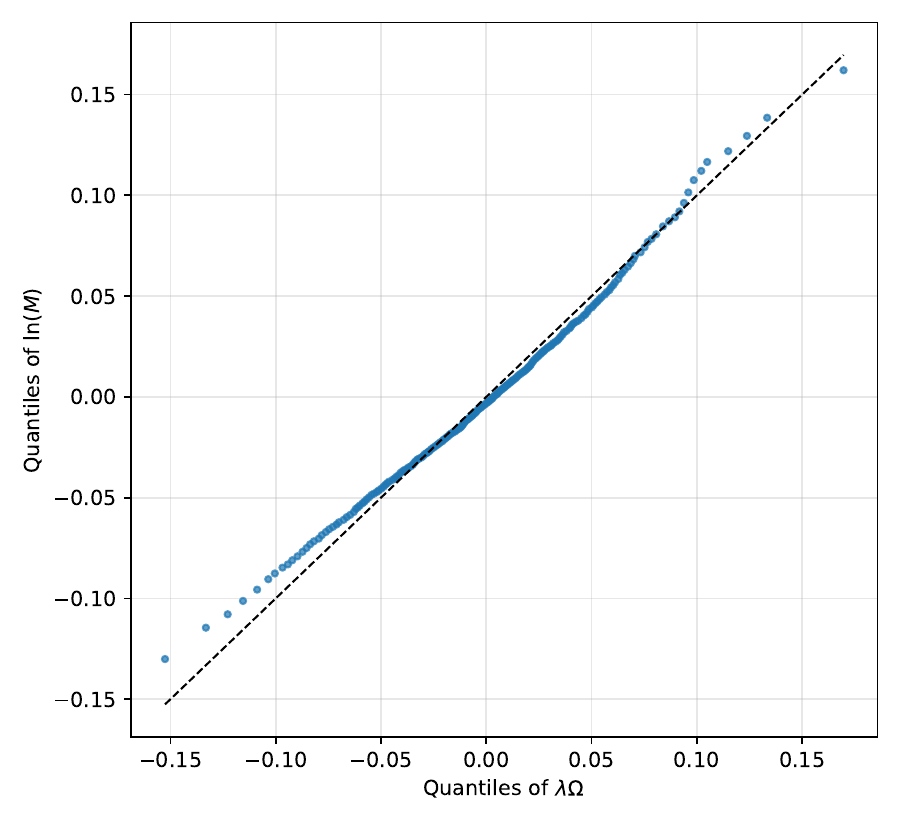}
        \caption{$\lambda^2=1e-4$}
        \label{fig:plot4}
    \end{subfigure}
\caption{Quantile--quantile plots comparing the empirical distribution of $\ln(M)$ with its Gaussian approximation for the same values of $\lambda^2$. Increasing linearity indicates improved agreement with the limiting normal distribution predicted by the central limit theorem.}
\label{fig:qq_clt}
\end{figure}
The Quantile--Quantile plots in Figure~\ref{fig:qq_clt} provide a more sensitive diagnostic. In the small intermittency regime, the empirical quantiles lie almost perfectly on the diagonal, indicating a good agreement with the limiting Gaussian law. As the intermittency increases, departures from linearity first appear in the tails, reflecting the influence of higher-order corrections neglected in the asymptotic theory. Overall, the simulations confirm that the Gaussian approximation accurately captures the finite-dimensional distribution of the normalized statistic over the practically relevant range of parameters, with the quality of the approximation improving as the intermittency decreases.

\subsection{Proof of Proposition \ref{prop:CLTsmallinterm}}
\label{app:proofpropCLTsmallinterm}

Here we denote for any $j\in \llbracket 1, n \rrbracket$,
\begin{eqnarray}
    \begin{cases}
    \Omega_j:=\frac{\Omega_{H,T,\Delta}((j-1)\Delta)}{\Delta}\\
    M_j:=\frac{M_{H,T,\Delta}((j-1)\Delta)}{\Delta}
    \end{cases}
\end{eqnarray}
One has thanks to Proposition 4 in \cite{wu2022rough}:
\[
\E\Big[\big(\Sigma_n\big)^p\Big] = \sum_{i_1,\dots,i_p=1}^n \E\big[\ln(M_{i_1})\cdots\ln(M_{i_p})\big] = \lambda^p\sum_{i_1,\dots,i_p=1}^n\E\big[\Omega_{i_1}\cdots\Omega_{i_p}\big] + n^p\cdot o(\lambda^p),
\]
Dividing by $\lambda^p$,
\[
\E\Big[\big(\lambda^{-1}\Sigma_n\big)^p\Big] = \sum_{i_1,\dots,i_p=1}^n\E\big[\Omega_{i_1}\cdots\Omega_{i_p}\big] + o(1) = \E\big[S_n^p\big] + o(1) \;\xrightarrow[\lambda\to0]{}\; \E\big[S_n^p\big].
\]
Thus every moment of $\lambda^{-1}\Sigma_n$ converges to the corresponding moment of $S_n$. Since
$S_n$ is Gaussian (Eq.~\eqref{eq:gaussianityofSn}), its moments grow as $\E[S_n^{2k}]=(2k-1)!!\,V_n(H)^k$,
which satisfies Carleman's condition ($\sum_k(\E[S_n^{2k}])^{-1/2k}=\infty$), so the law of $S_n$ is
the unique distribution with these moments. Convergence of all moments to those of a
moment-determinate distribution implies convergence in distribution. Hence,

\begin{eqnarray}
\label{eq:convgcelambdasmalllogOmega}
    \lambda^{-1}\Sigma_n\xrightarrow[\lambda\to 0]{d}S_n
\end{eqnarray}
which means that:
\begin{eqnarray}
    \lambda^{-1}\Sigma_n\xrightarrow[\lambda\to 0]{d}S_n \nonumber
\end{eqnarray}

\tab\tab\tab\tab\tab\tab\tab\tab\tab\tab\tab\tab\tab\tab\tab\tab\tab\tab\tab\tab\tab\tab\tab\tab\tab\tab\tab\tab\tab\tab\tab\tab\tab\tab $\blacksquare$
\subsection{Proof of Corollary \ref{prop:V02asH0}}
\label{app:proofcorrolaryV02asH0}
Identical to the proof established elsewhere in this thread (the statement concerns $V_n(H)$
itself, a property of the model, and does not depend on which test statistic is subsequently built
from it): writing $V_n(H)=\frac{L_n^2}{2H(1-2H)}f(H)$, $f(H):=T^{2H}-\frac{2L_n^{2H}}{(2H+1)(2H+2)}$,
$f(0)=0$ gives, by L'H\^opital, $\lim_{H\to0}V_n(H)=L_n^2f'(0)/2=\frac{L_n^2}{2}[2\ln(T/L_n)+3]$.
\\
\tab\tab\tab\tab\tab\tab\tab\tab\tab\tab\tab\tab\tab\tab\tab\tab\tab\tab\tab\tab\tab\tab\tab\tab\tab\tab\tab\tab\tab\tab\tab\tab\tab\tab $\blacksquare$

\section{The Log S-fBM process and small intermittency approximations}
\label{app:appendixLogS-fBM}
\subsection{Autocovariance of the Log S-fBM increments}
The process $\left(\delta_{\tau}X_t\right)_t$ defined in Eq.~\eqref{eq:deltaXtau} has a closed form autocovariance function.

\begin{prop}
\label{prop:autocovdeltaXtau}
    For any $(t,h,\tau)\in \R_{+}^3$, the increment process is covariance stationary with the given autocovariance function:
  \begin{eqnarray}
    \text{Cov}\left(\delta_{\tau}X_t,\delta_{\tau}X_{t+h}\right)
= (\tau-h)\,\mathds{1}_{\{h<\tau\}}
\nonumber
\end{eqnarray}
    
\end{prop}
\textit{\textbf{Proof}}\\
    According to Eq.~\eqref{eq:LogSfbmmodel}, one has:
    \begin{eqnarray}
        \delta_{\tau}X_t = \int_t^{t+\tau}e^{\frac{\omega_{H,T}(s)}{2}} dB_s \nonumber
    \end{eqnarray}

Then, one has:
 \begin{eqnarray}
        \text{Cov}\left(\delta_{\tau}X_t,\delta_{\tau}X_{t+h}  \right)=\E\left( \int_t^{t+\tau}e^{\frac{\omega_{H,T}(s)}{2}} dB_s\int_{t+h}^{t+\tau+h}e^{\frac{\omega_{H,T}(s)}{2}} dB_s\right) \nonumber
    \end{eqnarray}
Using Itô's isometry, we conclude that:
\begin{eqnarray}
        \text{Cov}\left(\delta_{\tau}X_t,\delta_{\tau}X_{t+h}\right)
= \E\left(\int_{t+h}^{t+\tau}e^{\omega_{H,T}(s)}\,ds\right)\mathds{1}_{\{h<\tau\}}
        \nonumber
    \end{eqnarray}
Besides,  we have pointwisely for any $s\geq 0$, $\omega_{H,T}(s) \sim \mathcal{N}\left(-\frac{\lambda^2}{4H(1-2H)},\frac{\lambda^2}{2H(1-2H)} \right)$. We obtain accordingly the claimed result.
    \\
\tab\tab\tab\tab\tab\tab\tab\tab\tab\tab\tab\tab\tab\tab\tab\tab\tab\tab\tab\tab\tab\tab\tab\tab\tab\tab\tab\tab\tab\tab\tab\tab\tab\tab $\blacksquare$
\subsection{Proof of Theorem \ref{theo:scaleinvariance} }
\label{app:scaleinvarianceproof}
Before diving into the proof, we suggest to present a definition of equivalence of two stochastic processes depending of a certain parameter say $a\in \R$ in the first order with respect $a$ as formulated in Notation 3.1 in \cite{bacry2008log}.
\begin{mydeff}
\label{def:firstorderequi}
    Two stochastic processes $\left(X^a_t\right)_t$ and $\left(Y^a_t\right)_t$ indexed by a parameter $a\in \R$ are equivalent to the first order in $a$ if for any arbitrary $(t_1,...,t_n)\in \R^n$ the following holds:
    \begin{eqnarray}
    \label{eq:firstordereq}
        \E\left(\prod_{i=1}^n X^a_{t_i}\right) = \E\left(\prod_{i=1}^n Y^a_{t_i}\right)+o\left(a\right)
    \end{eqnarray}
\end{mydeff}
In other words, the generalized moment of one process is the first order approximation when $a\rightarrow 0$ of the generalized moment of other. For the sake of simplicity, we denote $"\overset{\mathcal{L}}{\underset{\substack{ a \rightarrow 0}}{=}}"$ the equivalence to the first order in a. 

It is worth mentioning that if one considers two $\left(X^a_t\right)_t$ and $\left(Y^a_t\right)_t$  are equivalent to the first order in $a$ and two independent random variables $G$ and $G^{'}$ identically distributed independent from  $\left(X^a_t\right)_t$ and $\left(Y^a_t\right)_t$, the processes  $\left(G X^a_t\right)_t$ and $\left(G^{'} Y^a_t\right)_t$  remain equivalent to the first order in $a$.\\

An application of the previous definition is the following Proposition concerning the Log S-fBM process $X$ obtained similarly to Eq.(57) in \cite{bacry2008log} in the MRW context.
\begin{prop}
\label{prop:equivalencesmallintermittencylogsfbm}
For any $\tau,t\geq 0$, there exists a gaussian random variable $G$ independent of $\omega_{H,T}$ such that:
    \begin{eqnarray}
\delta_{\tau}X_t\overset{\mathcal{L}}{\underset{\substack{ \lambda^2 \rightarrow 0}}{=}}\sqrt{\tau}G e^{\frac{\lambda}{2\tau}\delta_{\tau}\Omega_{H,T}(t)-\frac{\lambda^2}{4}Var\left(\frac{\delta_{\tau}\Omega_{H,T}(t)}{\tau} \right)}
\end{eqnarray}
\end{prop}
\textit{\textbf{Proof}}\\
     We have thanks to Proposition 4 in \cite{wu2022rough} to the first order in $\lambda^2$:
    \begin{eqnarray}
        \ln\left(\frac{\delta_{\tau}M_{H,T}(t)}{\tau} \right) \overset{\mathcal{L}}{\underset{\substack{ \lambda^2 \rightarrow 0}}{=}} \frac{\lambda}{\tau} \delta_{\tau}\Omega_{H,T}(t)
        \nonumber
\end{eqnarray}
For arbitrary $\left(t_1,...,t_m\right)\in \R_{+}^n$, we have after some algebra:
\begin{eqnarray}
    \E\left(\prod_{i=1}^ne^{\frac{\lambda}{\tau}\delta_{\tau}\Omega_{H,T}(t_i)-\frac{\lambda^2}{2}Var\left(\frac{\delta_{\tau}\Omega_{H,T}(t_i)}{\tau} \right)}\right)=e^{\lambda^2\sum_{1\leq i< j \leq n}\text{Cov}\left(\frac{\delta_{\tau}\Omega_{H,T}(t_i)}{\tau},\frac{\delta_{\tau}\Omega_{H,T}(t_j)}{\tau} \right)} \nonumber\\
    =1+\lambda^2\sum_{1\leq i< j \leq n}\text{Cov}\left(\frac{\delta_{\tau}\Omega_{H,T}(t_i)}{\tau},\frac{\delta_{\tau}\Omega_{H,T}(t_j)}{\tau} \right)+o\left(\lambda^2\right)
    \nonumber
\end{eqnarray}
After some algebra using Fubini's theorem, we have as well:
\begin{eqnarray}
    \E\left(\prod_{i=1}^n\frac{\delta_{\tau}M_{H,T}(t_i)}{\tau}\right)=\frac{1}{\tau^n}\int_{t_1}^{t_1+\tau}...\int_{t_n}^{t_n+\tau}\E\left(e^{\sum_{j=1}^n\omega_{H,T}(u_j)} \right) du_1...du_n 
    \nonumber
\end{eqnarray}As a result:
\begin{eqnarray}
    \E\left(\prod_{i=1}^n\frac{\delta_{\tau}M_{H,T}(t_i)}{\tau}\right)=\frac{1}{\tau^n}\int_{t_1}^{t_1+\tau}...\int_{t_n}^{t_n+\tau}\E\left(e^{\sum_{j=1}^n\omega_{H,T}(u_j)-\mu_{H}} \right) e^{n\mu_{H}}du_1...du_n 
    \nonumber
\end{eqnarray}
After some algebra, we get:
\begin{eqnarray}
    \E\left(\prod_{i=1}^n\frac{\delta_{\tau}M_{H,T}(t_i)}{\tau}\right)=\frac{1}{\tau^n}\int_{t_1}^{t_1+\tau}...\int_{t_n}^{t_n+\tau}e^{\lambda^2\sum_{1\leq i< j\leq n}\text{Cov}\left(\frac{1}{\lambda}\left(\omega_{H,T}(u_i)-\mu_{H}\right),\frac{1}{\lambda}\left(\omega_{H,T}(u_j)-\mu_{H}\right)\right)}du_1...du_n 
    \nonumber
\end{eqnarray}
We have in small intermittency the following the first order Taylor approximation:
\begin{eqnarray}
    e^{\lambda^2\sum_{1\leq i< j\leq n}\text{Cov}\left(\frac{1}{\lambda}\left(\omega_{H,T}(u_i)-\mu_{H}(u_i)\right),\frac{1}{\lambda}\left(\omega_{H,T}(u_j)-\mu_{H}(u_j)\right)\right)}=1+\nonumber \\
    \lambda^2\sum_{1\leq i< j\leq n}\text{Cov}\left(\frac{1}{\lambda}\left(\omega_{H,T}(u_i)-\mu_{H}(u_i)\right),\frac{1}{\lambda}\left(\omega_{H,T}(u_j)-\mu_{H}(u_j)\right)\right)+ o\left(\lambda^2\right)\nonumber
\end{eqnarray}
Meaning that:
\begin{eqnarray}
    \E\left(\prod_{i=1}^n\frac{\delta_{\tau}M_{H,T}(t_i)}{\tau}\right)=1+
    \frac{\lambda^2}{\tau^2}\sum_{1\leq i< j\leq n}\text{Cov}\left(\delta_{\tau}\Omega_{H,T}(t_i),\delta_{\tau}\Omega_{H,T}(t_j)\right)+o\left(\lambda^2\right)\nonumber
    \nonumber
\end{eqnarray}
Consequently, we have that:
\begin{eqnarray}
\label{eq:equivmomentinteger}
     \E\left(\prod_{i=1}^n\frac{\delta_{\tau}M_{H,T}(t_i)}{\tau}\right)=\E\left(\prod_{i=1}^ne^{\frac{\lambda}{\tau}\delta_{\tau}\Omega_{H,T}(t_i)-\frac{\lambda^2}{2}Var\left(\frac{\delta_{\tau}\Omega_{H,T}(t_i)}{\tau} \right)}\right)+o\left(\lambda^2\right)
\end{eqnarray}
Meaning that for any $t\geq 0$, the generalized moments of $\frac{\delta_{\tau}M_{H,T}(t)}{\tau}$ are to the first order in $\lambda^2$ the ones of $e^{\frac{\lambda}{\tau}\delta_{\tau}\Omega_{H,T}(t)-\frac{\lambda^2}{2}Var\left(\frac{\delta_{\tau}\Omega_{H,T}(t)}{\tau} \right)}$.
As a result, the two previous processes are equivalent to the first order in $\lambda^2$ (in the sense of definition \ref{def:firstorderequi}).\\
On the other hand:
\begin{eqnarray}
    \forall t\geq 0,\tab X_t = B_{M_{H,T}(t)}
    \nonumber
\end{eqnarray}
Which leads to the equality in law of the two processes:
\begin{eqnarray}
    \left(\delta_{\tau}X_t\right)_t \overset{\mathcal{L}}=\left( B_{\delta_{\tau}M_{H,T}(t)}\right)_t
    \nonumber
\end{eqnarray}
Thus, we obtain thanks to the Independence of the Brownian increments as well as using Definition \ref{def:firstorderequi} the claimed result.
\\
\tab\tab\tab\tab\tab\tab\tab\tab\tab\tab\tab\tab\tab\tab\tab\tab\tab\tab\tab\tab\tab\tab\tab\tab\tab\tab\tab\tab\tab\tab\tab\tab\tab\tab $\blacksquare$

Now, we are ready for the proof of Theorem \ref{theo:scaleinvariance}. \\\\

    For any $t\geq 0$ and $p\in \N$, we have using the arguments of the proof of Proposition \ref{prop:equivalencesmallintermittencylogsfbm}:
\begin{eqnarray}
     m^{M}_{H,T,\tau}(p)=\tau^p e^{-\frac{\lambda^2p}{2}Var\left( \frac{\delta_{\tau}\Omega_{H,T}(t)}{\tau} \right)}\E\left(e^{p\lambda \frac{\delta_{\tau}\Omega_{H,T}(t)}{\tau}} \right)+o\left(\lambda^2\right)
    \nonumber
\end{eqnarray}
We have: 
\begin{eqnarray}
    Var\left( \delta_{\tau}\Omega_{H,T}(t)\right) = \frac{\tau^2}{2H (1-2H)}-\frac{2|\tau|^{2H+2}}{4HT^{2H}(1-4H^2)(H+1)}
    \nonumber
\end{eqnarray}
Thus:
\begin{eqnarray}
     m^{M}_{H,T,\tau}(p)=\tau^pe^{-\frac{\lambda^2}{2\tau^2}p\left(\frac{\tau^2}{2H (1-2H)}-\frac{2|\tau|^{2H+2}}{4HT^{2H}(1-4H^2)(H+1)}\right)}\E\left(e^{p\lambda \frac{\delta_{\tau}\Omega_{H,T}(t)}{\tau}} \right)+o\left(\lambda^2\right)
    \nonumber
\end{eqnarray}
On the other hand:
\begin{eqnarray}
    \E\left(e^{p\lambda \frac{\delta_{\tau}\Omega_{H,T}(t)}{\tau}} \right) = e^{\frac{\lambda^2p^2}{2}Var\left( \frac{\delta_{\tau}\Omega_{H,T}(t)}{\tau} \right)} \nonumber
\end{eqnarray}
Injecting this, we have:
\begin{eqnarray}
     m^{M}_{H,T,\tau}(p)\underset{\lambda^2\rightarrow 0}\sim \tau^pe^{\frac{\lambda^2}{2\tau^2}\left(p^2-p\right)\left(\frac{\tau^2}{2H (1-2H)}-\frac{2|\tau|^{2H+2}}{4HT^{2H}(1-4H^2)(H+1)}\right)}
    \nonumber
\end{eqnarray}
Consequently,
\begin{eqnarray}
    \ln\left( m^{M}_{H,T,\tau}(p)\right)\underset{\lambda^2\rightarrow 0}\sim p\ln\left(\tau\right)+\frac{\lambda^2}{2\tau^2}\left(p^2-p\right)\left(\frac{\tau^2}{2H (1-2H)}-\frac{2|\tau|^{2H+2}}{4HT^{2H}(1-4H^2)(H+1)}\right)
    \nonumber
\end{eqnarray}

    Which leads to the result.\\\\
On the other hand, we have using the same previous arguments related to the increments:
\begin{eqnarray}
    m^{X}_{H,T,\tau}(p)=\tau^{\frac{p}{2}}e^{\frac{\lambda^2}{4\tau^2}\left(p^2-p\right)\left(\frac{\tau^2}{2H (1-2H)}-\frac{2|\tau|^{2H+2}}{4HT^{2H}(1-4H^2)(H+1)}\right)}\E\left( |G|^p \right)+o\left(\lambda^2\right)
    \nonumber
\end{eqnarray}
Which leads to:
\begin{eqnarray}
    m^{X}_{H,T,\tau}(p)=\tau^{\frac{p}{2}}e^{\frac{\lambda^2}{\tau^2}\frac{p}{2}\left(\frac{p}{2}-1\right)\left(\frac{\tau^2}{2H (1-2H)}-\frac{2|\tau|^{2H+2}}{4HT^{2H}(1-4H^2)(H+1)}\right)}2^{\frac{p}{2}}\frac{\Gamma\left(\frac{p+1}{2}\right)}{\sqrt{\pi}}+o\left(\lambda^2\right)
    \nonumber
\end{eqnarray}
We conclude the proof arguing similarly to $m^{M}_{H,T,\tau}$.\\

\tab\tab\tab\tab\tab\tab\tab\tab\tab\tab\tab\tab\tab\tab\tab\tab\tab\tab\tab\tab\tab\tab\tab\tab\tab\tab\tab\tab\tab\tab\tab\tab\tab\tab $\blacksquare$

\subsection{Generalization of Theorem \ref{theo:scaleinvariance} to any positive real order $p$}

In the sequel, the positive upper bounding multiplicative constants may change from a line to another but we chose to denote them all $C$ for the sake of simplicity.\\\\

We denote:
\begin{eqnarray}
    \begin{cases}
    Z_{\tau}(t) := \ln\left(\frac{\delta_\tau M_{H,T}(t)}{\tau} \right)\\
    \Tilde{Z}_{\tau}(t) := \frac{\lambda}{\tau}\delta_\tau \Omega_{H,T}(t)
-\frac{\lambda^2}{2}
\mathrm{Var}\!\left(\frac{\delta_\tau \Omega_{H,T}(t)}{\tau}\right)\\
    \end{cases}.\nonumber
\end{eqnarray}
In the sequel, we denote $\mathcal{V}_{\tau}^{\Omega}:=\mathrm{Var}\!\left(\frac{\delta_\tau \Omega_{H,T}(t)}{\tau}\right)$.\\
We introduce the remainder process $R^{\lambda}_{\tau}$ such that:
\begin{eqnarray}
    Z_{\tau}(t) = \lambda \frac{\delta_\tau \Omega_{H,T}(t)}{\tau} + R^{\lambda}_{\tau}(t) - \frac{\lambda^2}{2}
\mathcal{V}_{\tau}^{\Omega} \nonumber
\end{eqnarray}
Then, for any \(p>0\), the generalized moments can be written as:
\begin{eqnarray}
    \begin{cases}
    \E\left[
\left(\frac{\delta_\tau M_{H,T}(t)}{\tau}\right)^p
\right]=\E\left[e^{pZ_{\tau}(t)}\right]\\
\E\left[\left(
e^{\frac{\lambda}{\tau}\delta_\tau \Omega_{H,T}(t)
-\frac{\lambda^2}{2}\mathcal{V}_{\tau}^{\Omega}}
\right)^p
\right] = \E\left[e^{p\Tilde{Z}_{\tau}(t)}\right]
    \end{cases}\nonumber
\end{eqnarray}
Let us compare the latter Laplace transforms.

One has straightforwardly that:

\begin{eqnarray}
\mathbb{E}\!\left[e^{pZ_{\tau}(t)}\right]
-
\mathbb{E}\!\left[e^{p\tilde Z_{\tau}(t)}\right]
&=
\mathbb{E}\!\left[
e^{p\tilde Z_{\tau}(t)}\big(e^{pR_{\tau}^{\lambda}(t)}-1\big)
\right].\nonumber
\end{eqnarray}

Using, Cauchy–Schwarz inequality one has:

\begin{eqnarray}
\left|
\mathbb{E}\!\left[e^{pZ_{\tau}(t)}\right]
-
\mathbb{E}\!\left[e^{p\tilde Z_{\tau}(t)}\right]
\right|
&\le
\left(\mathbb{E}\!\left[e^{2p\tilde Z_{\tau}(t)}\right]\right)^{1/2}
\left(\mathbb{E}\!\left[|e^{pR_{\tau}^{\lambda}(t)}-1|^2\right]\right)^{1/2}.\nonumber
\end{eqnarray}

Besides, using the mean value bound:
\begin{eqnarray}
\forall x\in \R, \tab |e^x - 1| \le |x| e^{|x|}, \nonumber
\end{eqnarray}
leads to:
\begin{eqnarray}
|e^{pR_{\tau}^{\lambda}(t)}-1|^2
\le
p^2 |R_{\tau}^{\lambda}(t)|^2 e^{2|pR_{\tau}^{\lambda}(t)|}.\nonumber
\end{eqnarray}

A second Cauchy–Schwarz inequality leads to:

\begin{eqnarray}
\mathbb{E}\!\left[|R_{\tau}^{\lambda}(t)|^2 e^{2|pR_{\tau}^{\lambda}(t)|}\right]
&\le
\left(\mathbb{E}[|R_{\tau}^{\lambda}(t)|^4]\right)^{1/2}
\left(\mathbb{E}[e^{4|pR_{\tau}^{\lambda}(t)|}]\right)^{1/2}.\nonumber
\end{eqnarray}

By definition, one has:

\begin{eqnarray}
    \left| R^{\lambda}_{\tau}(t) \right| \leq \left|Z_{\tau}(t)-\lambda \frac{\delta_\tau \Omega_{H,T}(t)}{\tau} \right|+\frac{\lambda^2}{2}
\mathcal{V}_{\tau}^{\Omega} \nonumber
\end{eqnarray}
This leads to:
\begin{eqnarray}
    \left| R^{\lambda}_{\tau}(t) \right| \leq \left|\ln\left(1+\left(\frac{\delta_\tau M_{H,T}(t)}{\tau}-1\right) \right)-\left(\frac{\delta_\tau M_{H,T}(t)}{\tau}-1\right) \right| +\left|\left(\frac{\delta_\tau M_{H,T}(t)}{\tau}-1\right)-\lambda \frac{\delta_\tau \Omega_{H,T}(t)}{\tau} \right| +\frac{\lambda^2}{2}
\mathcal{V}_{\tau}^{\Omega} \nonumber
\end{eqnarray}

Using Taylor-Lagrange inequality, there exist a positive constant $C$ such that:
\begin{eqnarray}
\left| R^{\lambda}_{\tau}(t) \right|
&\leq& C \left(
\frac{1}{2}\left|\frac{\delta_\tau M_{H,T}(t)}{\tau}-1\right|^2
+\frac{\lambda^2}{2} \int_t^{t+\tau}e^{\left|\omega_{H,T}(s)\right|}\left(\frac{\omega_{H,T}(s)}{\lambda}\right)^2ds
+\frac{\lambda^2}{2} \right)
\nonumber
\end{eqnarray}
Again by Taylor-Lagrange inequality, 
\begin{eqnarray}
\left| R^{\lambda}_{\tau}(t) \right|
&\leq& C \left(
\frac{\lambda^2}{2}\left|\delta_\tau \Omega_{H,T}(t)\right|^2
+\frac{\lambda^2}{2} \int_t^{t+\tau}e^{\left|\omega_{H,T}(s)\right|}\left(\frac{\omega_{H,T}(s)}{\lambda}\right)^2ds
+\frac{\lambda^2}{2} \right)
\nonumber
\end{eqnarray}

Leading to the following estimate:
\begin{eqnarray}
\mathbb{E}[|R_{\tau}^{\lambda}(t)|^4]
\le C \lambda^8.\nonumber
\end{eqnarray}
where $C>0$ an arbitrary multiplicative constant.

Concerning the exponential bound, one has for any $\alpha>0$, we write
\begin{eqnarray}
\mathbb{E}\big[e^{\alpha |R_{\tau}^{\lambda}(t)|}\big]
&=
\int_0^\infty \alpha e^{\alpha x}\,\mathbb{P}(|R_{\tau}^{\lambda}(t)|\ge x)\,dx.\nonumber
\end{eqnarray}

Using the shifted Gaussian tail estimate (see Theorem \ref{thm:deviationineqmrm}),
\begin{eqnarray}
\mathbb{P}(|R_{\tau}^{\lambda}(t)|\ge x)
\le
\exp\!\left(-c\frac{(x-\mu_\lambda)^2}{\lambda^2}\right), c\geq 0\nonumber
\end{eqnarray}
This means that:
\begin{eqnarray}
\mathbb{E}\big[e^{\alpha |R_{\tau}^{\lambda}(t)|}\big]
&\le&
\alpha \int_0^\infty
\exp\!\left(
\alpha x - \frac{c}{\lambda^2}(x-\mu_\lambda)^2
\right)\,dx.
\nonumber
\end{eqnarray}
After some algebra, 
\begin{eqnarray}
\mathbb{E}\big[e^{\alpha |R_{\tau}^{\lambda}(t)|}\big]
\le
\alpha e^{-\frac{c\mu_\lambda^2}{\lambda^2}}
\int_0^\infty
\exp\!\left(
-\frac{c}{\lambda^2}x^2 + A_\lambda x
\right)\,dx.
\nonumber
\end{eqnarray}
where:
\begin{eqnarray}
A_\lambda := \alpha + \frac{2c\mu_\lambda}{\lambda^2}.
\nonumber
\end{eqnarray}

By direct computations, 
\begin{eqnarray}
\mathbb{E}\big[e^{\alpha |R_{\tau}^{\lambda}(t)|}\big]
\le
\alpha
\exp\!\left(
-\frac{c\mu_\lambda^2}{\lambda^2}
+
\frac{\lambda^2 A_\lambda^2}{4c}
\right)
\int_0^\infty
\exp\!\left(
-\frac{c}{\lambda^2}
\left(
x - \frac{\lambda^2 A_\lambda}{2c}
\right)^2
\right)\,dx.
\nonumber
\end{eqnarray}
Hence:
\begin{eqnarray}
\mathbb{E}\big[e^{\alpha |R_{\tau}^{\lambda}(t)|}\big]
\le
C\lambda
\exp\!\left(
-\frac{c\mu_\lambda^2}{\lambda^2}
+
\frac{\lambda^2 A_\lambda^2}{4c}
\right).
\nonumber
\end{eqnarray}
After simplification:
\begin{eqnarray}
\mathbb{E}\big[e^{\alpha |R_{\tau}^{\lambda}(t)|}\big]
\le
C\lambda
\exp\!\left(
\alpha \mu_\lambda
+
\frac{\alpha^2}{4c}\lambda^2
\right).
\nonumber
\end{eqnarray}

Recall
\begin{eqnarray}
\mu_\lambda = \frac{\lambda^2}{2}\mathcal{V}_\tau^\Omega
= O(\lambda^2).\nonumber
\end{eqnarray}

Therefore, for fixed $\alpha$, this term is absorbed into the quadratic term, so:
\begin{eqnarray}
\mathbb{E}\big[e^{\alpha |R_{\tau}^{\lambda}(t)|}\big]
&\le
C \lambda \exp\!\left(C \alpha^2 \lambda^2\right).\nonumber
\end{eqnarray}

Taking \(\alpha = 4|p|\),
\begin{eqnarray}
\mathbb{E}[e^{4|pR_{\tau}^{\lambda}(t)|}]
\le
C\lambda \exp(C p^2 \lambda^2).\nonumber
\end{eqnarray}

Combining the bounds lead to:
\begin{eqnarray}
\mathbb{E}\!\left[|R|^2 e^{2|pR|}\right]
\le
C \lambda^9 \exp(C p^2 \lambda^2).\nonumber
\end{eqnarray}
which means that:
\begin{eqnarray}
\left|
\mathbb{E}\!\left[e^{pZ_{\tau}(t)}\right]
-
\mathbb{E}\!\left[e^{p\tilde Z_{\tau}(t)}\right]
\right|
&\le
\left(\mathbb{E}[e^{2p\tilde Z_{\tau}(t)}]\right)^{1/2}
|p|
\left(C \lambda^9 e^{C p^2 \lambda^2}\right)^{1/2}.\nonumber
\end{eqnarray}
This means that:
\begin{eqnarray}
\mathbb{E}\!\left[e^{pZ_{\tau}(t)}\right]
=
\mathbb{E}\!\left[e^{p\tilde Z_{\tau}(t)}\right]+o\left(\lambda^2\right).\nonumber
\end{eqnarray}
conducting the same derivations as in the proof of Theorem \ref{theo:scaleinvariance} leads to the generalisation of Theorem \ref{theo:scaleinvariance} for any positive real order $p$.\\
\tab\tab\tab\tab\tab\tab\tab\tab\tab\tab\tab\tab\tab\tab\tab\tab\tab\tab\tab\tab\tab\tab\tab\tab\tab\tab\tab\tab\tab\tab\tab\tab\tab\tab $\blacksquare$

\end{appendices}

\end{document}